\documentclass[journal]{IEEEtran} 

\usepackage[noadjust]{cite}
\usepackage{url}
\usepackage[hidelinks]{hyperref}
\usepackage{amsthm}
\usepackage{amssymb}
\usepackage{siunitx}
\usepackage{mathtools}
\usepackage{float}
\usepackage{xspace}
\usepackage[font=small,labelfont=bf]{caption}
\usepackage[labelformat=simple]{subcaption}

\usepackage{booktabs}

\usepackage{algorithmic}
\usepackage{algorithm}

\usepackage{xcolor}
\colorlet{mycolor}{blue}

\usepackage{amsmath}
\newtheorem{theorem}{Theorem}

\newtheorem{lemma}{Lemma}
\newtheorem{definition}{Definition}
\newtheorem{remark}{Remark}
\newtheorem{assumption}{Assumption}
\newtheorem{proposition}{Proposition}
\newtheorem{example}{Example}

\DeclareMathOperator{\Log}{Log}
\DeclareMathOperator*{\argmin}{arg\,min}

\newcommand{\grad}{\nabla}
\newcommand{\intersect}{\cap}

\newcommand{\R}{\mathbb{R}}

\begin{document}

\title{Collision Avoidance for Convex Primitives via Differentiable Optimization Based High-Order Control Barrier Functions}

\author{Shiqing~Wei,~\IEEEmembership{Graduate~Student~Member,~IEEE,} Rooholla~Khorrambakht,~\IEEEmembership{Student~Member,~IEEE,}
Prashanth~Krishnamurthy,~\IEEEmembership{Member,~IEEE,}
Vinicius~Mariano~Gonçalves, and Farshad~Khorrami,~\IEEEmembership{Fellow,~IEEE}
\thanks{This work was supported in part by ARO grants W911NF-21-1-0155 and W911NF-22-1-0028, NSF CMMI grant 2208189, and by the New York University Abu Dhabi (NYUAD) Center for Artificial Intelligence and Robotics (CAIR), funded by Tamkeen under the NYUAD Research Institute Award CG010.} 
\thanks{Shiqing Wei, Rooholla Khorrambakht, Prashanth Krishnamurthy, and Farshad Khorrami are with Control/Robotics Research Laboratory, Electrical and Computer Engineering Department, Tandon School of Engineering, New York University, Brooklyn, NY 11201, USA. E-mail: \{shiqing.wei, rk4342, prashanth.krishnamurthy, khorrami\}@nyu.edu.}
\thanks{Vinicius Mariano Gonçalves is with the Escola de Engenharia, Universidade Federal de Minas Gerais (UFMG), Belo Horizonte MG, 31270-901, Brazil. E-mail: mariano@cpdee.ufmg.br.}
\thanks{Digital Object Identifier (DOI): xxxxx.}
}


\markboth{IEEE Journal}
{Wei \MakeLowercase{\textit{et al.}}: Collision Avoidance for Convex Primitives via Differentiable Optimization Based HOCBF}

\maketitle

\begin{abstract}
Ensuring the safety of dynamical systems is crucial, where collision avoidance is a primary concern. Recently, control barrier functions (CBFs) have emerged as an effective method to integrate safety constraints into control synthesis through optimization techniques. However, challenges persist when dealing with convex primitives and tasks requiring torque control, as well as the occurrence of unintended equilibria. This work addresses these challenges by introducing a high-order CBF (HOCBF) framework for collision avoidance among convex primitives. We transform nonconvex safety constraints into linear constraints by differentiable optimization and prove the high-order continuous differentiability. Then, we employ HOCBFs to accommodate torque control, enabling tasks involving forces or high dynamics. Additionally, we analyze the issue of spurious equilibria in high-order cases and propose a circulation mechanism to prevent the undesired equilibria on the boundary of the safe set. Finally, we validate our framework with three experiments on the Franka Research 3 robotic manipulator, demonstrating successful collision avoidance and the efficacy of the circulation mechanism.
\end{abstract}


\IEEEpeerreviewmaketitle

\section{Introduction}
\subsection{Motivation}
The safety of dynamical systems has become increasingly important due to the growing demands in fields such as robotics, autonomous driving, and other safety-critical applications. While safety concerns differ across various applications, a common challenge is preventing collisions between the robot and the physical objects in the 3D world during task execution. Requiring a robot to stay out of a given region typically leads to a nonconvex constraint, which is known to be difficult to handle and challenging to use directly in convex optimization-based motion planning methods \cite{zhang2020optimization}. Many formulations resort to general nonlinear programming solvers or fall back to sampling-based planners. However, nonlinear solvers face a trade-off between the length of the preview horizon and computation time, and classical sampling-based planners \cite{kavraki1996probabilistic, lavalle2006planning} do not account for the dynamics of moving obstacles.

Control barrier functions (CBFs) \cite{ames2016control, xiao2021high, wei2024confidence} have attracted significant attention for enabling safety constraints to be imposed as linear inequalities in the control input for control-affine systems. These constraints can be further incorporated into an optimization problem, such as a quadratic program (QP). As a result, given a nominal control, this QP formulation (also known as CBF-QP) efficiently computes a minimally invasive control from the given nominal control, ensuring that safety constraints are met. Under the assumption that the robot can sense its surroundings and formulate safety constraints, the system can quickly respond to dynamic or even unexpected scenes, a crucial capability for practical applications.

In this work, we model obstacles and robots (or their parts) as convex primitives and address collision avoidance for both velocity- and torque-controlled systems. The motivation for this work is twofold. First, existing CBF formulations for torque-controlled systems do not account for collisions between \textit{general convex primitives}. Most formulations (potentially high-order) rely on specific geometries such as point mass \cite{ohnishi2019barrier}, planes \cite{ohnishi2019barrier, murtaza2022safety}, spheres \cite{emam2021data, shi2023safety, koutras2023enforcing, wang2017safety, cavorsi2023multi}, high-order ellipsoids \cite{murtaza2022safety}, and capsules \cite{shi2023safety, koutras2023enforcing}, which can lead to over-conservatism and limited applicability. Some exceptions consider convex polygons and polytopes \cite{thirugnanam2022duality, wei2024diffocclusion}, or general convex forms \cite{dai2023safe, thirugnanam2023nonsmooth}, but are restricted to velocity control. To overcome these limitations, we propose a collision avoidance framework for a broad class of convex primitives by representing them with \textit{scaling functions} and differentiating through the \textit{minimal scaling factor} at which two primitives intersect (where at least one of the two scaling functions has a positive definite Hessian). This framework utilizes differentiable optimization to transform nonconvex safety constraints into linear constraints in a CBF-QP, enabling efficient solutions. We prove the high-order differentiability of the minimal scaling factor and formulate HOCBFs to support safety-critical torque control tasks.

Second, like first-order CBFs, HOCBFs are susceptible to classical spurious equilibria. Prior research has shown that first-order CBFs can introduce unintended equilibria \cite{reis2020control, mestres2022optimization, tan2024undesired, gonccalves2024control}. We show that similar issues arise in high-order cases and propose a circulation mechanism, inspired by \cite{gonccalves2024control}, to prevent undesired equilibria on the boundary of the safe set. Although perturbations (such as noise) may help dynamical systems escape from these equilibria to some extent, we have observed from real-world experiments on a robotic manipulator that a designated mechanism is necessary. Our mechanism adds an additional linear constraint to the CBF-QP, incurring only a negligible increase in the computational complexity. Unlike \cite{gonccalves2024control}, which addresses only single integrator systems, we generalize to high-order systems satisfying a regularity condition (e.g., fully actuated torque-controlled Lagrangian systems). Our circulation mechanism also applies to first-order cases, extending the results of \cite{gonccalves2024control}.

\subsection{Contributions}
The contributions of this work are as follows:
\begin{itemize}
    \item We develop a novel framework for collision avoidance between general convex primitives using HOCBFs. We represent convex primitives as scaling functions and prove that the minimal scaling factor, where two convex primitives intersect, is $k$-times continuously differentiable if the scaling functions are $k+1$-times continuously differentiable and at least one has a positive definite Hessian\footnote{Code available at \url{https://github.com/shiqingw/DiffOpt-HOCBF-Pub}.}.
    \item We introduce a systematic approach to construct smooth scaling functions for convex primitives such as planes, polygons/polytopes, and ellipses/ellipsoids.
    \item We propose an HOCBF formulation based on the smoothness of the minimal scaling factor, identify spurious equilibria on the boundary of the safe set, and introduce a circulation mechanism to avoid them for fully actuated torque-controlled robotic systems.
    \item We demonstrate the effectiveness of our approach through three experiments on the Franka Research 3 robotic manipulator, showcasing successful collision avoidance and the effectiveness of the circulation mechanism.
\end{itemize}

This paper is organized as follows. Section~\ref{sec:related_work} presents related work. Section~\ref{sec:notations} summarizes the notations used in this paper. Section~\ref{sec:preliminaries} briefly recalls the HOCBF definition and related theoretical results. Section~\ref{sec:diff_opt} offers a detailed explanation of scaling functions and the conditions under which the minimal scaling factor is $k$-times continuously differentiable. Section~\ref{sec:scaling_funcs} outlines a systematic method for constructing smooth scaling functions for different convex primitives. Section~\ref{sec:hocbf} describes the HOCBF formulation and the circulation mechanism. Section~\ref{sec:experiments} presents three real-world examples that validate our approach. Finally, Section~\ref{sec:conclusion} provides concluding remarks.

\section{Related Work}\label{sec:related_work}
\subsection{Safe Planning and Control}
Classical sampling-based motion planning methods, such as the probabilistic roadmap (PRM) \cite{kavraki1996probabilistic} and rapidly-exploring random trees (RRT) \cite{lavalle2006planning}, generate a safe path by randomly sampling the configuration space and connecting collision-free samples into a graph or tree. The works \cite{ahmad2022adaptive, manjunath2021safe} combine sampling-based planners with CBF techniques. Convex optimization-based approaches, such as those in \cite{liu2018convex, marcucci2023motion}, focus on identifying convex feasible sets within the safe configuration space or approximating the safe space as a union of convex regions. Other trajectory optimization methods, such as model predictive control (MPC) \cite{sathya2018embedded, zhang2020optimization}, cast collision avoidance as nonconvex constraints, which can be addressed using general nonlinear programming solvers. These optimization-based methods typically require discretized system dynamics. Alternatively, artificial potential fields (APF) \cite{kim1992real} and barrier functions \cite{tee2009barrier, ames2016control}, along with their variants, offer safe control strategies for continuous-time systems. See \cite{brunke2022safe} for a comprehensive review of safe planning and control methods in robot control and learning. 

\begin{figure}[t]
    \centering
    \begin{subfigure}[b]{0.24\textwidth}
        \centering
        \includegraphics[width=1.0\textwidth]{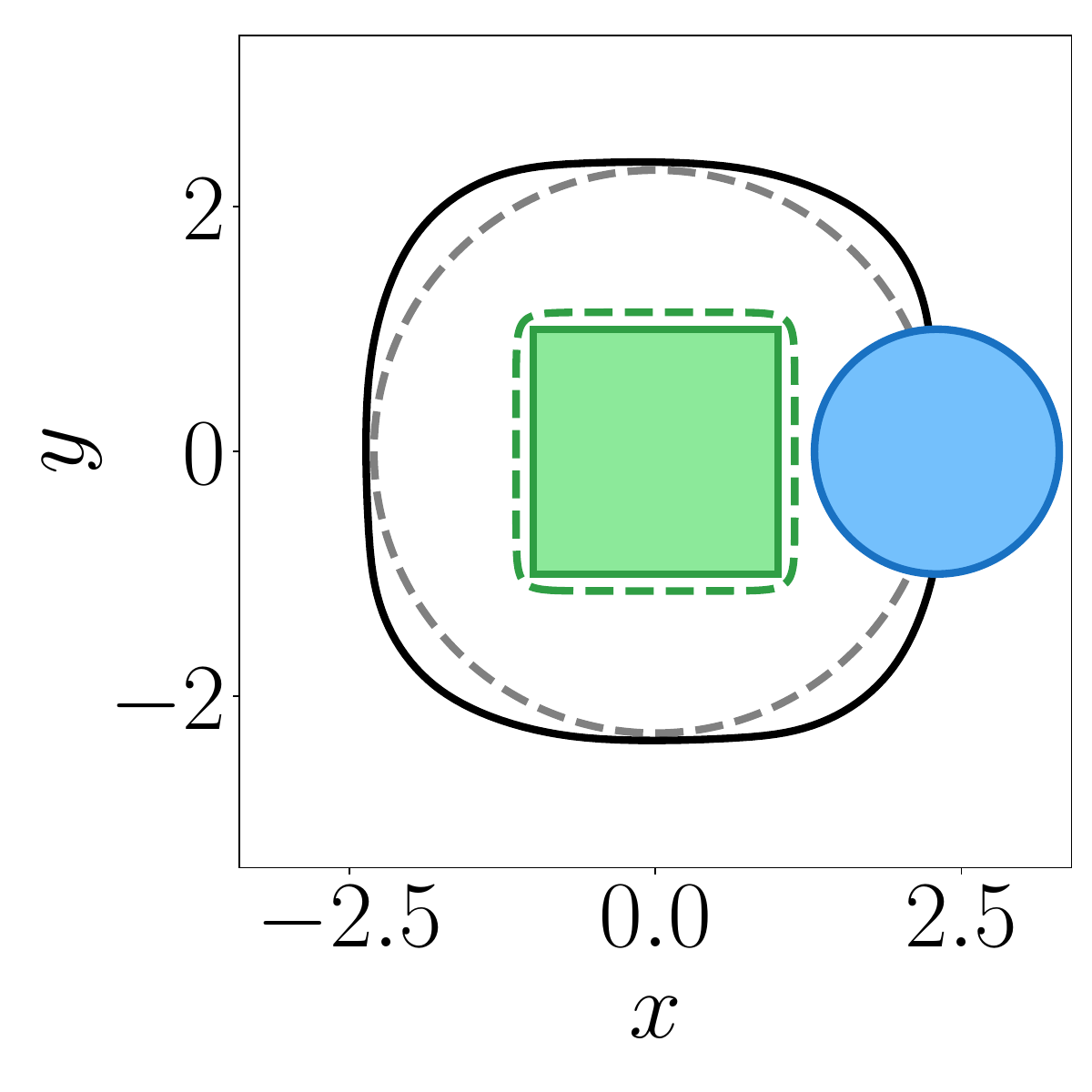}
    \end{subfigure}
    \begin{subfigure}[b]{0.24\textwidth}
        \centering
        \includegraphics[width=1.0\textwidth]{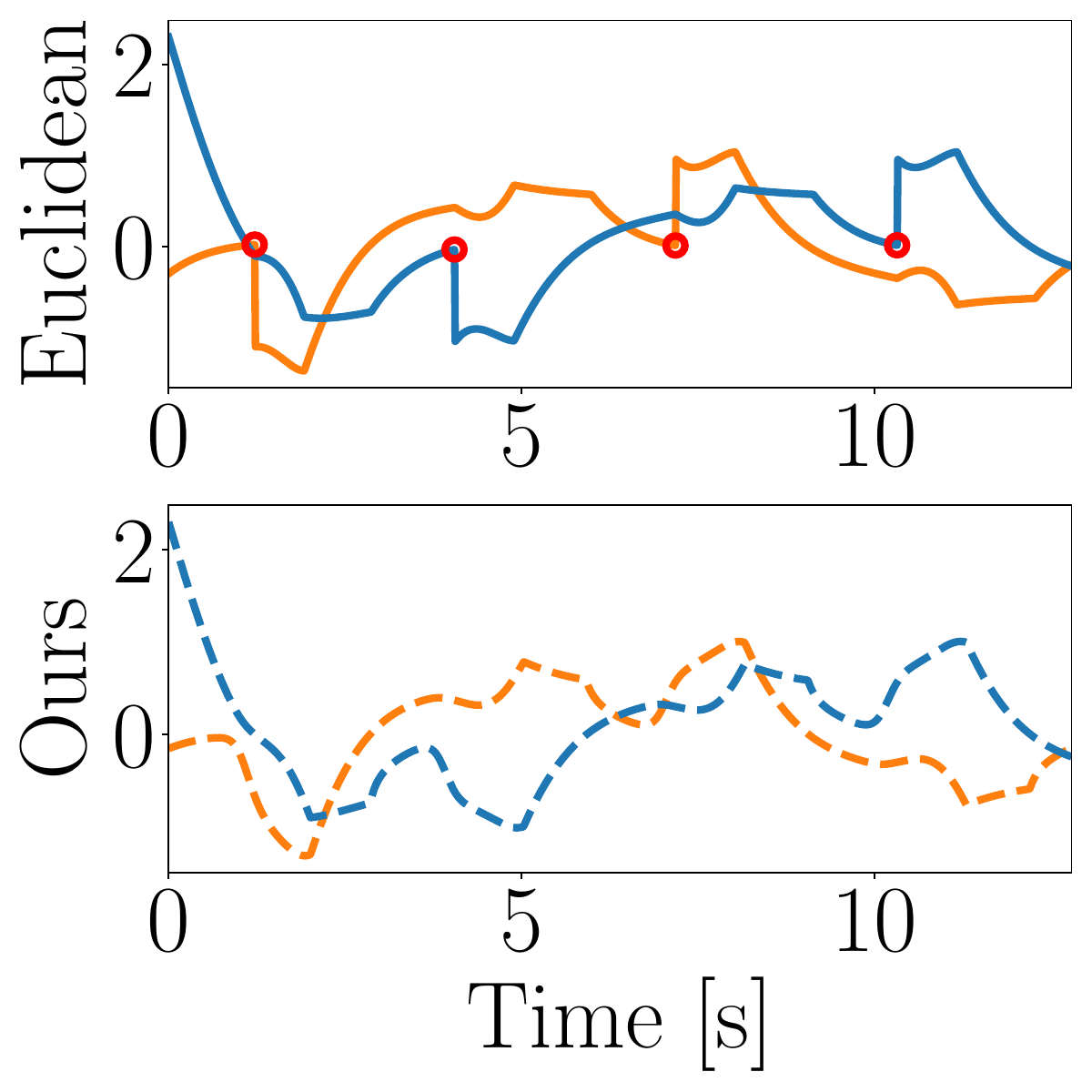}
    \end{subfigure}
    \caption{Nonsmoothness of Euclidean distance and resulting control discontinuities. Left: A circular robot of radius 1, governed by double integrator dynamics, tracks a reference trajectory (gray) while avoiding a square (green) of edge length 2 centered at $(0,0)$. The green dashed curve shows the boundary of the safe set induced by our proposed smooth scaling function, and the black solid line shows the collision-free trajectory. Right: Controls based on Euclidean distance with HOCBF exhibit discontinuities (top), whereas controls using our metric are continuous (bottom).}
    \label{fig:nonsmoothness_euclidean}
\end{figure}

\subsection{Collision Avoidance for Convex Regions with CBFs}
CBFs have been widely applied in tasks such as single-robot and multi-robot collision avoidance \cite{ohnishi2019barrier, murtaza2022safety, thirugnanam2022duality, xiao2023barriernet, dai2023safe, shi2023safety, koutras2023enforcing, wei2024diffocclusion, wang2017safety, emam2021data, cavorsi2023multi, thirugnanam2023nonsmooth, long2024safe}. In existing studies, obstacles and robots are typically modeled using simple geometries such as point mass \cite{ohnishi2019barrier}, planes \cite{ohnishi2019barrier, murtaza2022safety}, spheres \cite{emam2021data, shi2023safety, koutras2023enforcing, wang2017safety, cavorsi2023multi}, high-order ellipsoids \cite{murtaza2022safety}, and capsules \cite{shi2023safety, koutras2023enforcing}. Some works have extended the analysis to collisions between convex polygons and polytopes \cite{thirugnanam2022duality, wei2024diffocclusion}. The work \cite{long2024safe} considers the collision between polygons and ellipses in 2D and proposes a first-order CBF. To handle more general convex shapes, researchers have employed differentiable optimization \cite{dai2023safe, wei2024diffocclusion, thirugnanam2023nonsmooth, tracy2023differentiable}, a class of optimization problems whose solutions are differentiable w.r.t. the configuration parameters. Based on the Karush–Kuhn–Tucker (KKT) conditions, these works differentiate through the minimal distance or scaling factor between two convex shapes. Continuous differentiability was established in \cite{dai2023safe, wei2024diffocclusion}, but only for first-order cases, limiting applicability to velocity control. In this work, we extend the analysis by proving the continuous differentiability of a minimal scaling factor between general convex shapes in higher-order cases, thereby enabling safe torque control. Although prior studies \cite{murtaza2022safety, dai2023safe, shi2023safety, koutras2023enforcing} applied CBFs to manipulator collision avoidance, our formulation applies broadly to tasks involving mobile robots or manipulators, and our circulation mechanism applies to any system satisfying a regularity condition (e.g., fully actuated Lagrangian systems). We remark that the classical exponential CBFs (ECBFs) \cite{nguyen2016exponential} and high-order CBFs (HOCBFs) \cite{xiao2021high} are general frameworks and do not formulate a smooth metric for collision avoidance. Moreover, we have introduced the circulation mechanism that reduces the possibility of getting stuck at spurious equilibria, which is not discussed in ECBFs or HOCBFs. We also note that the classical Euclidean distance can be nondifferentiable when multiple point pairs yield the same minimal distance \cite{gonccalves2024smooth}, and is not twice continuously differentiable, despite having a unique solution (see Fig.~\ref{fig:nonsmoothness_euclidean}).

\subsection{Avoiding Spurious Equilibria in CBF Formulation}
One limitation of CBFs is the possibility of causing the robot to enter deadlocks \cite{wang2017safety, cavorsi2023multi} or spurious equilibria \cite{reis2020control, mestres2022optimization, tan2024undesired, gonccalves2024control}, both referring to the situation where the robot is unable to continue moving despite having incomplete tasks. Deadlocks are more commonly seen in multi-robot systems where the safety constraints of different robots conflict with each other \cite{wang2017safety} or the overall team formation \cite{cavorsi2023multi}. Spurious equilibria, on the other hand, result from a more profound interplay between the CBF constraint and nominal control within the CBF-QP framework. These equilibria can sometimes be stable, causing the closed-loop system to converge toward them \cite{reis2020control}. Similar issues also exist in APF-based methods \cite{chang2004gyroscopic, gao2023non}, along with a discussion on the potential topological obstruction \cite{koditschek1987exact}. The studies \cite{reis2020control, mestres2022optimization, tan2024undesired} focus on the interaction between CBFs and control Lyapunov functions (CLFs). In \cite{reis2020control}, the authors rotate the level sets of the reference Lyapunov function using an additional virtual state to avoid the equilibria on the obstacle boundary. The work in \cite{tan2024undesired} leverages prior knowledge of a stable nominal controller to eliminate equilibria in the interior of the safe set. Unlike \cite{reis2020control, mestres2022optimization, tan2024undesired}, \cite{gonccalves2024control} considers general nominal control and introduces a circulation inequality as a constraint, forcing the system to explicitly navigate around obstacles under certain conditions. However, all these studies \cite{reis2020control, mestres2022optimization, tan2024undesired, gonccalves2024control} focus on first-order CBFs, which are not applicable when using torque control. In this work, we build upon the concept of ``circulation'' from previous studies \cite{chang2004gyroscopic, gao2023non, gonccalves2024control} and design a circulation mechanism that eliminates undesired equilibria on the boundary of the safe set for high-order cases. This is achieved by adding an extra linear constraint to the CBF-QP, which enhances the applicability of CBF approaches with minimal increase in complexity. 

\section{Notations} \label{sec:notations}
Let $\R$ be the set of real numbers, $\R_+$ positive real numbers, $\mathbb{Z}_+$ positive integers, and $\mathbb{S}_{++}^n$ symmetric positive definite $n$-by-$n$ matrices. $\lVert \cdot \rVert_p$ is the $\ell_p$ norm of a vector. $(\cdot)^\dagger$ is the pseudoinverse of a matrix. $I$ is the identity matrix with proper dimensions. A class $\mathcal{C}^k$ function is a function that has a continuous $k$-th order derivative. A class $\mathcal{K}$ function $\Gamma: [0, a) \rightarrow [0,\infty)$ is a continuous strictly increasing function with $\Gamma(0) = 0 $. For a $\mathcal{C}^1$ function $h: \R^n \rightarrow \R$, its gradient $\grad h$ is a column vector while the partial derivative $\frac{\partial h}{\partial x}$ is a row vector. For a $\mathcal{C}^1$ function $f: \R^n \rightarrow \R^m$, the partial derivative $\frac{\partial f}{\partial x}$ is an $m$-by-$n$ matrix. $L_f h (x)$ is the Lie derivative of a scalar function $h$ w.r.t. a vector field $f$ with $m=n$. $\Log_{\text{SO}(3)}$ is the SO(3) logarithm from the Lie group SO(3) to the Lie algebra $\mathfrak{so}(3)$. $\partial A$ denotes the boundary of a set $A$.

\section{Preliminaries}\label{sec:preliminaries}
In this section, we briefly recall the definition of HOCBFs and their theoretical properties. Consider the following control-affine system 
\begin{equation}\label{eq:affine_sys}
    \dot{x} = f(x) + g(x)u
\end{equation}
where $f: \R^{n_x} \rightarrow \R^{n_x}$ and $g: \R^{n_x} \rightarrow \R^{n_x \times n_u}$ are Lipschitz functions, $x \in \mathcal{X} \subset \R^{n_x}$ the states and $u \in \mathcal{U} \subset \R^{n_u}$ the control (where $\mathcal{X}$ is a compact subset of $\R^{n_x}$ and $\mathcal{U}$ is a compact set to model the control limits).  Let $h: \R^{n_x} \rightarrow \R$ be a function of relative degree $m$ for the system \eqref{eq:affine_sys}. Assume that 
$h$ is $\mathcal{C}^m$, and $f$ and $g$ are $\mathcal{C}^{m-1}$. Define
\begin{subequations}\label{eq:intermediate_barrier_funcs}
\begin{align}
&\psi_0 (x) = h(x),\\
&\psi_i (x) = \dot{\psi}_{i-1} (x) + \Gamma_i (\psi_{i-1} (x)), \ i \in \{1, ..., m-1\}, \\
&\psi_m (x,u) = \dot{\psi}_{m-1} (x) + \Gamma_m (\psi_{m-1} (x)),
\end{align}
\end{subequations}
where $\psi_i: \R^{n_x} \rightarrow \R$ for $i=1,..., m-1$, $\psi_m: \R^{n_x} \times \R^{n_u} \rightarrow \R$, and $\Gamma_i \in \mathcal{C}^{m-i}$ are class $\mathcal{K}$ functions. Define also a sequence of sets $C_i \subset \R^{n_x}$
\begin{equation}\label{eq:high_order_safe_sets}
    C_i = \{ x\in \R^{n_x} \mid \psi_{i-1} (x) \geq 0 \}, \quad i \in \{1, ..., m\}.
\end{equation}

\begin{definition}[HOCBFs \cite{xiao2021high}]
Let $\psi_i$ with $i \in \{0, 1, ..., m\}$ be defined by \eqref{eq:intermediate_barrier_funcs} and $C_i$ with $i \in \{1, ..., m\}$ by \eqref{eq:high_order_safe_sets}. A $\mathcal{C}^m$ function $h: \R^{n_x} \rightarrow \R$ is an HOCBF of relative degree $m$ for system \eqref{eq:affine_sys} if there exist functions $\Gamma_i$ that are of both class $\mathcal{C}^{m-i}$ and class $\mathcal{K}$ with $i \in \{1,..., m\}$ such that 
\begin{equation}
\begin{aligned}
    \sup_{u \in \mathcal{U}} \biggl[ L_f^m h (x) &+ L_g L_f^{m-1} h (x) u  \\
    &+ \sum_{i=0}^{m-1} L_f^i (\Gamma_{m-i} \circ \psi_{m-i-1}) (x) \biggr] \geq 0
\end{aligned}
\end{equation}
for all $x \in C_1 \intersect ... \intersect C_m$.
\end{definition}

Additionally, given the HOCBF $h$, define
\begin{equation}\label{eq:control_cbf}
\begin{aligned}
K_{\text{hocbf}}(x) = \biggl\{ u\in \mathcal{U} \mid L_f^m h (x) + L_g L_f^{m-1} h (x) u  \\
    + \sum_{i=0}^{m-1} L_f^i (\Gamma_{m-i} \circ \psi_{m-i-1}) (x)\geq 0 \biggr\}.
\end{aligned}
\end{equation}
Then, any Lipschitz continuous controller $u(t) \in  K_{\text{hocbf}}(x(t))$ renders $C_1 \intersect ... \intersect C_m$ forwardinvariant for system \eqref{eq:affine_sys} if the initial condition $x_0 \in C_1 \intersect ... \intersect C_m$ \cite[Theorem 4]{xiao2021high}.

\section{Differentiable Optimization}\label{sec:diff_opt}
In this section, we consider a class of parameterized optimization problems and show that their solutions are $k$th-order continuously differentiable w.r.t. the underlying parameters.

\begin{definition}[Scaling functions \cite{wei2024diffocclusion}]\label{def:scaling_function}
    A class $\mathcal{C}^2$ function $\mathcal{F}_A: \R^{n_p} \times \R^{n_\theta} \rightarrow \R$ is a scaling function (with parameters $\theta \in \R^{n_\theta}$) for a closed set $A \subset \R^{n_p}$ with non-empty interior if $\mathcal{F}_A$ is convex in the first argument for each $\theta$, $A = \left\{ p \in \R^{n_p} \mid \mathcal{F}_A(p,\theta) \leq 1 \right\}$, and there exists $p \in A$ such that $\mathcal{F}_A (p, \theta) < 1$ for each $\theta$.
\end{definition}

In Definition~\ref{def:scaling_function}, the scaling function $\mathcal{F}_A$ is an implicit representation of the set $A$ in the $n_p$-dimensional space. As a result of $A = \left\{ p \in \R^{n_p} \mid \mathcal{F}_A(p,\theta) \leq 1 \right\}$, we know that $\mathcal{F}_A(p, \theta) = 1$ for any $p \in \partial A$. In addition, $\mathcal{F}_A(p, \theta) < \infty$ for any $p \in \R^{n_p}$ as $\mathcal{F}_A$ is $\mathcal{C}^2$. For example, given a 2D (or 3D) ball of radius $r$, we can define its scaling function as $\mathcal{F}_A(p,\theta) = (p-\theta)^\top P (p-\theta)$ where $\theta$ is the center of the ball and $P = I/r^2$. The set $A$ depends on $\theta$, but we omit its dependency for brevity, since we mainly focus on the position and orientation of the robots and obstacles. It is important to note that the scaling functions do not always exist for any given set. However, for a given scaling function $\mathcal{F}_A$, the convexity of $\mathcal{F}_A$ (w.r.t. $p$) implies the convexity of the associated set $A$, as $A$ is the 1-sublevel set of $\mathcal{F}_A$.

Let $\mathcal{F}_A$ and $\mathcal{F}_B$ be the scaling functions of the sets $A$ and $B$ and assume that $A \intersect B = \varnothing$. Then, the minimal scaling factor such that the scaled version of $A$ intersects with $B$ can be found via the following optimization problem:

\begin{equation} \label{eq:opt}
    \begin{aligned}
    \alpha^\star(\theta) = \min_{p \in \R^{n_p}} \quad & \mathcal{F}_A(p, \theta) \\
    \textrm{s.t.} \quad & \mathcal{F}_B(p, \theta) \leq 1.  \\
    \end{aligned}
\end{equation}

\begin{figure}[t]
    \centering
    \includegraphics[width=0.47\textwidth]{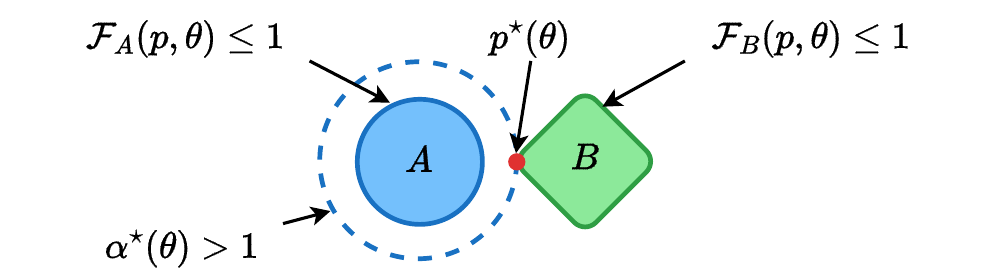}
    \caption{2D visualization of the optimization problem \eqref{eq:opt}. }
    \label{fig:diff_opt_2d}
\end{figure}

Problem \eqref{eq:opt} is strictly feasible because there exists $p \in B$ such that $\mathcal{F}_B (p, \theta) < 1$ for each $\theta$ by Definition~\ref{def:scaling_function}. The scaling functions $\mathcal{F}_A$ and $\mathcal{F}_B$ may have different parameters, but we concatenate their parameters and use a single parameter vector $\theta$ for notational simplicity. When \eqref{eq:opt} has a unique solution for $p$, \eqref{eq:opt} can be seen as a mapping from the parameters $\theta$ to the optimal solution $p^\star (\theta)$ and the optimal value $\alpha^\star (\theta)$. We visualize the 2D case of \eqref{eq:opt} in Fig.~\ref{fig:diff_opt_2d} as an illustrative example. In physical terms, $\alpha^\star (\theta)$ (represented by the blue dotted line in Fig.~\ref{fig:diff_opt_2d}) means how much we need to magnify $A$ until it touches $B$. Next, we would like to study the sensitivity (see, e.g., \cite[Section 5.6]{boyd2004convex}) of $\alpha^\star (\theta)$ w.r.t. $\theta$. The following lemma results from the first-order condition of a convex function.

\begin{lemma}\label{lemma:non_zero_grad}
    Given a scaling function $\mathcal{F}_A$, for any $p$ such that $\mathcal{F}_A(p, \theta) > 1$, $\frac{\partial \mathcal{F}_A}{\partial p}(p, \theta) \neq 0$.
\end{lemma}
\begin{proof}
    As $\mathcal{F}_A$ is convex and belongs to class $\mathcal{C}^2$, we have
    \begin{equation*}
        \mathcal{F}_A(p',\theta) \geq \mathcal{F}_A(p,\theta) + \frac{\partial \mathcal{F}_A}{\partial p}(p,\theta) (p'-p), \quad \forall p' \in \R^{n_p}.
    \end{equation*}
    If  $\frac{\partial \mathcal{F}_A}{\partial p}(p,\theta) = 0$, then $\mathcal{F}_A(p') \geq \mathcal{F}_A(p)$ for any $p' \in \R^{n_p}$, i.e., $p$ is a global minimizer of $\mathcal{F}_A$. However, this is impossible because $A = \{p \in \R^{n_p} \mid \mathcal{F}_A(p,\theta) \leq 1 \}$ is non-empty. Therefore, $\frac{\partial \mathcal{F}_A}{\partial p}(p,\theta) \neq 0$ must hold.
\end{proof}

The Lagrangian function of \eqref{eq:opt} is
\begin{equation}
    L(p, \theta, \lambda) = \mathcal{F}_A(p,\theta) + \lambda (\mathcal{F}_B(p,\theta) - 1).
\end{equation}
Since \eqref{eq:opt} is a strictly feasible convex problem (i.e., Slater's condition holds), the KKT conditions are necessary. Therefore, any pair of optimal primal variable $p^\star$ and dual variable $\lambda^\star$ satisfy the following KKT conditions (we omit the dependence of $\alpha^\star$, $p^\star$, and $\lambda^\star$ on $\theta$ for notational convenience):
\begin{subequations}
\begin{align}
    \frac{\partial \mathcal{F}_A}{\partial p} (p^\star, \theta) + \lambda^\star \frac{\partial \mathcal{F}_B}{\partial p} (p^\star, \theta) &= 0, \label{eq:kkt_p}\\
    \lambda^\star (\mathcal{F}_B(p^\star, \theta) - 1) &= 0, \label{eq:kkt_B}\\
    \lambda^\star &\geq 0. \label{eq:kkt_dual_pos}
\end{align}
\end{subequations}

\begin{lemma}\label{lemma:KKT}
    Let $\mathcal{F}_A$ and $\mathcal{F}_B$ be two scaling functions associated with the sets $A$ and $B$, respectively. If $A \intersect B = \varnothing$, then any pair of optimal primal variable $p^\star$ and dual variable $\lambda^\star$ of \eqref{eq:opt} satisfy: $\mathcal{F}_A (p^\star, \theta) > 1$, $\lambda^\star>0$, $\frac{\partial \mathcal{F}_B}{\partial p} (p^\star, \theta) \neq 0$, and $\mathcal{F}_B(p^\star, \theta) = 1$.
\end{lemma}
\begin{proof}
    First, it is easy to see $\mathcal{F}_A (p^\star, \theta) > 1$ as $A \intersect B = \varnothing$ and the problem is feasible. Next, we prove $\lambda^\star>0$. If $\lambda^\star = 0$, we get $\frac{\partial \mathcal{F}_A}{\partial p} (p^\star, \theta)=0$ from \eqref{eq:kkt_p}, which contradicts Lemma~\ref{lemma:non_zero_grad} as $\mathcal{F}_A (p^\star, \theta) > 1$. Therefore, $\lambda^\star>0$. As a result of $\lambda^\star>0$ and $\frac{\partial \mathcal{F}_A}{\partial p} (p^\star, \theta) \neq 0$, it follows that $\frac{\partial \mathcal{F}_B}{\partial p} (p^\star, \theta) \neq 0$ by \eqref{eq:kkt_p} and 
    $\mathcal{F}_B(p^\star, \theta) = 1$ by \eqref{eq:kkt_B}.
\end{proof}

\begin{lemma}\label{lemma:uniqueness}
Let $\mathcal{F}_A$ and $\mathcal{F}_B$ be two scaling functions associated with the sets $A$ and $B$, respectively. If $A \intersect B = \varnothing$ and one of $\mathcal{F}_A$ and $\mathcal{F}_B$ is strictly convex in $p$ for each $\theta$, the optimal primal variable $p^\star$ and dual variable $\lambda^\star$ are unique. 
\end{lemma}
\begin{proof}
    The proof is provided in Appendix~\ref{sec:proof_uniqueness}.
\end{proof}

The following theorem establishes the first-order continuous differentiability of $\alpha^\star$ w.r.t. $\theta$.

\begin{theorem}\label{thm:cd_alpha}
    Let $\mathcal{F}_A$ and $\mathcal{F}_B$ be two scaling functions associated with the sets $A$ and $B$, respectively. If $A \intersect B = \varnothing$ and one of $\mathcal{F}_A$ and $\mathcal{F}_B$ has a positive definite Hessian in $p$ for each $\theta$, then the optimal value $\alpha^\star$ of problem \eqref{eq:opt} is $\mathcal{C}^1$ in $\theta$.
\end{theorem}

\begin{proof}
    We know from Lemma~\ref{lemma:uniqueness} that the optimal primal variable $p^\star$ and dual variable $\lambda^\star$ are unique. Differentiating \eqref{eq:kkt_p} and $\mathcal{F}_B(p^\star, \theta) = 1$ w.r.t. $\theta$ gives
    \begin{equation}\label{eq:implicit_eq}
        \underbrace{\begin{bmatrix}
            M & c^\top \\ c & 0
        \end{bmatrix}}_{N}
        \begin{bmatrix}
            \frac{\partial p^\star}{\partial \theta} \\ \frac{\partial \lambda^\star}{\partial \theta}
        \end{bmatrix}=
        \underbrace{\begin{bmatrix}
            \Omega_1 \\ \Omega_2
        \end{bmatrix}}_{\Omega}
    \end{equation}
    where
    \begin{subequations}
    \begin{align}
        M &= \frac{\partial^2 \mathcal{F}_A}{\partial p^2}(p^\star, \theta) + \lambda^\star \frac{\partial^2 \mathcal{F}_B}{\partial p^2}(p^\star, \theta), \label{eq:invert_M}\\
        c &= \frac{\partial \mathcal{F}_B}{\partial p}(p^\star, \theta),\label{eq:invert_c}\\
        \Omega_1 &= - \frac{\partial^2 \mathcal{F}_A}{\partial \theta \partial p}(p^\star, \theta) - \lambda^\star \frac{\partial^2 \mathcal{F}_B}{\partial \theta \partial p}(p^\star, \theta), \label{eq:invert_omega_1} \\
        \Omega_2 &= - \frac{\partial \mathcal{F}_B}{\partial \theta}(p^\star, \theta) \label{eq:invert_omega_2}.
    \end{align}
    \end{subequations}
    Since $\mathcal{F}_A$ and $\mathcal{F}_B$ are both convex in $p$, one of them has a positive definite Hessian in $p$, and $\lambda^\star > 0$, the matrix $M$ in \eqref{eq:invert_M} is positive definite. Furthermore, the vector $c$ in \eqref{eq:invert_c} is nonzero from Lemma~\ref{lemma:KKT}. Therefore, by Schur complement, the matrix $N$ defined in \eqref{eq:implicit_eq} is invertible. By the implicit function theorem, $p^\star$ and $\lambda^\star$ are continuously differentiable in $\theta$, and their derivatives are given by 
    \begin{equation}\label{eq:gradient_implicit}
        \begin{bmatrix}
            \frac{\partial p^\star}{\partial \theta} \\ \frac{\partial \lambda^\star}{\partial \theta}
        \end{bmatrix} = N^{-1} \Omega.
    \end{equation}
    Finally, $\alpha^\star = \mathcal{F}_A(p^\star,\theta)$ is continuously differentiable in $\theta$ because $\mathcal{F}_A$ is $\mathcal{C}^2$, and its derivative is given by \begin{equation}\label{eq:alpha_gradient}
        \frac{\partial \alpha^\star}{\partial \theta}(\theta) = \frac{\partial \mathcal{F}_A}{\partial p}(p^\star, \theta) \frac{\partial p^\star}{\partial \theta}(\theta) + \frac{\partial \mathcal{F}_A}{\partial \theta}(p^\star, \theta).
    \end{equation}
\end{proof}

The next theorem proves the higher-order smoothness of the optimal value $\alpha^\star$ if the scaling functions are smooth enough.

\begin{theorem}\label{thm:high_order_cd_alpha}
    Under the same conditions of Theorem~\ref{thm:cd_alpha}, if the scaling functions $\mathcal{F}_A$ and $\mathcal{F}_B$ are $\mathcal{C}^{k+1}$ ($k \geq 1$) in $p$ and $\theta$, then $\alpha^\star$ is $\mathcal{C}^k$ in $\theta$.
\end{theorem}
\begin{proof}
    The proof is provided in Appendix~\ref{sec:proof_high_order_cd_alpha}.
\end{proof}

Theorems~\ref{thm:cd_alpha} and \ref{thm:high_order_cd_alpha} establish the first-order and higher-order sensitivities of the optimal value $\alpha^\star$ of \eqref{eq:opt} contingent on the smoothness of the scaling functions. The constructive proofs directly illustrate the computation of the derivatives to the desired order. The first-order local differentiability of solutions to general parametric programs has been studied in works such as \cite{jittorntrum1984solution, robinson1980strongly}. In contrast, Theorem~\ref{thm:cd_alpha} of this work proves global continuous differentiability under a more specific setting, and Theorem~\ref{thm:high_order_cd_alpha} further establishes high-order continuous differentiability of the minimal scaling factor when the scaling functions are sufficiently smooth.

\begin{remark}
    Existing works (e.g., \cite{wei2024diffocclusion, dai2023safe, tracy2023differentiable}) employ an optimization problem of (or similar to) the form
    \begin{equation} \label{eq:opt_old}
        \begin{aligned}
        \alpha^\star (\theta) = \min_{p \in \R^{n_p}, \alpha \in \R} \quad & \alpha \\
        \textrm{s.t.} \quad & \mathcal{F}_A (p, \theta) \leq \alpha, \quad \mathcal{F}_B (p, \theta) \leq \alpha,  \\
        \end{aligned}
    \end{equation}
    and only involve first-order derivatives of $\alpha^\star$. Compared with \eqref{eq:opt_old}, our formulation \eqref{eq:opt} requires one fewer decision variable and one fewer constraint, and the metrics provided by \eqref{eq:opt} and \eqref{eq:opt_old} are different. Moreover, \eqref{eq:opt} preserves continuous differentiability under similar conditions and admits closed-form solutions in some cases. Specifically, for an ellipsoid-ellipsoid pair, the analytical solution is given in Section~\ref{sec:john_ellipsoid}, and for an ellipsoid-plane pair, \eqref{eq:opt} reduces to a QP with one inequality constraint, whose solution can be given in analytical form. Note that the works \cite{wei2024diffocclusion, dai2023safe, tracy2023differentiable} employ velocity control, while this work deals with torque control (see Section~\ref{sec:experiments}), requiring a more efficient differentiable optimization formulation.
\end{remark}

\section{Construction of Scaling Functions}\label{sec:scaling_funcs}
This section presents a systematic method for constructing scaling functions for convex primitives (see Table~\ref{tab:scaling_functions} for the proposed scaling functions). The scaling functions are first defined in the body frame and then mapped to the world frame. Let $o_b^w \in \R^{n_p}$ be the vector describing the origin of the body frame w.r.t. the world frame and $R_b^w \in \R^{n_p \times n_p}$ the rotation matrix of the body frame w.r.t. the world frame. Then, the relationship between the same vector expressed in the world frame $\prescript{w}{}{p} \in \R^{n_p}$ and in the body frame $\prescript{b}{}{p} \in \R^{n_p}$ is
\begin{equation}\label{eq:pos_b_to_w}
    \prescript{w}{}{p} = R_b^w \prescript{b}{}{p} + o_b^w.
\end{equation}
Given a set $A$, the scaling function is first defined in the body frame, e.g., $\prescript{b}{}{\mathcal{F}}_A (\prescript{b}{}{p})$, and then mapped to the world frame by
\begin{equation}\label{eq:scaling_fun_from_b_to_w}
    \mathcal{F}_A (p, \theta) = \prescript{b}{}{\mathcal{F}}_A (R_b^{w \top} (p - o_b^w))
\end{equation}
where we drop the prescript $\prescript{w}{}{(\cdot)}$ when it does not cause confusion.
The parameters $\theta$ in this case represent the orientation and translation of the set $A$. In the 2D case ($n_p = 2$), $R_b^w$ can be represented by a single rotation angle $\beta \in \R$ and therefore
\begin{equation}\label{eq:params_2d}
    \theta = [o_b^{w \top}, \beta]^\top \in \R^3.
\end{equation}
In the 3D case ($n_p = 3$), $R_b^w$ can be represented by a unit quaternion $\xi \in \R^4$ and we have 
\begin{equation}\label{eq:params_3d}
    \theta = [o_b^{w \top}, \xi^\top]^\top \in \R^7.
\end{equation}

\begin{table*}[t]
    \centering
    \caption{Proposed scaling functions for different shapes in 2D and 3D and their convexity.}
    \label{tab:scaling_functions}
    \begin{tabular}{cccc}
        \toprule
        Shape & Dimension & Proposed scaling function in the body frame & Convexity\\
        \midrule
        Planes & 2D or 3D & $\prescript{b}{}{\mathcal{F}}_H (\prescript{b}{}{p}) = a^\top \prescript{b}{}{p} + b $ & Convex \\
        Polygons/Polytopes & 2D or 3D & $\prescript{b}{}{\mathcal{F}}_{P_r} (\prescript{b}{}{p}) = \frac{1}{\kappa} \ln \left( \frac{1}{N} \sum_{i=1}^N e^ { \kappa \left(a_i^\top \prescript{b}{}{p} + b_i \right) } \right)+1, \ \kappa > 0$ & Convex \\
        Ellipses/Ellipsoids & 2D or 3D & $\prescript{b}{}{\mathcal{F}}_{\mathcal{E}} (\prescript{b}{}{p}) = (\prescript{b}{}{p} - \mu)^\top P (\prescript{b}{}{p} - \mu), \ P \in \mathbb{S}_{++}^{n_p}$ & Strongly convex \\
        \bottomrule
    \end{tabular}
\end{table*}

\subsection{Convex Polygons or Polytopes}
\begin{figure}[t]
 \centering
 \begin{subfigure}[b]{0.14\textwidth}
     \centering
     \includegraphics[width=\textwidth]{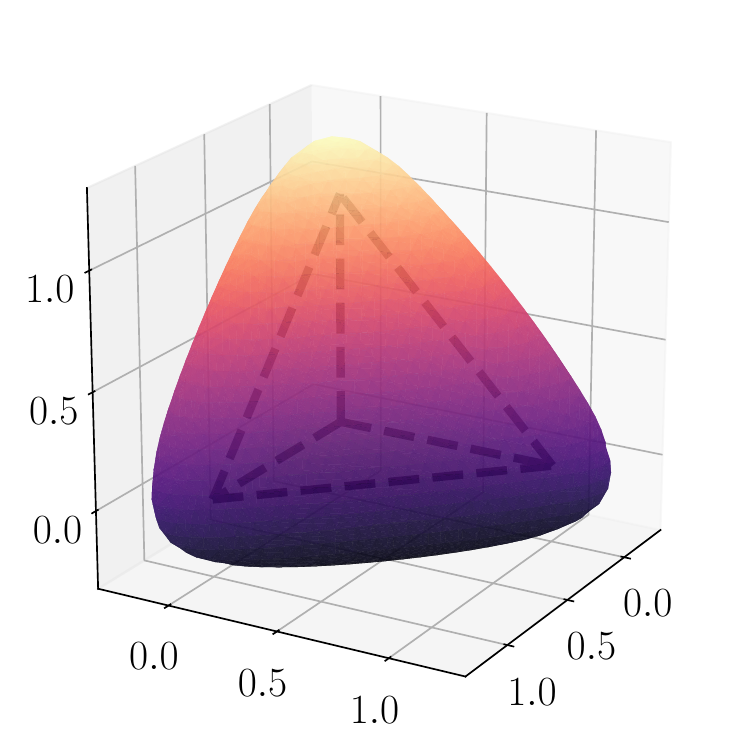}
     \caption{$\kappa=4$.}
     \label{fig:tetrahedron_keq4}
 \end{subfigure}
 \begin{subfigure}[b]{0.14\textwidth}
     \centering
     \includegraphics[width=\textwidth]{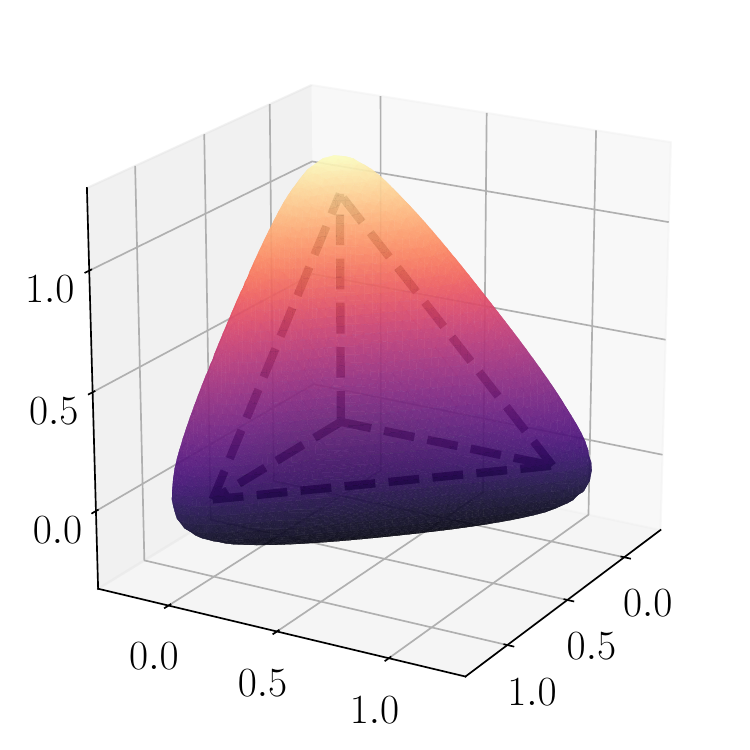}
     \caption{$\kappa=6$.}
     \label{fig:tetrahedron_keq6}
 \end{subfigure}
 \begin{subfigure}[b]{0.14\textwidth}
     \centering
     \includegraphics[width=\textwidth]{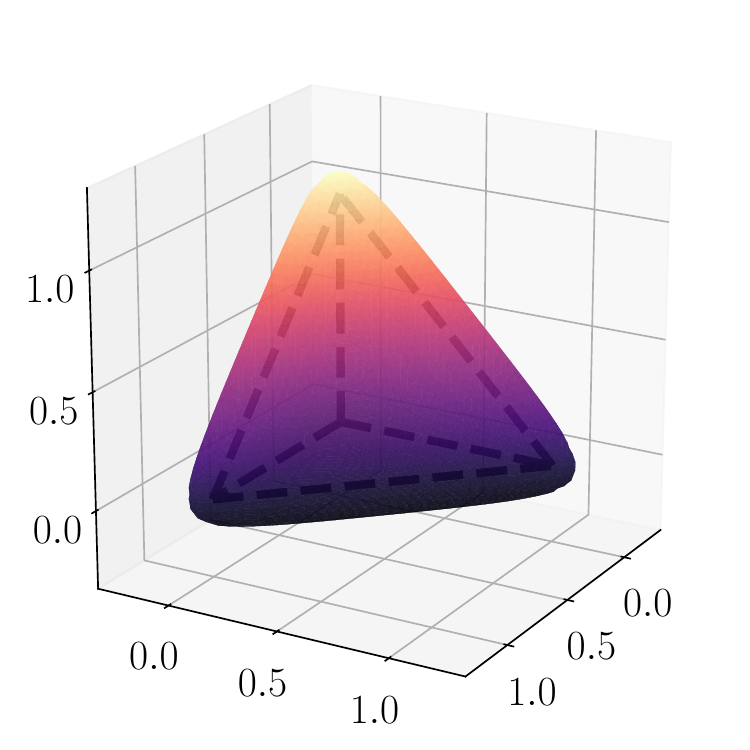}
     \caption{$\kappa=10$.}
     \label{fig:tetrahedron_keq10}
 \end{subfigure}
 \begin{subfigure}[b]{0.14\textwidth}
     \centering
     \includegraphics[width=\textwidth]{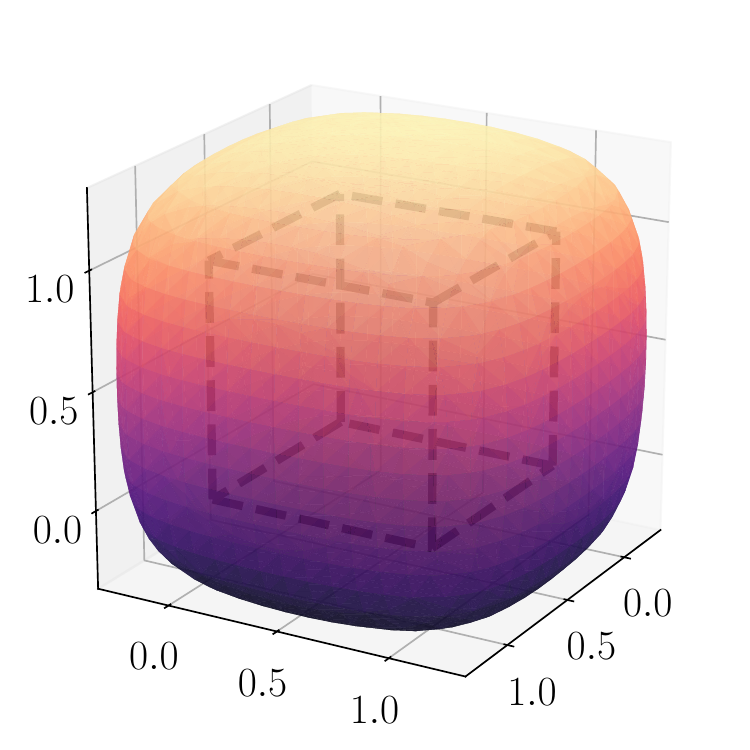}
     \caption{$\kappa=4$.}
     \label{fig:cube_keq4}
 \end{subfigure}
 \begin{subfigure}[b]{0.14\textwidth}
     \centering
     \includegraphics[width=\textwidth]{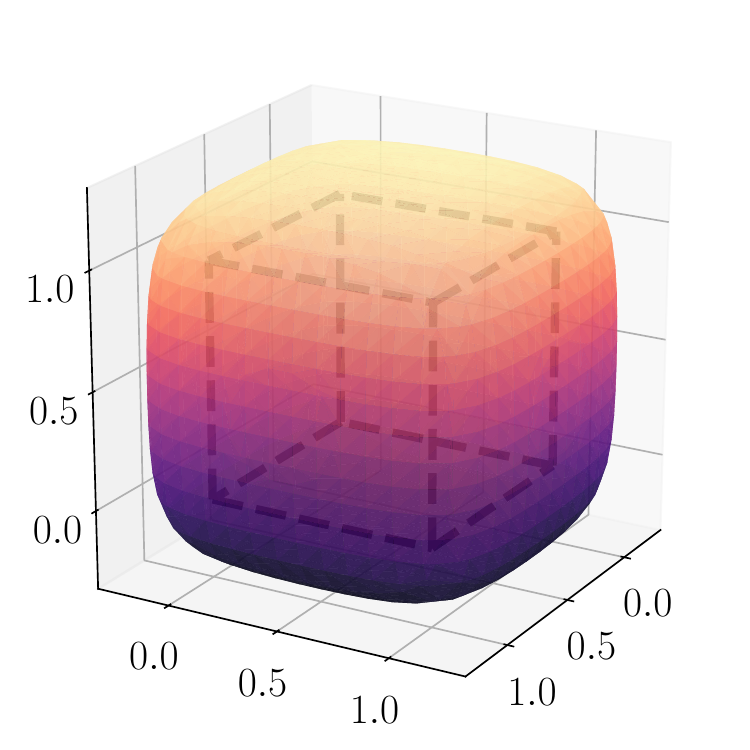}
     \caption{$\kappa=6$.}
     \label{fig:cube_keq6}
 \end{subfigure}
 \begin{subfigure}[b]{0.14\textwidth}
     \centering
     \includegraphics[width=\textwidth]{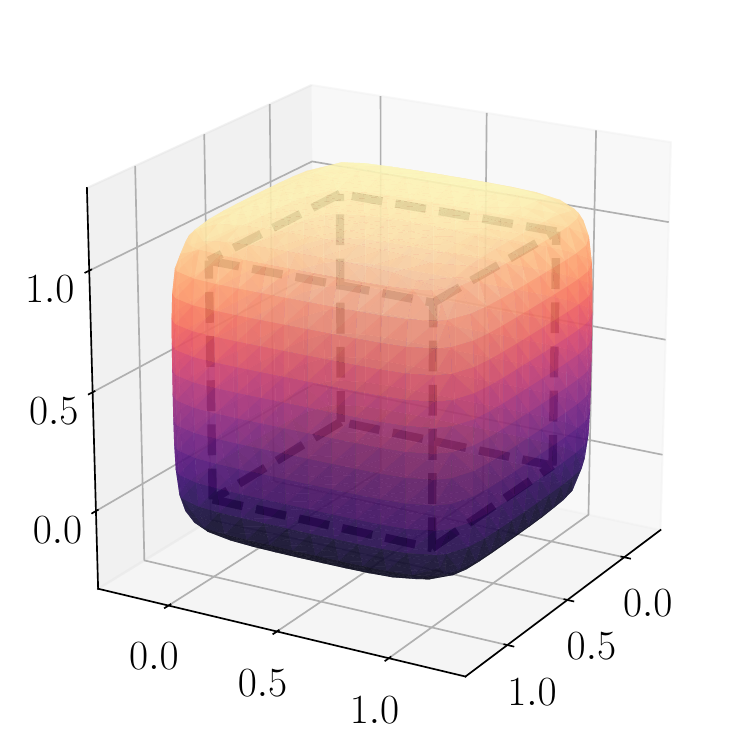}
     \caption{$\kappa=10$.}
     \label{fig:cube_keq10}
 \end{subfigure}
\caption{\subref{fig:tetrahedron_keq4}-\subref{fig:tetrahedron_keq10}: Padded tetrahedrons with $\kappa = 4$, $6$, and $10$. \subref{fig:cube_keq4}-\subref{fig:cube_keq10}: Padded cubes with $\kappa = 4$, $6$, and $10$.}
\label{fig:padded_shapes}
\end{figure}

Consider a convex polytope (or polygon) $P$ defined by a set of inequalities $P = \{ \prescript{b}{}{p} \in \R^{n_p} \mid a_i^\top \prescript{b}{}{p} + b_i \leq 0, \ i = 1, ..., N \}$ with $a_i \in \R^{n_p}$, $b_i \in \R$, and $n_p = 2$ or $3$. We propose the following relaxation of $P$ with a tunable accuracy
\begin{equation}
    \prescript{b}{}{\mathcal{F}}_{P_r} (\prescript{b}{}{p}) = \frac{1}{\kappa} \ln \left( \frac{1}{N} \sum_{i=1}^N e^ { \kappa \left(a_i^\top \prescript{b}{}{p} + b_i \right) } \right)+1
\end{equation}
where $\kappa > 0$ is a tunable scalar. The function $\prescript{b}{}{\mathcal{F}_{P_r}}$ is convex in $p$ because it is the composition of the standard log-sum-exp function and an affine function. It is easy to see that $P \subset P_r = \{ \prescript{b}{}{p} \in \R^{n_p} \mid \prescript{b}{}{\mathcal{F}}_{P_r} (\prescript{b}{}{p}) \leq 1 \}$ because any $\prescript{b}{}{p} \in P$ will satisfy $a_i^\top \prescript{b}{}{p} + b_i \leq 0$ and thus $\prescript{b}{}{\mathcal{F}}_{P_r} (\prescript{b}{}{p}) \leq 1$. In the world frame, the scaling function is
\begin{equation}\label{eq:log_sum_exp}
    \mathcal{F}_{P_r} (p,\theta) = \frac{1}{\kappa} \ln \left( \frac{1}{N} \sum_{i=1}^N e^ {\kappa  (a_i'^\top p + b_i')} \right) + 1
\end{equation}
where $a_i ' = R_b^w a_i$ and $b_i' = b_i - a_i^\top R_b^{w \top} o_b^w$. The tunable accuracy of padded shape $P_r$ is visualized in Fig.~\ref{fig:padded_shapes} for different values of $\kappa$. See also \cite{gonccalves2024smooth} for some other smooth distance techniques. We provide additional details on how the numerical stability can be improved in the evaluation of \eqref{eq:log_sum_exp} and its derivatives in Appendix~\ref{sec:log_sum_exp_suggestions}.

\subsection{Minimum Volume Ellipsoids}\label{sec:john_ellipsoid}
The minimum volume ellipsoid (or ellipse in 2D) containing a given set is called the Löwner-John ellipsoid of that set \cite[Section 8.4]{boyd2004convex}, and can be used as an efficient method to find scaling functions. Consider a set of points $C = \{ \prescript{b}{}{p}_1,..., \prescript{b}{}{p}_N \}$ expressed in the body frame. A general ellipsoid in the body frame can be characterized as $\mathcal{E} = \{ \prescript{b}{}{p} \in \R^{n_p} \mid \lVert D \prescript{b}{}{p} + d \rVert_2 \leq 1 \}$ with $D \in \mathbb{S}_{++}^{n_p}$ and $d \in \R^{n_p}$. Then, the minimum volume ellipsoid containing the set $C$ can be found by solving

\begin{equation} \label{eq:min_vol_ellipsoid}
    \begin{aligned}
    \min_{D \in \mathbb{S}_{++}^{n_p}, d \in \R^{n_p}} \quad & \ln \det D^{-1} \\
    \textrm{s.t.} \quad & \lVert D \prescript{b}{}{p}_i + d \rVert_2 \leq 1, \ i = 1,...,N. \\
    \end{aligned}
\end{equation}

One can construct the scaling function by
\begin{equation}
    \prescript{b}{}{\mathcal{F}}_{\mathcal{E}} (\prescript{b}{}{p}) = (\prescript{b}{}{p} - \mu')^\top P' (\prescript{b}{}{p} - \mu')
\end{equation}
where $P' = D^2$ and $\mu' = - D^{-1}d$. In the world frame, the scaling function is
\begin{equation}
    \mathcal{F}_\mathcal{E} (p,\theta) = (p-\mu)^\top P (p-\mu)
\end{equation}
where $P = R_b^w P' R_b^{w \top}$ and $\mu = o_b^w + R_b^w \mu'$. When both scaling functions take the above form, a closed-form solution to \eqref{eq:opt} can be found directly.

\begin{theorem}[Adapted from \cite{rimon1997obstacle}]\label{thm:rimon_solution}
    Let $\mathcal{F}_A(p,\theta) = (p-\mu_A)^\top P_A (p-\mu_A)$ and $\mathcal{F}_B(p,\theta) = (p-\mu_B)^\top P_B (p-\mu_B)$ with $P_A, P_B \in \mathbb{S}_{++}^{n_p}$. Let $L_A$ be the Cholesky decomposition of $P_A$ such that $P_A = L_A L_A^\top$. Then, the optimal solution $p^\star$ to \eqref{eq:opt} is
    \begin{equation}
        p^\star = \mu_A + \lambda_{\min} (M) L_A^{- \top} [\lambda_{\min} (M) \bar{P}_B - I]^{-1} \bar{P}_B \bar{\mu}_B
    \end{equation}
    where $\lambda_{\min} (M)$ is the minimal eigenvalue of 
    \begin{equation}
        M = \begin{bmatrix}
        \Tilde{P}_B & -I \\
        -\Tilde{\mu}_B \Tilde{\mu}_B^\top & \Tilde{P}_B
        \end{bmatrix}
    \end{equation}
    with $\Tilde{P}_B = \bar{L}_B^{-1} \bar{L}_B^{-\top}$, $\Tilde{\mu}_B = \bar{L}_B^{-1} \bar{\mu}_B$, $\bar{P}_B = L_A^{-1} P_B L_A^{- \top }$, $\bar{\mu}_B = L_A^{\top}(\mu_B-\mu_A)$, and $\bar{L}_B$ is the Cholesky decomposition of $\bar{P}_B$ such that $\bar{P}_B = \bar{L}_B \bar{L}_B^{\top}$.
\end{theorem}
\begin{remark}
   The result presented in Theorem~\ref{thm:rimon_solution} has a few refinements compared with that in \cite[Section 3.1]{rimon1997obstacle}. We use the Cholesky decomposition instead of the eigendecomposition ($P_A = P_A^{1/2} P_A^{1/2}$ and $P_B = P_B^{1/2} P_B^{1/2}$), because the Cholesky decomposition is usually computationally cheaper and numerically more stable. Furthermore, operations such as $L^{-1}v$ (or $L^{-\top}v$) are efficiently done by forward (or backward) substitution for a lower triangular matrix $L$ and a vector $v$. 
\end{remark}

\section{HOCBFs and Circulation Mechanism to Avoid Spurious Equilibria}\label{sec:hocbf}
In this section, we will first formulate differentiable optimization-based HOCBFs (see Section~\ref{sec:dohocbf}) and then introduce a circulation mechanism to mitigate the issue of spurious equilibria (see Section~\ref{sec:circ_mechanism}).

\subsection{Differentiable Optimization Based HOCBFs}\label{sec:dohocbf}
Consider a robotic system (e.g., a robotic manipulator) described by the Lagrangian dynamics \cite{siciliano2008robotics}
\begin{equation}\label{eq:lagrange_mechanics}
    M(q) \ddot{q} + \underbrace{C(q, \dot{q}) \dot{q} + F_f(\dot{q}) + G(q)}_{\sigma(q, \dot{q})} = u
\end{equation}
where $q \in \R^{n_q}$ are the generalized coordinates, $u \in \R^{n_q}$ are the input torques, $M(q) \in \mathbb{S}^{n_q}_{++}$ is the inertia matrix, $C(q, \dot{q}) \in \R^{n_q \times n_q}$ accounts for the centrifugal and Coriolis effects, $F_f(\dot{q}) \in \R^{n_q}$ is the friction vector (with $F_f(0) = 0$), and $G(q) \in \R^{n_q}$ is the gravity vector. Then, we have 
\begin{equation}
    \frac{d}{dt} \underbrace{\begin{bmatrix}
        q \\ \dot{q}
    \end{bmatrix}}_{x} = \underbrace{\begin{bmatrix}
        \dot{q} \\ -M^{-1}(q) \sigma (q, \dot{q})
    \end{bmatrix}}_{f(x)} + \underbrace{\begin{bmatrix}
        0 \\ M^{-1}(q)
    \end{bmatrix}}_{g(x)} u.
\end{equation}

We model a robot (or its part) as a set $A$ with the scaling function $\mathcal{F}_A$ and an obstacle as $B$ with $\mathcal{F}_B$. The parameter vector $\theta$ encodes the orientation and translation of the robot body frame (see \eqref{eq:params_2d} and \eqref{eq:params_3d} for its definition in 2D and 3D, respectively). In 3D, $\theta = [o_b^{w \top}, \xi^\top]^\top \in \R^7$, and $\theta$ is a function of $x = [q^\top, \dot{q}^\top]^\top$ by the direct kinematic equation, and we omit this dependence for brevity. Then, 
\begin{subequations}
\begin{align}
    \dot{o}_b^w &= v, \\ 
    \dot{\xi} &= \frac{1}{2} \underbrace{\begin{bmatrix}
        \xi_w & \xi_z & -\xi_y \\
        -\xi_z & \xi_w & \xi_x \\
        \xi_y & -\xi_x & \xi_w \\
        -\xi_x & -\xi_y & -\xi_z 
    \end{bmatrix}}_{Q(\xi)} \omega,
\end{align}
\end{subequations}
and the linear and angular velocities $v, \omega \in \R^3$ are given by
\begin{equation} \label{eq:v_Jq}
    \begin{bmatrix}
        v \\ \omega 
    \end{bmatrix} = J(q) \dot{q}
\end{equation}
where $J(q) \in \R^{6 \times n_q}$ is the geometric Jacobian matrix. Define the HOCBF $h: \R^{n_x} \rightarrow \R$ by
\begin{equation}\label{eq:cbf}
    h(x) = \alpha^\star (\theta(x)) - \alpha_0
\end{equation}
with $\alpha_0 > 1$. As the system \eqref{eq:lagrange_mechanics} is torque controlled, the relative degree of $h$ is $m=2$. Then, 
\begin{subequations}
\begin{align}
    \dot{h}(x) &= \frac{\partial \alpha^\star}{\partial \theta} \dot{\theta} = \frac{\partial \alpha^\star}{\partial \theta} \underbrace{\begin{bmatrix}
        I & 0 \\ 0 & \frac{1}{2} Q(\xi)
    \end{bmatrix}}_{T(\xi)} J(q) \dot{q} \label{eq:dh_first_order}, \\
    \ddot{h}(x) &= (T(\xi) J(q) \dot{q})^\top \frac{\partial^2 \alpha^\star}{\partial \theta^2} T(\xi) J(q) \dot{q}  \notag \\
    +& \frac{\partial \alpha^\star}{\partial \theta} [\dot{T}(\xi) J(q) \dot{q} + T(\xi) \dot{J}(q) \dot{q} + T(\xi) J(q) \ddot{q}],
\end{align}
\end{subequations}
where $\frac{\partial^2 \alpha^\star}{\partial \theta^2}$ can be calculated following Appendix~\ref{sec:cal_of_hessian}. The coefficient before the control input $u$ in $\ddot{h}(x)$ is 
\begin{equation}\label{eq:actuation_matrix}
    L_g L_f^{m-1} h (x) = \frac{\partial \alpha^\star}{\partial \theta} T(\xi) J(q) M^{-1}(q).
\end{equation}

\begin{proposition}
Under the conditions of Theorem \ref{thm:high_order_cd_alpha}, if $h(x) \geq 0$ and the geometric Jacobian matrix $J(q)$ has full row rank, then the vector $L_g L_f^{m-1} h (x)$ defined in \eqref{eq:actuation_matrix} is nonzero.
\end{proposition}
\begin{proof}
When $h(x) \geq 0$, we have $\alpha^\star (\theta) \geq \alpha_0 > 1$ and $A \intersect B = \varnothing$. From \eqref{eq:v_Jq} and \eqref{eq:dh_first_order}, we see that $\dot{h}(x) = \dot{\alpha}^\star (\theta) = \frac{\partial \alpha^\star}{\partial \theta} T(\xi) [v^\top, \omega^\top]^\top$. If $\frac{\partial \alpha^\star}{\partial \theta} T(\xi) = 0$, then no infinitesimal translation or rotation can change the value of $\alpha^\star$. This is impossible as $A \intersect B = \varnothing$ and one of $\mathcal{F}_A$ and $\mathcal{F}_B$ has a positive definite Hessian in $p$ (hence strictly convex), which means that there always exists some infinitesimal translation or rotation that will change the value of $\alpha^\star$. As $J(q)$ has full row rank and $M(q)$ is nonsingular, $L_g L_f^{m-1} h (x) \neq 0$.  
\end{proof}

Let the HOCBF-QP of this work be 
\begin{equation} \label{eq:hocbf_qp}
    \begin{aligned}
    \pi (x) = \argmin_{u \in \mathcal{U}} \quad & \lVert u-u_n(x) \rVert_2^2 \\
    \textrm{s.t.} \quad & a(x)^\top u \geq b(x),  \\
    \end{aligned}
\end{equation}
where $u_n$ is a nominal controller, $a(x) = (L_g L_f^{m-1} h (x))^\top$ and $b(x) = -L_f^{m} h (x) - \sum_{i=0}^{m-1} L_f^i (\Gamma_{m-i} \circ \psi_{m-i-1}) (x)$ with $m=2$. Assuming that \eqref{eq:hocbf_qp} is feasible, if the initial condition $x_0 \in C_1 \intersect C_2$ (see \eqref{eq:high_order_safe_sets} for definitions of $C_1$ and $C_2$), then $h(x) \geq 0$ holds forward in time by \eqref{eq:control_cbf}, i.e., the collision between $A$ and $B$ is avoided. The solution of \eqref{eq:hocbf_qp} is unique (if it exists) because the objective function is strongly convex.

\subsection{Circulation Mechanism to Avoid Spurious Equilibria}\label{sec:circ_mechanism}
Existing works (such as \cite{reis2020control, tan2024undesired, gonccalves2024control}) point out that CBF-based QPs may introduce spurious equilibria in cases where the CBF is of relative degree one. We demonstrate that these equilibria also exist in high-order cases and propose a circulation mechanism to avoid them for systems such as \eqref{eq:lagrange_mechanics}. 

\begin{figure}[t]
 \centering
 \begin{subfigure}[b]{0.24\textwidth}
     \centering
     \includegraphics[width=\textwidth]{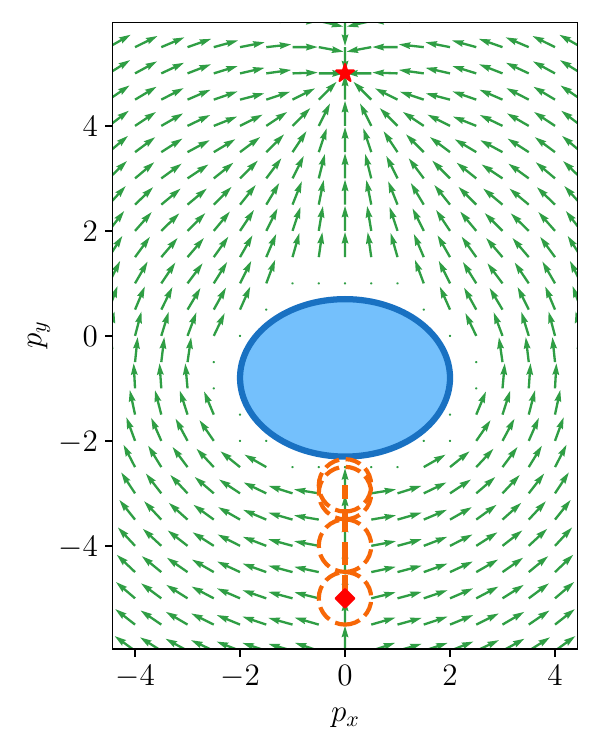}
     \caption{Without circulation.}
     \label{fig:wo_circulation}
 \end{subfigure}
  \begin{subfigure}[b]{0.24\textwidth}
     \centering
     \includegraphics[width=\textwidth]{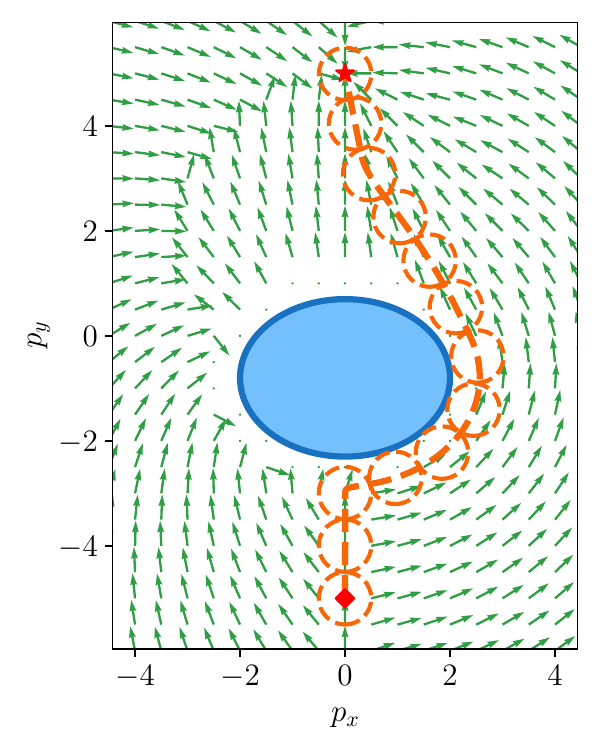}
     \caption{With circulation.}
     \label{fig:w_circulation}
 \end{subfigure}
\caption{A 2D navigation example with and without the circulation mechanism. The elliptical obstacle is colored in blue, and the robot trajectory is colored in orange. The green arrows represent the control $\pi$ (given by \eqref{eq:hocbf_qp}) when the robot starts from different positions.}
\label{fig:circulation_2d}
\end{figure}

\begin{example}\label{eg:2d_integrator}
Consider a 2D double integrator system navigating on a plane with dynamics
\begin{equation}\label{eq:eg_double_integrator_2d}
    \frac{d}{dt} \underbrace{\begin{bmatrix}
        p \\
        v
    \end{bmatrix}}_{x} = \underbrace{\begin{bmatrix}
        v \\
        0
    \end{bmatrix}}_{f(x)} + 
    \underbrace{\begin{bmatrix}
        0 \\ 
        I
    \end{bmatrix}}_{g(x)} u,
\end{equation}
where $p = [p_x, p_y]^\top$, $v = [v_x, v_y]^\top$, $u = [u_x, u_y]^\top$. The robot is modeled by a ball (with a radius of $0.5$) initially centered at $(0.0,-5.0)$, and the obstacle is modeled by an ellipse (with a semi-minor axis of $1.5$ and a semi-major axis of $2.0$) centered at $(0.0,-0.8)$. To reach the goal position $p_g = [0.0, 5.0]^\top$, the following feedback control is chosen as the nominal control
\begin{equation}
    u_n = - K_p (p-p_g) - K_d v
\end{equation}
with $K_p, K_d > 0$. As in Fig.~\ref{fig:wo_circulation}, the safe control $\pi$ given by \eqref{eq:hocbf_qp} leads to an equilibrium $x_e = [0.0, -2.8, 0.0, 0.0]$ due to the symmetry in this particular setting. 
\end{example}

The closed-loop system equilibrium $(x_e, u_e)$ of \eqref{eq:affine_sys} satisfies 
\begin{equation}\label{eq:equilibrium}
    0 = f(x_e) + g(x_e) u_e.
\end{equation}
We note that the set of equilibria might not be discrete. For torque-controlled robotic systems as in \eqref{eq:lagrange_mechanics}, any static configuration with only gravity compensation can be an equilibrium, which will be made clear in later analysis. The following proposition describes in detail when the solution to \eqref{eq:hocbf_qp} leads to an equilibrium for general control-affine systems \eqref{eq:affine_sys} with an HOCBF of relative degree $m$.

\begin{proposition}\label{prop:equilibrium_on_boundary}
    Assume that \eqref{eq:hocbf_qp} is feasible, $\pi(x) \notin \partial \mathcal{U}$, and $h$ is of relative degree $m$ for the system \eqref{eq:affine_sys}. If the system \eqref{eq:affine_sys} controlled by $\pi$ (given by \eqref{eq:hocbf_qp}) enters an equilibrium $(x_e, u_e)$ satisfying \eqref{eq:equilibrium}, only two cases can happen: (i) $u_n (x_e) = u_e$; (ii) $h(x_e) = 0$.
\end{proposition}
\begin{proof}
    Assuming $\pi(x) \notin \partial \mathcal{U}$, we discuss whether the constraint $a(x)^\top u \geq b(x)$ is active (binding). If this constraint is not active, we have $u_e = \pi(x_e) = u_n (x_e)$, which proves the first case. If the constraint is active, we have $a(x_e)^\top u_e = b(x_e)$. Then, $0 = a(x_e)^\top u_e - b(x_e) = L_f^m h (x_e) + L_g L_f^{m-1} h (x_e) u_e + \sum_{i=0}^{m-1} L_f^i (\Gamma_{m-i} \circ \psi_{m-i-1}) (x_e) = \dot{\psi}_{m-1} (x_e) + \Gamma_m (\psi_{m-1} (x_e))$ by the construction in \eqref{eq:intermediate_barrier_funcs}. As $\dot{\psi}_{m-1} (x_e) = 0$ (by \eqref{eq:equilibrium}), we have $\Gamma_m (\psi_{m-1} (x_e))=0$ and thus $\psi_{m-1} (x_e) = 0$ as $\Gamma_m$ is a class $\mathcal{K}$ function. By recurrence, we obtain that $\psi_{m-1} (x_e) = ... = \psi_0 (x_e) = 0$. As $\psi_0 = h$, we have $h(x_e) = 0$, which proves the second case.
\end{proof}

Proposition~\ref{prop:equilibrium_on_boundary} shows that, when $\pi(x) \notin \partial \mathcal{U}$, the equilibria introduced by \eqref{eq:hocbf_qp} belong to the set $\partial C_1 = \{x \in \R^{n_x} \mid h(x) = 0\}$ with $C_1$ defined in \eqref{eq:high_order_safe_sets} unless commanded by the nominal control $u_n$ (i.e., case (i)). If we further assume that $a (x) \neq 0$, the solution to \eqref{eq:hocbf_qp} is 
\begin{equation}
    \pi(x) = u_n (x) - \frac{\min (0, a(x)^\top u_n (x) - b(x))}{a(x)^\top a (x)} a(x),
\end{equation}
which can be seen as a projection onto the hyperplane $H: a(x)^\top u = b(x)$. If the result of the projection is $\pi(x_e) = u_e$ at $x=x_e$, an equilibrium is created. To mitigate this issue, we propose a circulation mechanism. 

\begin{algorithm}[t]
\caption{Implementation of the circulation mechanism}
\label{alg:circulation_mechanism}
\begin{algorithmic}[1]
\REQUIRE System dynamics $f$ and $g$, HOCBF $h$, nominal controller $u_n$;

\STATE Compute the set of equilibria $\mathcal{X}_e$;

\STATE Define the projection operator onto $\mathcal{X}_e$: $P_{\mathcal{X}_e} (x) = \argmin \{ \lVert x - x' \rVert_2 \mid x' \in \mathcal{X}_e \}$;

\STATE Define $\zeta (x_e) = -g(x_e)^\dagger f(x_e)$ for all $x_e \in \mathcal{X}_e$;

\STATE Decide a skew-symmetric matrix $\Phi$ and a function $d$;

\FOR{$i=1,2,...$} 
\STATE Compute the CBF inequality: $a(x) = (L_g L_f^{m-1} h (x))^\top$ and $b(x) = -L_f^{m} h (x) - \sum_{i=0}^{m-1} L_f^i (\Gamma_{m-i} \circ \psi_{m-i-1}) (x)$;

\STATE Compute the circulation inequality: $c(x) = \Phi a(x)$, $\zeta(P_{\mathcal{X}_e} (x))$, and $d(h(x), \lVert x - P_{\mathcal{X}_e} (x) \rVert_2)$;

\STATE Solve \eqref{eq:chocbf_qp} to obtain $\pi'(x)$;
\ENDFOR
\end{algorithmic}
\end{algorithm}

\begin{assumption} \label{ass:circulation_ass}
We assume that system \eqref{eq:affine_sys} satisfies the following properties:
\begin{enumerate}
\item The set $\mathcal{X}_e = \{ x_e \in \R^{n_x} \mid \exists u_e \in \mathcal{U},  0 = f(x_e) + g(x_e) u_e \}$ is closed, convex, and known;
\item For $x_e \in \mathcal{X}_e$, $g(x_e)$ has full column rank.
\end{enumerate}
\end{assumption}

We remark that under the assumption $\mathcal{X}_e$ is a closed convex set, the projection (from $\R^{n_x}$ onto $\mathcal{X}_e$) $P_{\mathcal{X}_e} (x) = \argmin \{ \lVert x - x' \rVert_2 \mid x' \in \mathcal{X}_e \}$ is well defined and a continuous function in its argument \cite[Chapter 8]{boyd2004convex}. In addition, as $g(x_e)$ has full column rank and $f(x_e)$ is in the column space of $g(x_e)$, there exists one unique $u_e \in \R^{n_u}$ such that $0 = f(x_e) + g(x_e) u_e$, and $u_e$ is a continuous function of $x_e$, i.e., $u_e = \zeta (x_e) = -g(x_e)^\dagger f(x_e)$. We note that Assumption~\ref{ass:circulation_ass} requires pre-computing $\mathcal{X}_e$, and for torque-controlled robotic systems as in \eqref{eq:lagrange_mechanics}, Assumption~\ref{ass:circulation_ass} holds with $\mathcal{X}_e = \{[q^\top, 0]^\top \in \R^{2 n_q} \mid q_l \leq q \leq q_u \}$ (where $q_l, q_u \in \R^{n_q}$ are the lower and upper bounds of the joint angles, respectively) and $u_e = \zeta (x_e) = G(q)$ (i.e., gravity compensation) assuming sufficient control capacities, and thus $\lVert x - P_{\mathcal{X}_e} (x) \rVert_2 = \lVert \dot{q}\rVert_2$. Next, we introduce the following Circulating HOCBF-QP (CHOCBF-QP):
\begin{equation} \label{eq:chocbf_qp}
    \begin{aligned}
    \pi' (x) &= \argmin_{u \in \mathcal{U}} \quad \lVert u-u_n(x) \rVert_2^2 \\
    \textrm{s.t.} \quad & a(x)^\top u \geq b(x),  \\
    \quad & c(x)^\top [u - \zeta(P_{\mathcal{X}_e} (x))] \geq d(h(x), \lVert x - P_{\mathcal{X}_e} (x) \rVert_2)
    \end{aligned}
\end{equation}
where $c(x)$ is a vector orthogonal to $a(x)$ and continuous in $x$, and $d(\cdot, \cdot)$ is a continuous scalar function satisfying $d(0, 0) > 0$. The extra constraint is designed to encourage solutions tangential to the hyperplane $H$ when both $h(x)$ and $\lVert x - P_{\mathcal{X}_e} (x) \rVert_2$ are close to zero (see Fig.~\ref{fig:constraint_normal}) in the aim of preventing equilibria on the boundary of the safe set. 

In the following, we discuss how to design $c(x)$ and $d(\cdot, \cdot)$. Given the vector $a(x)$, a possible choice of $c(x)$ is $c(x) = \Phi a(x)$ where $\Phi$ is a nonsingular skew-symmetric matrix (which is only possible when $n_u$ is even). If $n_u$ is odd, one can add one virtual component to the control $u$ to avoid this issue. See \cite[Section V.B]{gonccalves2024control} for a more detailed explanation of constructing skew-symmetric matrices. $d(\cdot, \cdot)$ needs to satisfy $d(0,0) > 0$. Some options can be $d(h(x), \lVert x - P_{\mathcal{X}_e} (x) \rVert_2) = 1 - d_1 h(x) - d_2 \lVert x - P_{\mathcal{X}_e} (x) \rVert_2$ or $d(h(x), \lVert x - P_{\mathcal{X}_e} (x) \rVert_2) = d_1 (1-e^{d_2 (h(x) - d_3)}) + d_4 (e^{-(\lVert x - P_{\mathcal{X}_e} (x) \rVert_2 / d_5)^2}-1)$ with $d_1, ..., d_5 > 0$. The implementation of the circulation mechanism is given in Algorithm~\ref{alg:circulation_mechanism}.

As shown in Fig.~\ref{fig:w_circulation} in Example~\ref{eg:2d_integrator}, the robot can circulate around the obstacle under the control generated by \eqref{eq:chocbf_qp}. In this example, $\mathcal{X}_e = \{ [p^\top, 0] \in \R^4 \mid p \in \R^2\}$ for \eqref{eq:eg_double_integrator_2d} (therefore $ \lVert x - P_{\mathcal{X}_e} (x) \rVert_2 = \lVert v \rVert_2$), and we used $c(x) = \begin{bmatrix} 0 & 1 \\ -1 & 0 \end{bmatrix} a(x)$ and $d(h(x), \lVert x - P_{\mathcal{X}_e} (x) \rVert_2) = 1 - h(x) - \lVert v \rVert_2$.

\begin{figure}[t]
 \centering
 \begin{subfigure}[b]{0.24\textwidth}
     \centering
     \includegraphics[width=\textwidth]{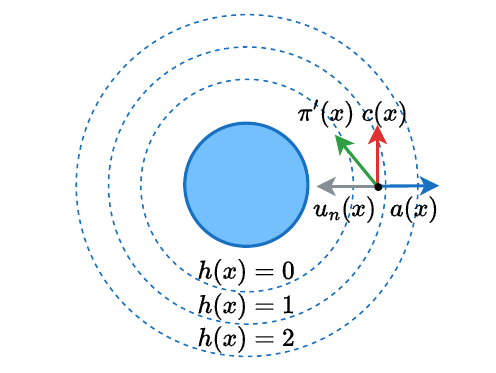}
     \caption{General case.}
     \label{fig:constraint_normal}
 \end{subfigure}
 \begin{subfigure}[b]{0.24\textwidth}
     \centering
     \includegraphics[width=\textwidth]{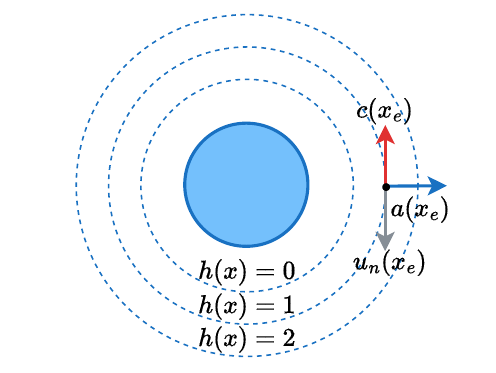}
     \caption{Equilibrium ($\pi'(x_e) = u_e$).}
     \label{fig:constraint_corner}
 \end{subfigure}
\caption{Visualization of the proposed circulation inequality. The vector $c(x)$ (in red) is orthogonal to the vector $a(x)$ (in blue). Let $d(h(x), \lVert x - P_{\mathcal{X}_e} (x) \rVert_2) = 1 - h(x)$ and $u_e = 0$. \subref{fig:constraint_normal}: $d(h(x), \lVert x - P_{\mathcal{X}_e} (x) \rVert_2) > 0$ and $\pi'(x)$ (in green) is no longer a projection of $u_n(x)$ (in gray) on $a(x)^\top u = b(x)$. \subref{fig:constraint_corner}: Equilibrium corresponding to the second case of Proposition~\ref{prop:chocbf_qp_equilibrium} where $d(h(x_e), 0) = 0$ and $u_n(x_e) = \zeta (x_e) - \rho c(x_e)$ with $\rho \geq 0$.   }
\label{fig:constraint}
\end{figure}

\begin{proposition}\label{prop:chocbf_qp_equilibrium}
    Assume that \eqref{eq:chocbf_qp} is feasible for each $x \in \mathcal{X}$, $\pi' (x) \notin \partial \mathcal{U}$, $h$ is of relative degree $m$ for the system \eqref{eq:affine_sys}, and system \eqref{eq:affine_sys} satisfies Assumption~\ref{ass:circulation_ass}. If the system \eqref{eq:affine_sys} controlled by $\pi'$ (given by \eqref{eq:chocbf_qp}) enters any equilibrium $(x_e, u_e)$ satisfying \eqref{eq:equilibrium}, only two cases can happen: (i) $u_n (x_e) = u_e$; (ii) $d(h(x_e), 0) = 0$ and $u_n (x_e) = \zeta (x_e) - \rho c(x_e)$ with $\rho \geq 0$.
\end{proposition}
\begin{proof}
Under Assumption~\ref{ass:circulation_ass}, we have $P_{\mathcal{X}_e} (x_e) = x_e$ and $u_e = \zeta (x_e)$. First, it is easy to see that the constraint $a(x)^\top u \geq b(x)$ will be inactive at $(x_e, u_e)$. If it is active, we have $h(x_e) = 0$ (as in Proposition~\ref{prop:equilibrium_on_boundary}), and the constraint $c(x)^\top [u - \zeta(P_{\mathcal{X}_e} (x))] \geq d(h(x), \lVert x - P_{\mathcal{X}_e} (x) \rVert_2) $ gives $0 \geq d(h(x_e), \lVert x_e - P_{\mathcal{X}_e} (x_e) \rVert_2 ) = d(0,0) > 0$, leading to a contradiction. If the constraint $c(x)^\top [u - \zeta(P_{\mathcal{X}_e} (x))] \geq d(h(x), \lVert x - P_{\mathcal{X}_e} (x) \rVert_2) $ is inactive at $(x_e, u_e)$, then both constraints are inactive, and we have $u_n (x_e) = u_e$. If it is active, we have $d(h(x_e), \lVert x_e - P_{\mathcal{X}_e} (x_e) \rVert_2)  = d(h(x_e), 0)= c(x_e)^\top (u_e-\zeta (x_e)) = 0$, and $u_n (x_e) = \zeta (x_e) - \rho c(x_e)$ with $\rho \geq 0$ by analyzing the KKT conditions.
\end{proof}

Proposition~\ref{prop:chocbf_qp_equilibrium} characterizes the equilibria introduced by \eqref{eq:chocbf_qp}. To avoid the equilibria belonging to the second case (see Fig.~\ref{fig:constraint_corner}), we can fine-tune the $d$ function so that $d(h(x_e), 0) = 0$ and $u_n (x_e) = \zeta (x_e) - \rho c(x_e)$ do not hold at the same time. If $d(h(x), \lVert x - P_{\mathcal{X}_e} (x) \rVert_2) = 1 - d_1 h(x) - d_2 \lVert x - P_{\mathcal{X}_e} (x) \rVert_2$ with $d_1, d_2 > 0$, a practical approach would be to gradually increase (or decrease) $d_1$ to encourage the circulation constraint to be active earlier (or later). Although our proposed circulation mechanism cannot fully eliminate all spurious equilibria, it prevents equilibria on the boundary of the safe set and significantly enhances the applicability of CBF-QP approaches, with only a minimal increase in complexity. Unlike \cite{reis2020control, mestres2022optimization, tan2024undesired} that study equilibria arising from the interplay of CLFs and CBFs, our work considers a CBF with an arbitrary nominal controller. Because the spurious equilibria depend on both the controller (e.g., stabilizing, tracking) and the CBF, a full stability analysis of the equilibria becomes intractable. Finally, the next proposition establishes the continuity of the control $\pi'$ given by \eqref{eq:chocbf_qp}. 

\begin{proposition}\label{prop:continuity}
    Assume that \eqref{eq:chocbf_qp} is feasible for each $x \in \mathcal{X}$, $\pi' (x) \notin \partial \mathcal{U}$, and system \eqref{eq:affine_sys} satisfies Assumption~\ref{ass:circulation_ass}. If $a(x) \neq 0$, $c(x) \neq 0$, and $u_n$ is continuous in $x$, then $\pi'$ defined in \eqref{eq:chocbf_qp} is continuous in $x$. 
\end{proposition}
\begin{proof}
    Since the problem \eqref{eq:chocbf_qp} is assumed to be feasible and has a strongly convex objective function, the solution to \eqref{eq:chocbf_qp} is unique. As $a(x)$ and $c(x)$ are linearly independent, the regularity conditions of \cite[Theorem~1]{hager1979lipschitz} hold. As a result, $\pi'$ is Lipschitz continuous in the data $u_n (x), a(x), b(x), c(x)$, $d(h(x), \lVert x - P_{\mathcal{X}_e} (x) \rVert_2)$, and $\zeta (P_{\mathcal{X}_e} (x))$. As the data is continuous in $x$, $\pi'$ is therefore continuous in $x$.
\end{proof}

For simplicity, it is also desirable to use joint acceleration as the intermediate control $\Tilde{u}$, and we have $\Tilde{u}_e = 0$ in this case. Then, the actual input torque is calculated by $u = M(q) \Tilde{u} + \sigma (q, \dot{q})$ as in \eqref{eq:lagrange_mechanics}. If $K$ CBFs ($h_1, ..., h_K$) are employed, a practical way is to use the smooth minimum of the CBFs and then construct the circulation constraint based on the smooth minimum \cite{usevitch2021adversarial}. Define the smooth minimum operator
\begin{equation}\label{eq:smooth_min}
    \varphi (h_1, ..., h_K) = -\frac{1}{\eta} \ln \left ( 
    \frac{1}{K} \sum_{i=1}^K e^{-\eta h_i} \right ) 
\end{equation}
where $\eta > 0$ is a tunable parameter that controls how tightly $\varphi (h_1, ..., h_K)$ approximates $h_{\min} =   \min (h_1, ..., h_K)$. Then, we construct a new CBF $\phi (x) = \varphi (h_1 (x), ..., h_K (x)) - \varphi_0$ (with $\varphi_0 > 0$ to compensate for the approximation error) and build the circulation constraint w.r.t. $\phi$. In practice, we use the log-sum-exp trick in the numerical evaluation of $\varphi (h_1, ..., h_K)$, i.e., $\varphi (h_1, ..., h_K) =  -\frac{1}{\eta} \ln ( \frac{1}{K} \sum_{i=1}^K e^{-\eta (h_i - h_{\min})} ) + h_{\min}$, where $h_{\min}$ is subtracted from the exponents to improve numerical stability. 

Although our control design is primarily model-based, HOCBFs also render the safe set asymptotically stable \cite{tan2021high}, providing robustness against model mismatch and dynamics errors. We can also increase $\alpha_0$ to make the HOCBFs more conservative. Our theoretical guarantee relies on the feasibility of the online QP-based controller. In the case of multiple CBFs, we can use the smooth minimum as described previously to avoid competition among the CBF constraints. We can also introduce slack variables that transform hard constraints into soft constraints to mitigate potential infeasibility \cite{wang2017safety, lee2023hierarchical}. Our approach can be extended to time-varying CBFs by including time $t$ in the parameter $\theta$, provided the conditions of Theorems~\ref{thm:cd_alpha} and \ref{thm:high_order_cd_alpha} still hold. However, explicit time dependence adds terms to the CBF constraint \cite{xiao2021high}, invalidating the characterization of spurious equilibria in Proposition~\ref{prop:equilibrium_on_boundary}. Our results in Sections~\ref{sec:diff_opt} and \ref{sec:scaling_funcs} can also be used in the safe backstepping method proposed in \cite{taylor2022safe, cohen2024safety}, which is another line of research on safe control for high relative degree cases.

\section{Experiments}\label{sec:experiments}
We validate our method through three real-world experiments using the Franka Research 3 (FR3) robotic manipulator, which has seven revolute joints and a gripper. The computational engine is a PC with 64GB RAM and an Intel Core i9-13900K processor. We use OSQP \cite{stellato2020osqp} for QPs, SCS \cite{odonoghue2016conic} for exponential cone problems, and Pinocchio \cite{carpentier2019pinocchio} for frame pose calculations. In the first two experiments (Sections~\ref{sec:pick_and_place} and \ref{sec:whiteboard_cleaning}), the environment is assumed known. In the third (Section~\ref{sec:exp_flying_ball}), we use a Vicon motion capture system to track a moving obstacle. In the first two experiments, we demonstrate the effectiveness of the circulation mechanism and show that the lack of the circulation inequality leads to the robot getting stuck at a spurious equilibrium. The video of the experiments can be found at 
\url{https://youtu.be/uZmM3_wBjGY}.

\subsection{Pick and Place}
\label{sec:pick_and_place}
\begin{figure}[t]
    \centering
    \begin{subfigure}[b]{0.2\textwidth}
     \centering
     \includegraphics[width=\textwidth]{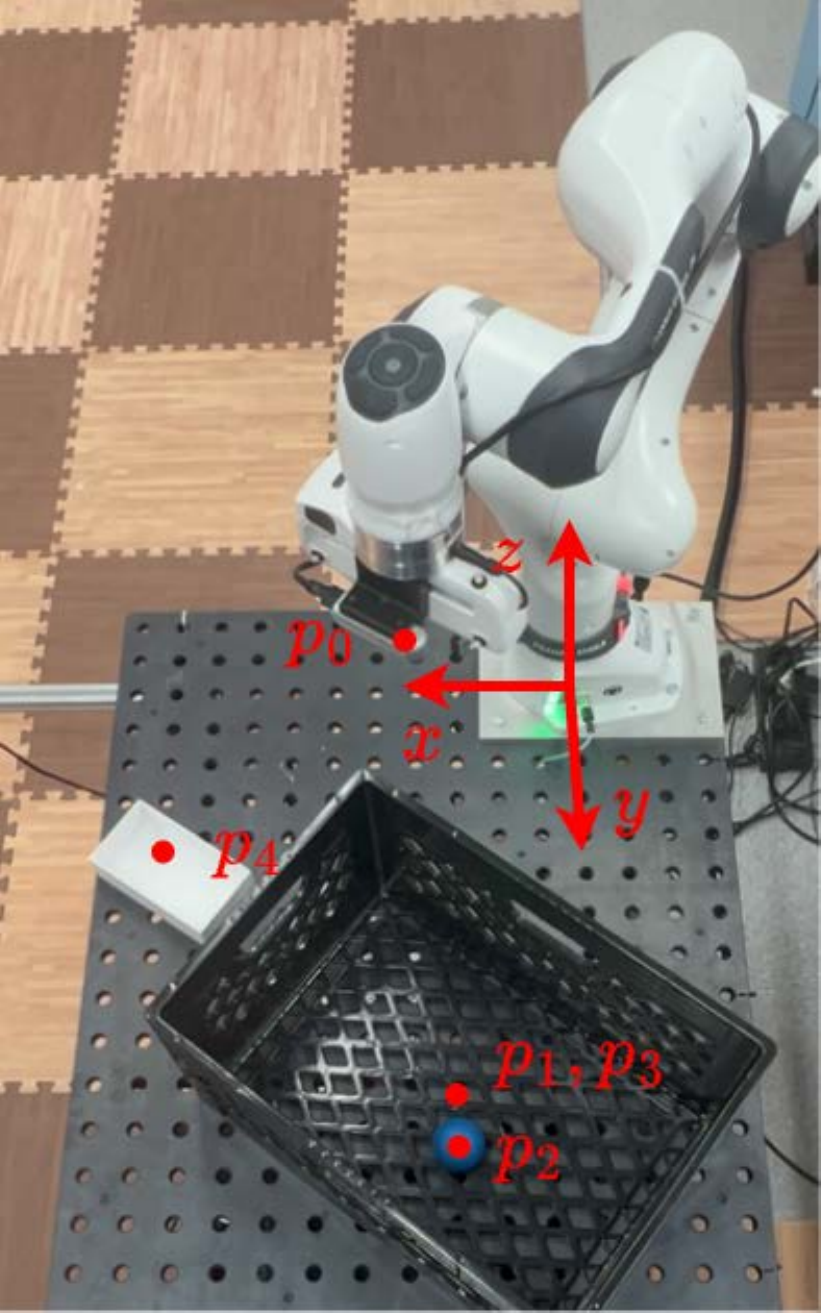}
     \caption{Overall setting.}
     \label{fig:exp1_task}
 \end{subfigure}
 \hspace{10pt}
 \begin{subfigure}[b]{0.2\textwidth}
     \centering
     \includegraphics[width=\textwidth]{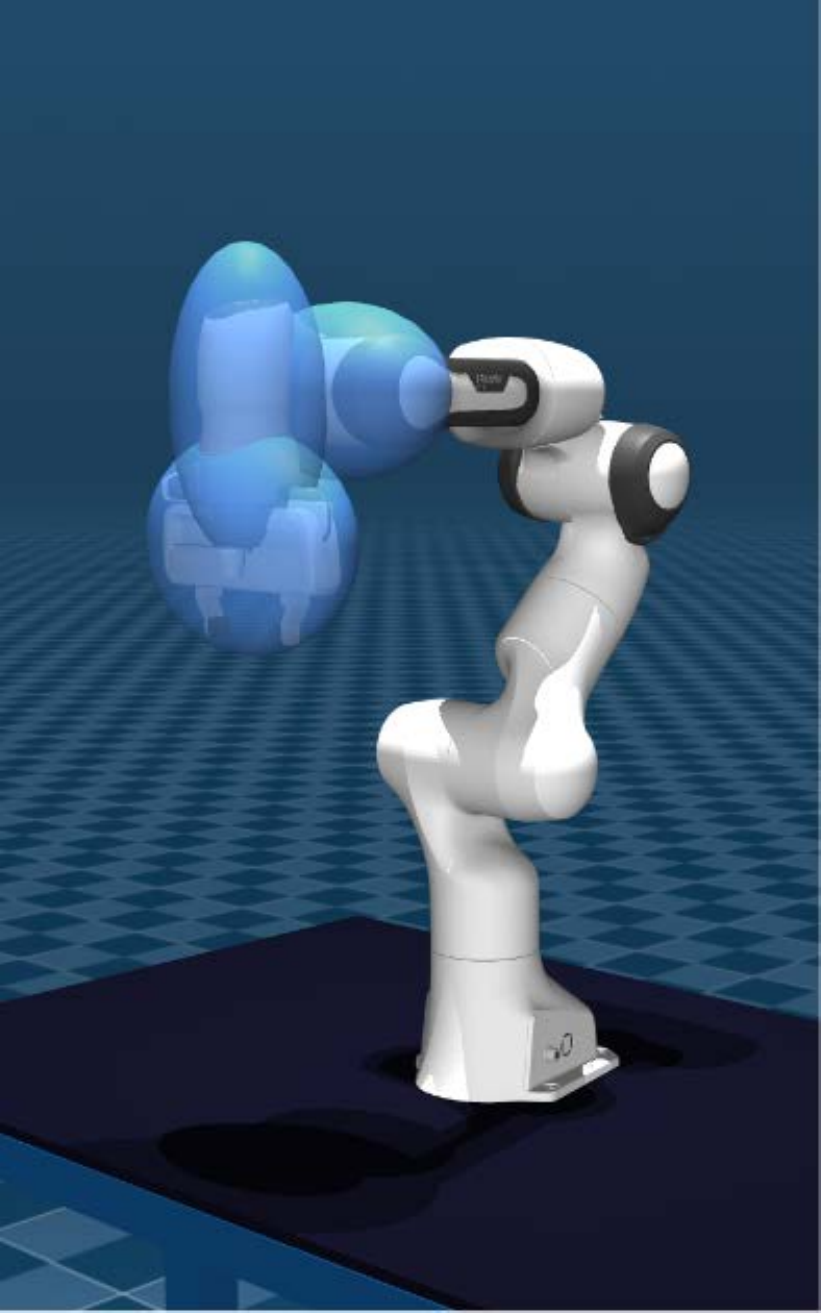}
     \caption{Bounding ellipsoids.}
     \label{fig:exp1_bounding_boxes}
 \end{subfigure}
 \caption{\subref{fig:exp1_task}: Overall setting of the pick-and-place task. Coordinates in \subref{fig:exp1_task}: $p_0 = (0.16, 0.26, 0.49)$, $p_1 = p_3 = (0.15, 0.62, 0.10)$, $p_2 = (0.15, 0.62, 0.02)$, $p_4 = (0.50, 0.25, 0.08)$. \subref{fig:exp1_bounding_boxes}: Bounding ellipsoids of the 6th and 7th links and the end-effector of the robotic manipulator. The simulation figure is generated using MuJoCo.}
 \label{fig:exp1}
\end{figure}

\begin{figure*}[t]
    \centering
    \begin{subfigure}[b]{0.16\textwidth}
        \centering
        \includegraphics[width=\textwidth]{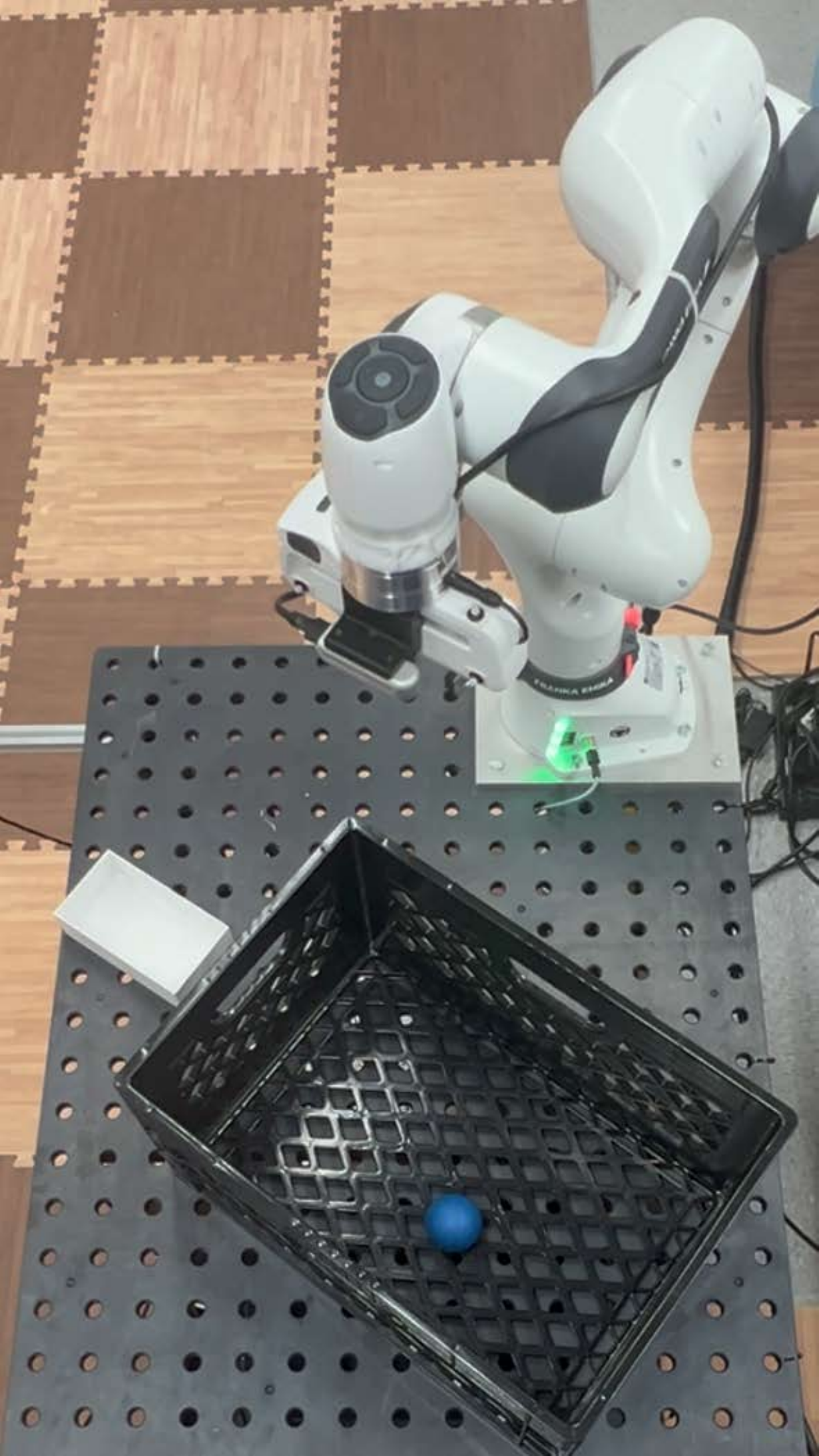}
        \caption{$t=0$~\si{s}.}
        \label{fig:exp1_with_circ_t_0}
    \end{subfigure}
    \begin{subfigure}[b]{0.16\textwidth}
        \centering
        \includegraphics[width=\textwidth]{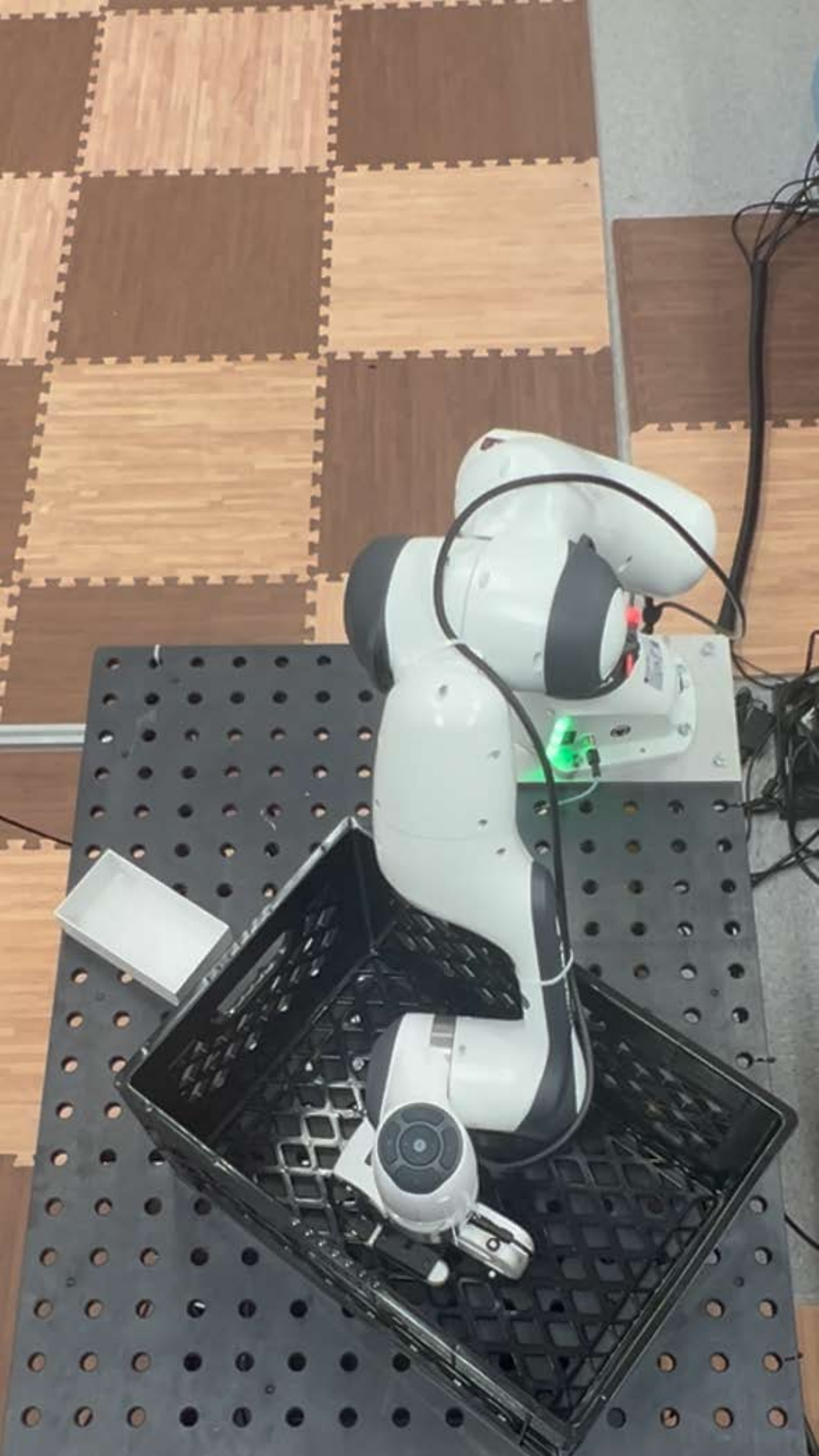}
        \caption{$t=12$~\si{s} (w/o circ.).}
        \label{fig:exp1_no_circ_t_12}
    \end{subfigure}
    \begin{subfigure}[b]{0.16\textwidth}
        \centering
        \includegraphics[width=\textwidth]{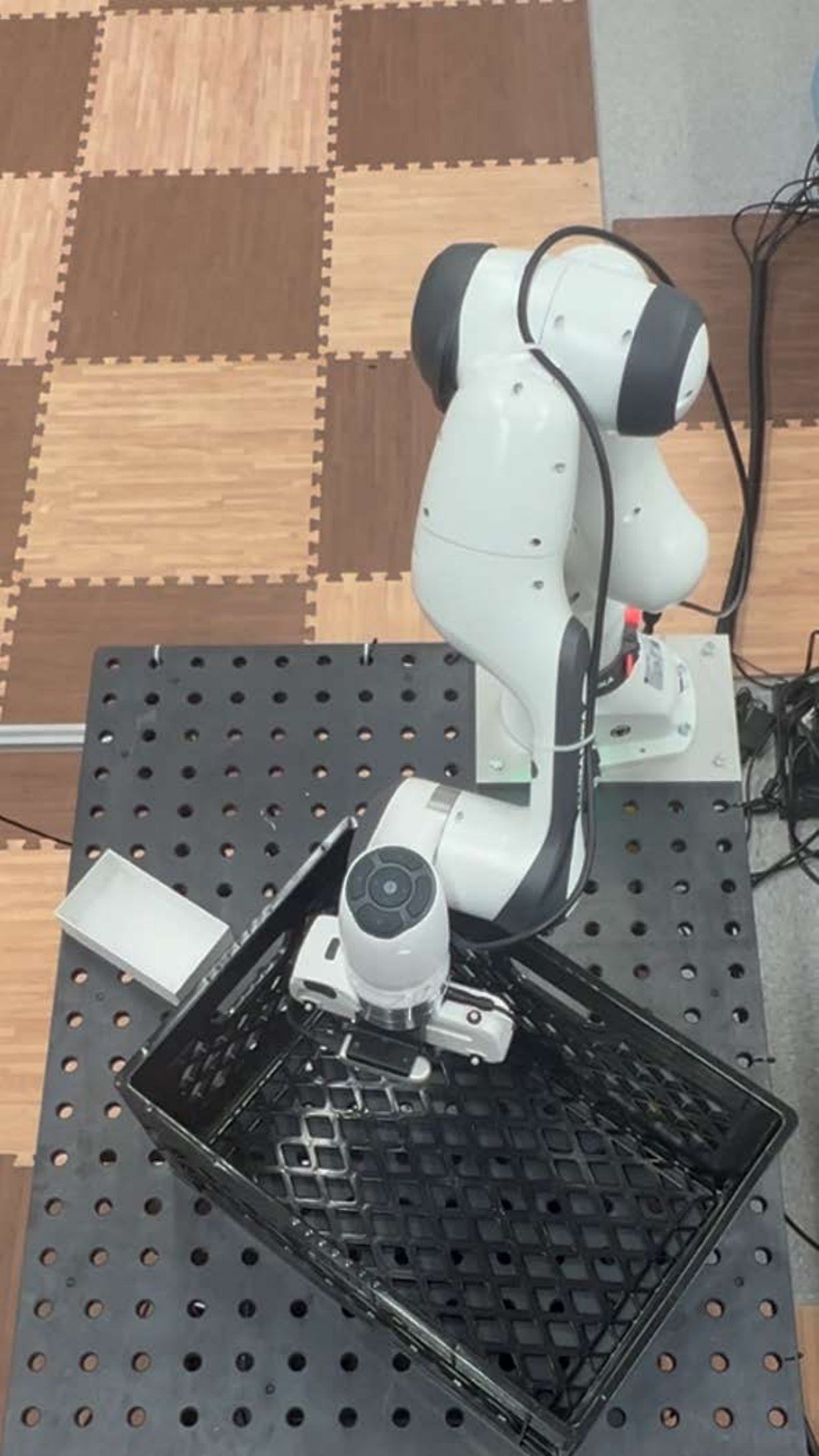}
        \caption{$t=24$~\si{s} (w/o circ.).}
        \label{fig:exp1_no_circ_t_24}
    \end{subfigure}
    \begin{subfigure}[b]{0.16\textwidth}
        \centering
        \includegraphics[width=\textwidth]{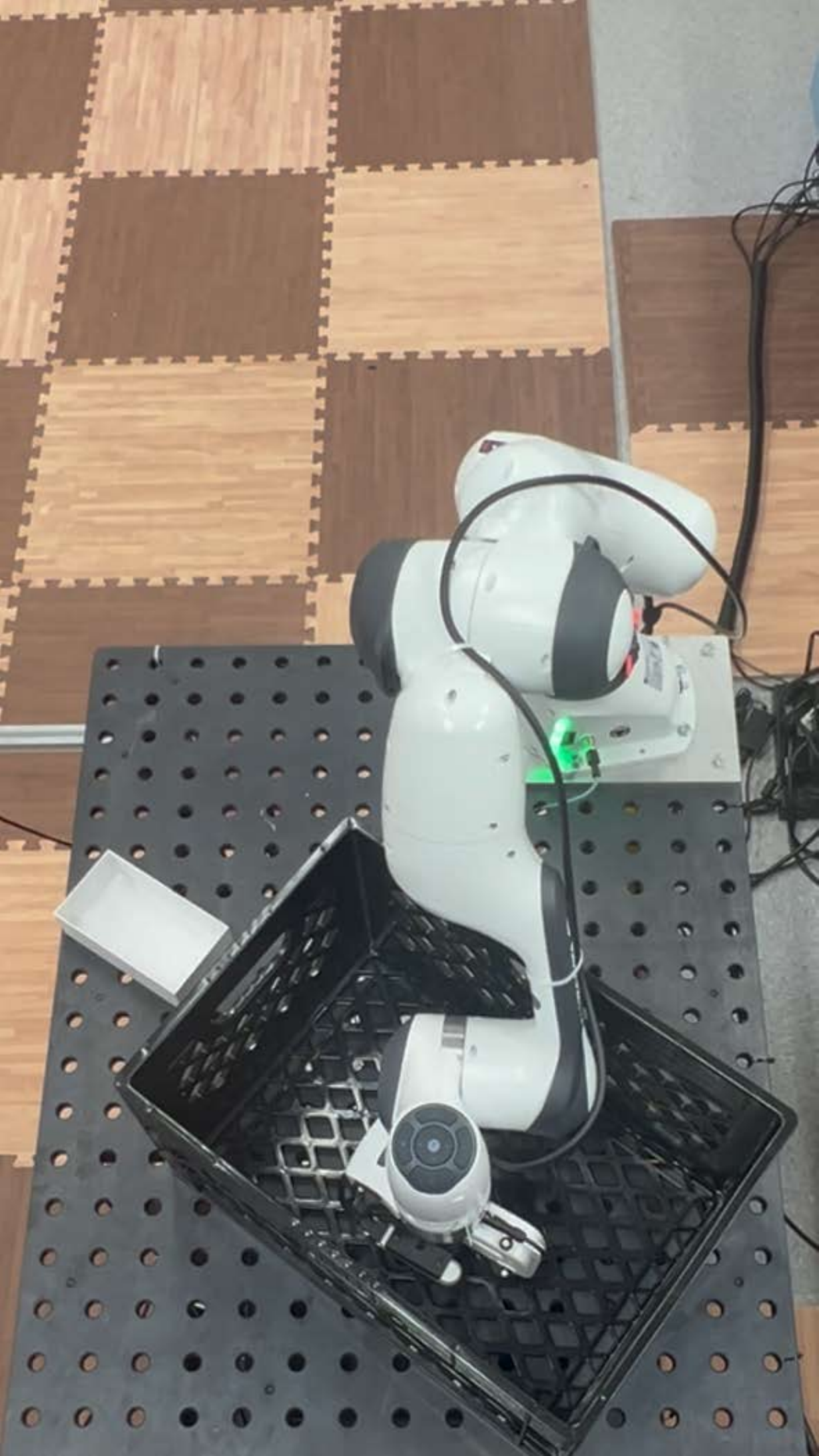}
        \caption{$t=12$~\si{s} (w/ circ.).}
        \label{fig:exp1_with_circ_t_12}
    \end{subfigure}
    \begin{subfigure}[b]{0.16\textwidth}
        \centering
        \includegraphics[width=\textwidth]{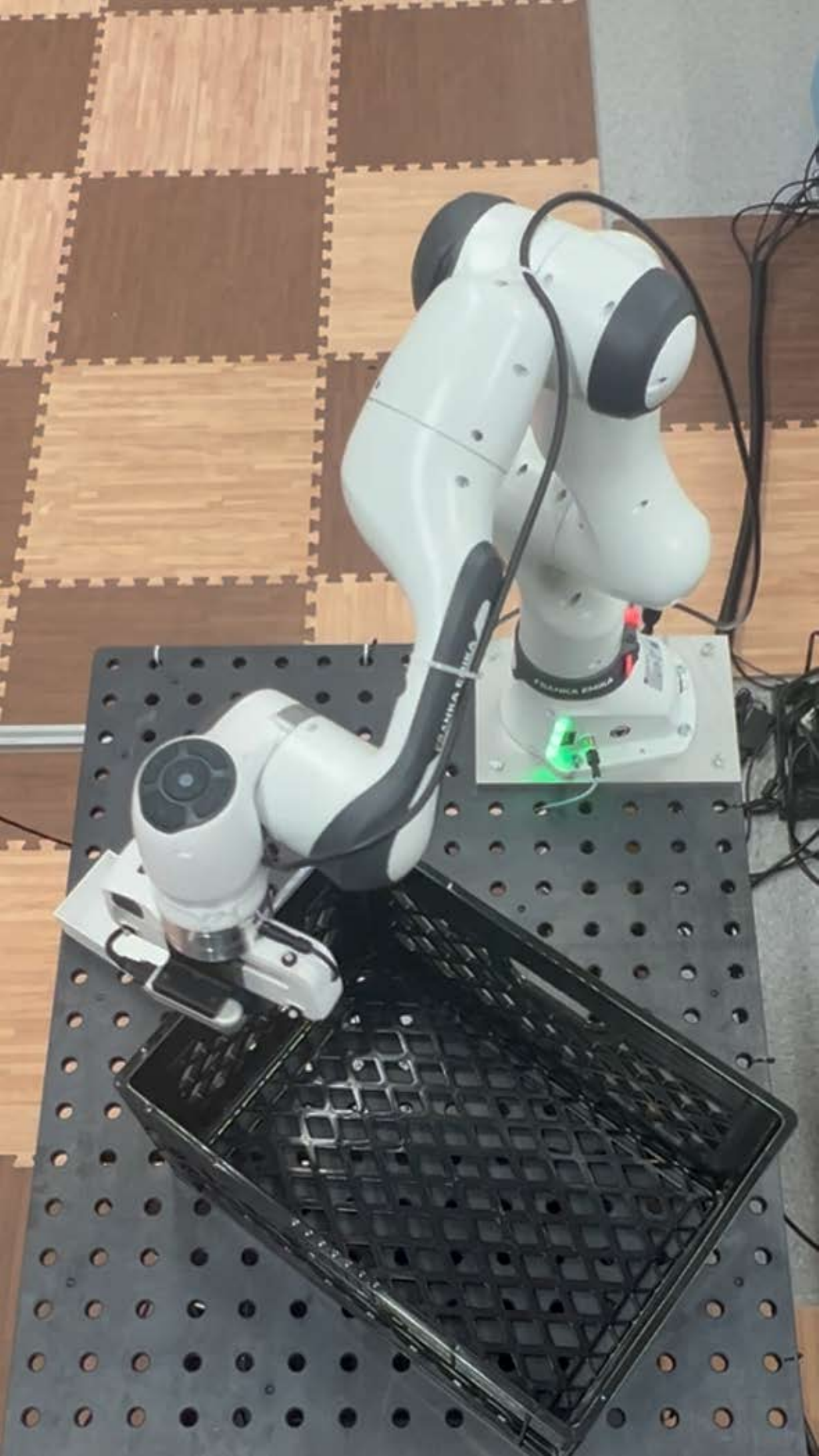}
        \caption{$t=23$~\si{s} (w/ circ.).}
        \label{fig:exp1_with_circ_t_23}
    \end{subfigure}
    \begin{subfigure}[b]{0.16\textwidth}
        \centering
        \includegraphics[width=\textwidth]{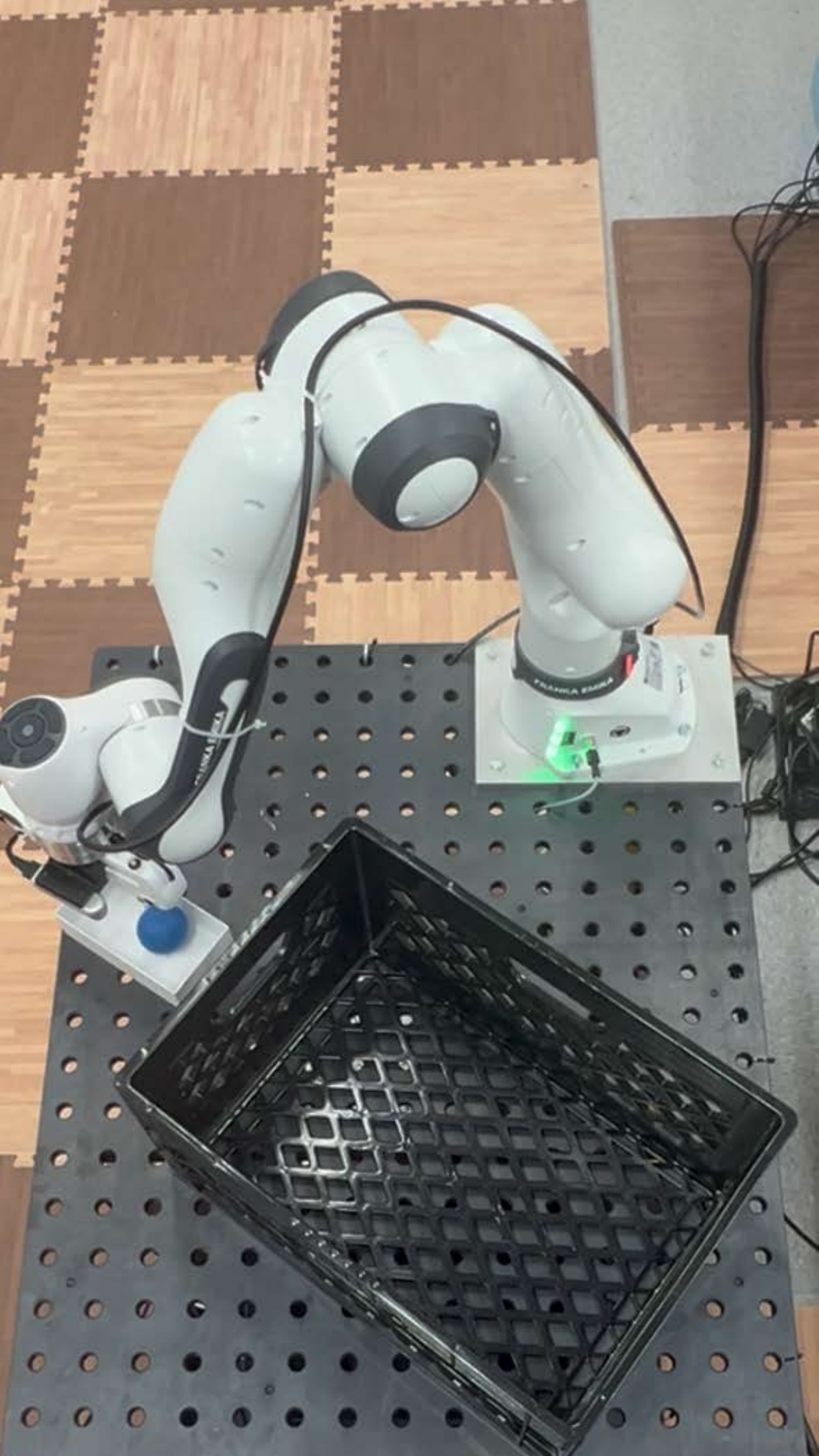}
        \caption{$t=28$~\si{s} (w/ circ.).}
        \label{fig:exp1_with_circ_t_28}
    \end{subfigure}
\caption{Snapshots of the experiment. \subref{fig:exp1_with_circ_t_0}: Initial pose. \subref{fig:exp1_no_circ_t_12} and \subref{fig:exp1_no_circ_t_24}: Poses without the circulation constraint. The manipulator enters an equilibrium from $t=24$~\si{s} (see \subref{fig:exp1_no_circ_t_24}). \subref{fig:exp1_with_circ_t_12}-\subref{fig:exp1_with_circ_t_28}: Poses with the circulation constraint. The task is successful.}
\label{fig:exp1_clips}
\end{figure*}

\begin{figure}[t]
    \centering
    \begin{subfigure}[b]{0.24\textwidth}
     \centering
     \includegraphics[width=\textwidth]{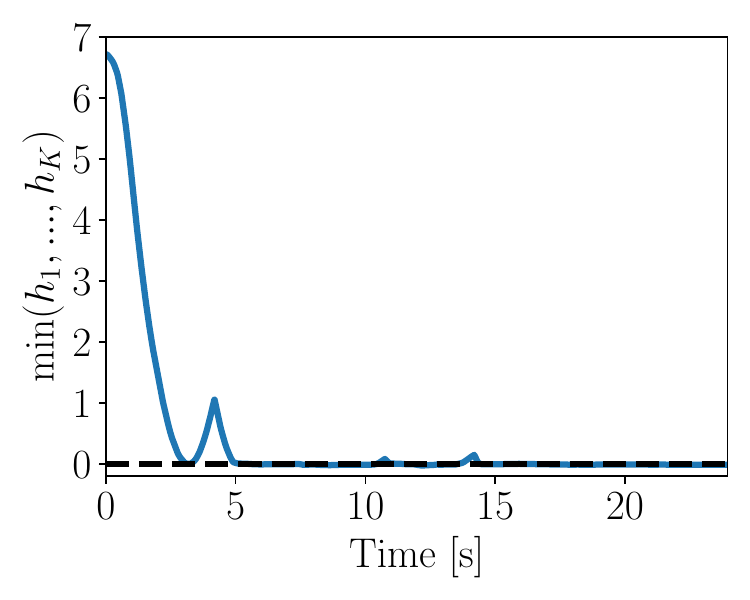}
     \caption{Without circulation.}
     \label{fig:exp1_cbf_no_circ}
 \end{subfigure}
 \begin{subfigure}[b]{0.24\textwidth}
     \centering
     \includegraphics[width=\textwidth]{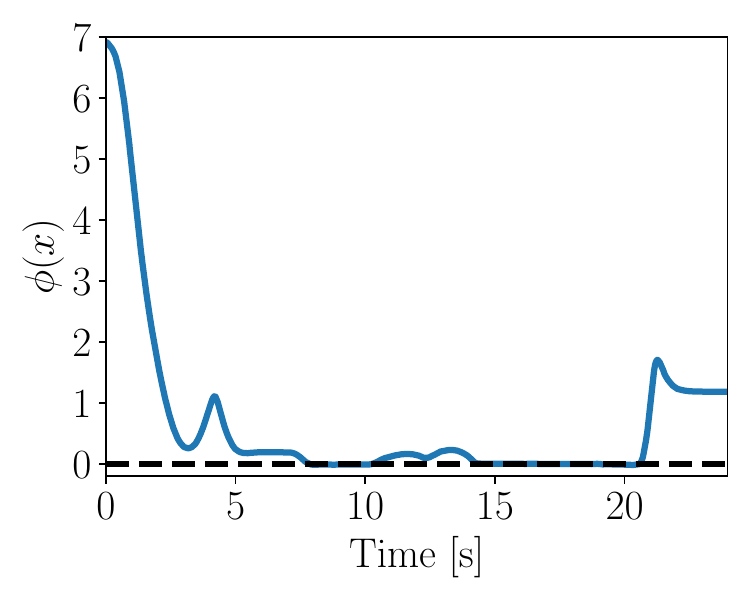}
     \caption{With circulation.}
     \label{fig:exp1_cbf_with_circ}
 \end{subfigure}
 \caption{Evolution of the CBFs.}
 \label{fig:exp1_cbf_values}
\end{figure}

We consider a pick-and-place task using the FR3 robotic manipulator. As shown in Fig.~\ref{fig:exp1_task}, the robotic manipulator needs to pick up the blue ball inside the black utility box and place it inside the white box. The nominal controller commands the gripper to follow line trajectories that sequentially connect the positions $p_0$, $p_1$, $p_2$, $p_3$, and $p_4$ with a constant orientation. The gripper closes at $p_2$ to pick up the ball and releases the ball at $p_4$. To avoid collision with the box and the table surface, we bound the end-effector and the 6th and 7th links of the manipulator using ellipsoids (see Fig.~\ref{fig:exp1_bounding_boxes}). We model the table surface (and the space under it) with a half-space and the four vertical sides of the box with rectangular cuboids (therefore convex polytopes). There are, in total, five obstacles and three moving parts of the robot under consideration, leading to $K = 15$ CBFs as in \eqref{eq:cbf}. More details on this example are given in Appendix~\ref{sec:pick_and_place_details}.

For simplicity, we write the control design in terms of joint accelerations $\Tilde{u}$ and then map the control commands to the actual torque input $u$ using \eqref{eq:lagrange_mechanics}. The system with joint accelerations as control input is
\begin{equation}\label{eq:double_integrator}
    \frac{d}{dt} \underbrace{\begin{bmatrix}
        q \\ \dot{q}
    \end{bmatrix}}_{x} = \underbrace{\begin{bmatrix}
        \dot{q} \\ 0
    \end{bmatrix}}_{\Tilde{f}(x)} + \underbrace{\begin{bmatrix}
        0 \\ I
    \end{bmatrix}}_{\Tilde{g}(x)} \Tilde{u}.
\end{equation}
Let $q_l, q_u \in \R^7$ be the lower and upper bounds of the joint angles. Then, for \eqref{eq:double_integrator}, $\mathcal{X}_e = \{[q^\top, 0]^\top \in \R^{14} \mid q_l \leq q \leq q_u \}$ and $\zeta (x_e) = 0$.

Let $p_{\text{tcp}} (t) \in \R^3$ and $R_{\text{tcp}} (t) \in \R^{3 \times 3}$ be the position and orientation of the frame centered at the tool center point (TCP), which is the middle point between the two fingertips of the gripper. Given the desired time-varying position $p_d (t)$ and constant orientation $R_d$, we define the tracking control
\begin{equation*}
    \Tilde{u}_{\text{task}} = J_{\text{tcp}}^\dagger \left( \begin{bmatrix}
        \ddot{p}_d - K_{p}^p (p_{\text{tcp}} - p_d) - K_{p}^d (v_{\text{tcp}} - \dot{p}_d)\\
        - K_{r}^p \Log_{\text{SO}(3)} (R_{\text{tcp}} R_d^\top) - K_{r}^d \omega_{\text{tcp}} 
    \end{bmatrix} - \dot{J}_{\text{tcp}} \dot{q} \right)
\end{equation*}
where $J_{\text{tcp}} \in \R^{6 \times 7}$ is the geometric Jacobian of the TCP frame, $v_{\text{tcp}}$ and $\omega_{\text{tcp}}$ are the linear and angular velocities, and $K_{p}^p$, $K_{p}^d$, $K_{r}^p$, $K_{r}^d > 0$  are the 3-by-3 gain matrices. Let $\bar{q} = (q_l + q_u)/2$, and define 
\begin{equation}\label{eq:joint_centering}
    \Tilde{u}_{\text{joint}} = - K_q^p (q - \bar{q}) - K_q^d \dot{q}
\end{equation}
to encourage the joints to stay in the middle of the joint limits, where $K_q^p, K_q^d > 0$ are the 7-by-7 gain matrices. Then, the nominal control is defined as 
\begin{equation}
    \Tilde{u}_n = \Tilde{u}_{\text{task}} + (I - J_{\text{tcp}}^\dagger J_{\text{tcp}}) \Tilde{u}_{\text{joint}}.
\end{equation}

We take $\Gamma_1 (r) = \gamma_1 r$ and $\Gamma_2 (r) = \gamma_2 r$ as the class $\mathcal{K}$ functions in \eqref{eq:intermediate_barrier_funcs}. For each CBF $h_i$ ($i = 1, ..., K$), define $\psi_{i,0} (x) = h_i (x)$ and $\psi_{i,1} (x) = \dot{\psi}_{i,0} (x)+ \gamma_1 \psi_{i,0} (x)$. In addition to the CBFs for collision avoidance, we define six extra CBFs to regulate the linear velocity of the TCP. More specifically, let $v_l = [v_{l,1}, v_{l,2}, v_{l,3}]^\top$ and $v_u = [v_{u,1}, v_{u,2}, v_{u,3}]^\top$ be the lower and upper bounds of the TCP's linear velocity. Define $h_{l,i}(x) = v_{\text{tcp},i} - v_{l,i}$ and $h_{u,i}(x) = v_{u,i}-v_{\text{tcp},i}$ with $i = 1,2,3$. Then, the HOCBF-QP without circulation is defined as 
\begin{equation} \label{eq:exp1_hocbf_qp}
    \begin{aligned}
    \Tilde{\pi} (t,x) &= \argmin_{\Tilde{u} \in \Tilde{\mathcal{U}}} \quad  \lVert \Tilde{u}- \Tilde{u}_n(t,x) \rVert_2^2 \\
    \textrm{s.t.} \quad & \dot{\psi}_{i,1} (x) \geq -\gamma_2 \psi_{i,1} (x), \ i=1,...,K,  \\
    & \dot{h}_{l,i} (x) \geq -\gamma_l h_{l,i} (x), \, i=1,...,3, \\
    & \dot{h}_{u,i} (x) \geq -\gamma_u h_{u,i} (x), \, i=1,...,3,
    \end{aligned}
\end{equation}
where $\gamma_l, \gamma_u > 0$. To enable the circulation constraint, we define a new CBF based on the smooth minimum (see \eqref{eq:smooth_min}) of $h_i$: $\phi(x) = \varphi (h_1 (x), ..., h_K (x)) - \varphi_0$. Define in addition $\psi_{\phi, 0} (x) = \phi(x)$ and $\psi_{\phi,1} (x) = \dot{\psi}_{\phi,0} (x)+ \gamma_1 \psi_{\phi,0} (x)$. The CHOCBF-QP is defined as 
\begin{equation} \label{eq:exp1_chocbf_qp}
    \begin{aligned}
    \Tilde{\pi}' (t,x) & = \argmin_{\Tilde{u} \in \Tilde{\mathcal{U}}} \quad \lVert \Tilde{u}- \Tilde{u}_n(t,x) \rVert_2^2 \\
    \textrm{s.t.} \quad & \dot{\psi}_{\phi,1} (x) \geq -\gamma_2 \psi_{\phi,1} (x),  \\
    & c(x)^\top [\Tilde{u} - \zeta (P_{\mathcal{X}_e} (x))] \geq d(\phi(x), \lVert x - P_{\mathcal{X}_e} (x) \rVert_2), \\
    & \dot{h}_{l,i} (x) \geq -\gamma_l h_{l,i} (x), \, i=1,...,3, \\
    & \dot{h}_{u,i} (x) \geq -\gamma_u h_{u,i} (x), \, i=1,...,3,
    \end{aligned}
\end{equation}
where the term for the circulation inequality $c(x)$ is a vector orthogonal to $a(x) = (L_{\Tilde{g}} L_{\Tilde{f}} \phi (x))^\top$ (the coefficient before $\Tilde{u}$ in $\dot{\psi}_{\phi,1} (x)$), and 
\begin{align*}
    d(\phi(x), \lVert x - P_{\mathcal{X}_e} (x) \rVert_2) &= d_1 (1-e^{d_2 (\phi(x) - d_3)}) \\
    & + d_4 (e^{-(\lVert \dot{q} \rVert_2 / d_5)^2}-1)
\end{align*}
 with $d_1, ..., d_5 > 0$. In this experiment, we picked $c(x) = [0, -a_3(x), a_2(x), -a_5(x), a_4(x), -a_7(x), a_6(x)]^\top$ where $a_i (x)$ is the $i$-th component of $a(x)$.

Under the controller $\Tilde{\pi}$ from \eqref{eq:exp1_hocbf_qp} without the circulation constraint, the end-effector enters the utility box and picks up the blue ball (see Fig.~\ref{fig:exp1_no_circ_t_12}), but gets stuck at an equilibrium from $t=24$~\si{s} (see Fig.~\ref{fig:exp1_no_circ_t_24}). This is because the nominal controller $\Tilde{u}_n$ tracks a straight line trajectory from $p_3$ to $p_4$ (see Fig.~\ref{fig:exp1_task}), but the CBFs are preventing the robot from colliding with the utility box, creating an equilibrium. This demonstrates that the current levels of noise during the manipulator's normal operation are insufficient to help it escape from the equilibrium. We depict the evolution of $h_{\min} = \min (h_1, ..., h_K)$ in Fig.~\ref{fig:exp1_cbf_no_circ}, and $h_{\min} (x_e)$ is near zero as expected in Proposition~\ref{prop:equilibrium_on_boundary}.

On the other hand, under the controller $\Tilde{\pi}'$ from \eqref{eq:exp1_chocbf_qp} with the circulation constraint, the manipulator successfully completes the task. After picking up the ball (Fig.~\ref{fig:exp1_with_circ_t_12}), the gripper gets out of the box (Fig.~\ref{fig:exp1_with_circ_t_23}), and then places the ball in the white box (Fig.~\ref{fig:exp1_with_circ_t_28}). The values of the parameters used in this experiment are: $\alpha_0 = 1.03$, $\gamma_1 = \gamma_2 = \gamma_l = \gamma_u = 20$, $\eta= 5$, $\varphi_0 = 0.3$, $d_1 = 600$, $d_2 = 2$, $d_3 = 0.02$, $d_4 = 0.1$, and $d_5 = 10$. In this example, 90\% of the control loops are solved in 3.0~\si{ms} or less (the 90th-percentile (p90) latency).

\subsection{Whiteboard Cleaning}
\label{sec:whiteboard_cleaning}
\begin{figure*}[t]
    \centering
    \begin{subfigure}[b]{0.16\textwidth}
        \centering
        \includegraphics[width=\textwidth]{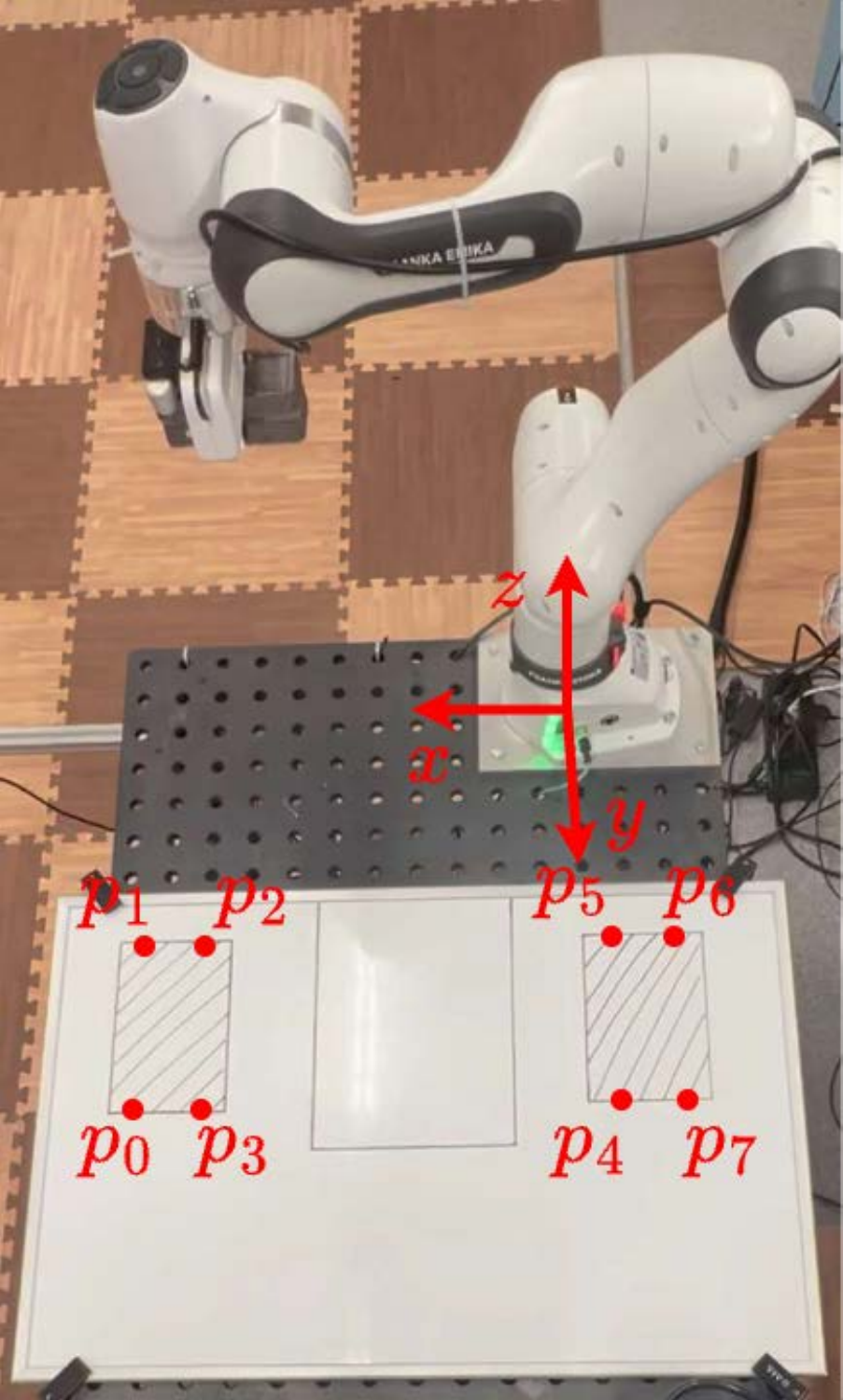}
        \caption{Overall setting.}
        \label{fig:exp2_coordinates}
    \end{subfigure}
    \begin{subfigure}[b]{0.16\textwidth}
        \centering
        \includegraphics[width=\textwidth]{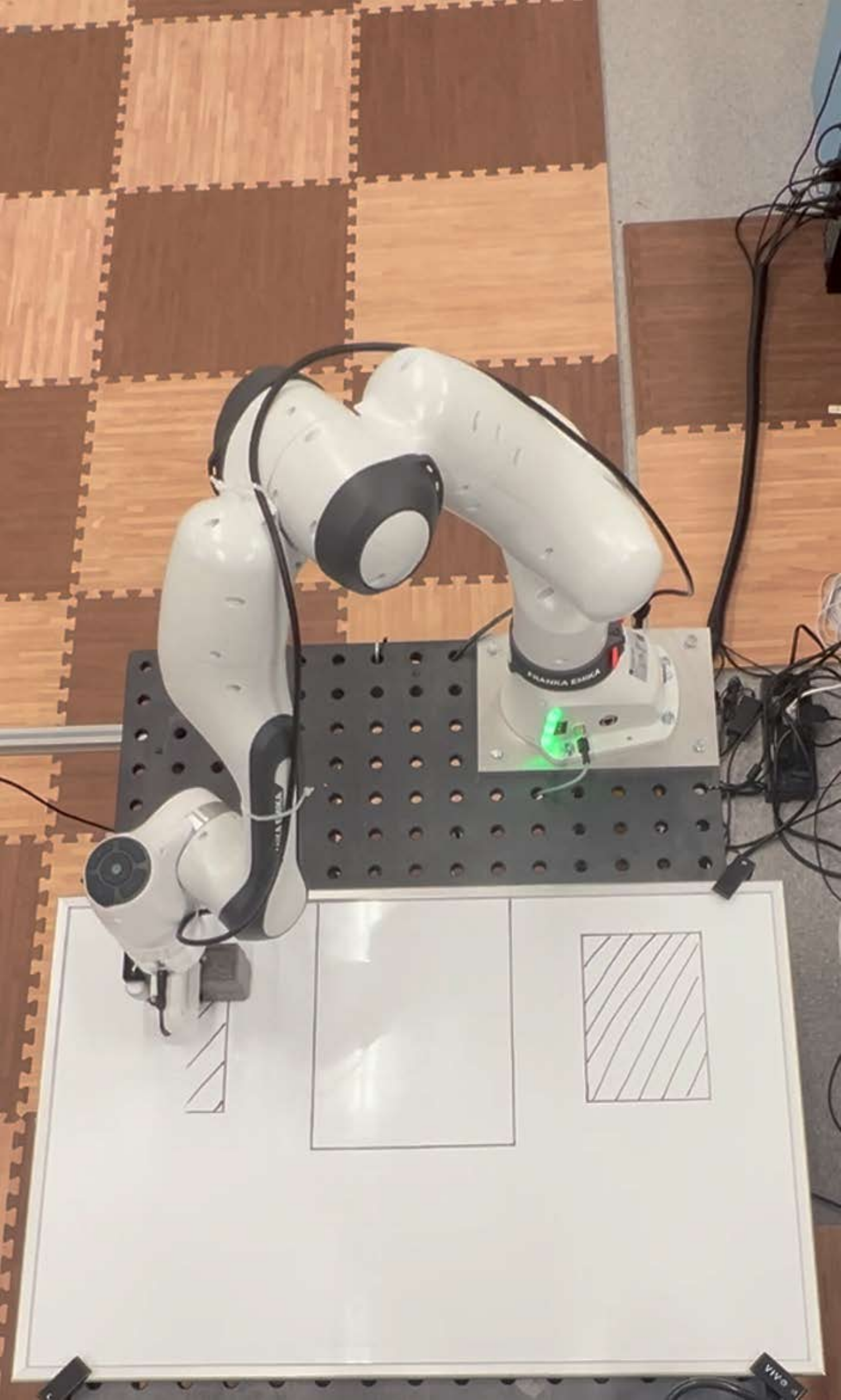}
        \caption{$t=4$~\si{s} (w/o circ.).}
        \label{fig:exp2_no_circ_t_4}
    \end{subfigure}
    \begin{subfigure}[b]{0.16\textwidth}
        \centering
        \includegraphics[width=\textwidth]{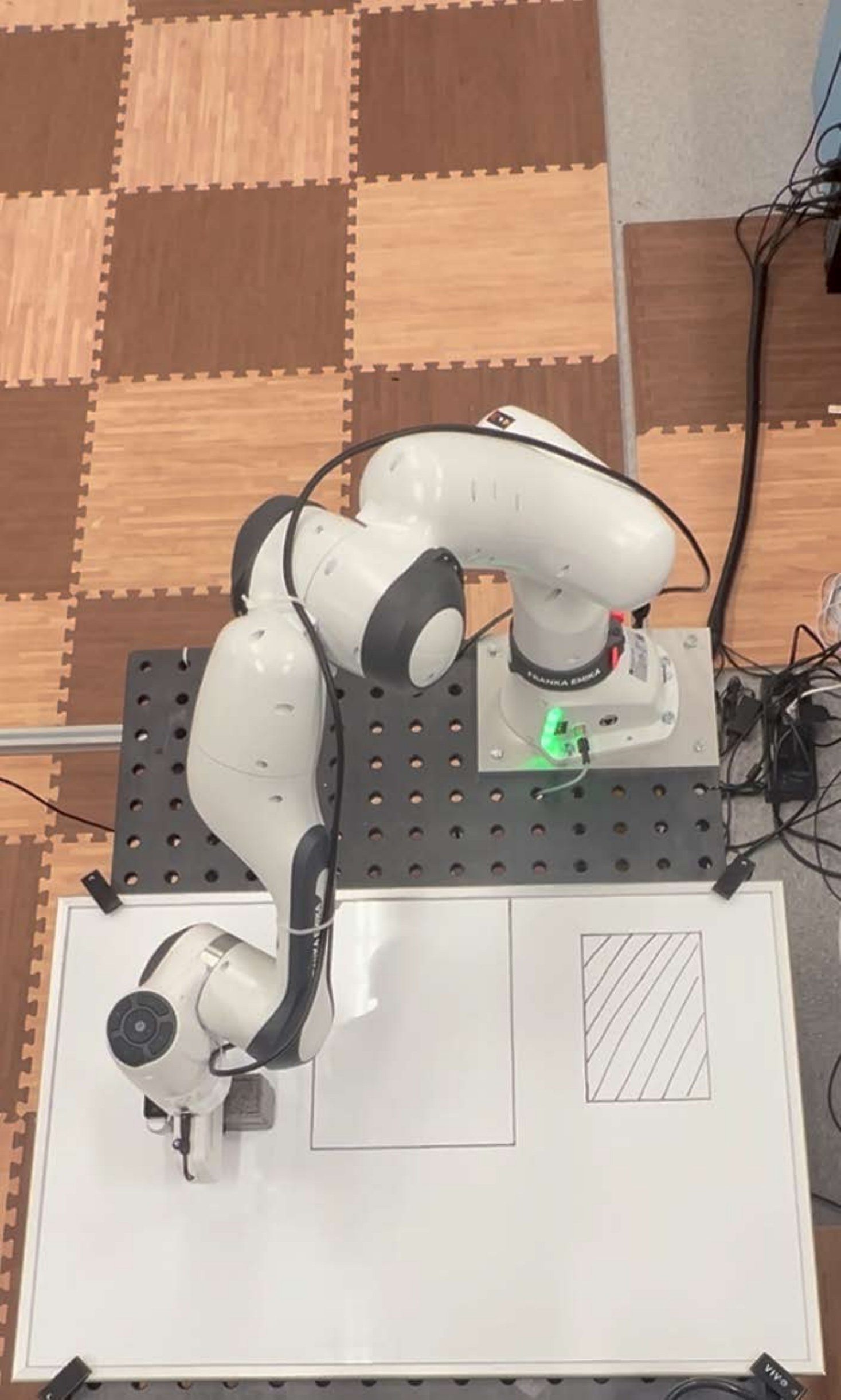}
        \caption{$t=10$~\si{s} (w/o circ.).}
        \label{fig:exp2_no_circ_t_10}
    \end{subfigure}
    \begin{subfigure}[b]{0.16\textwidth}
        \centering
        \includegraphics[width=\textwidth]{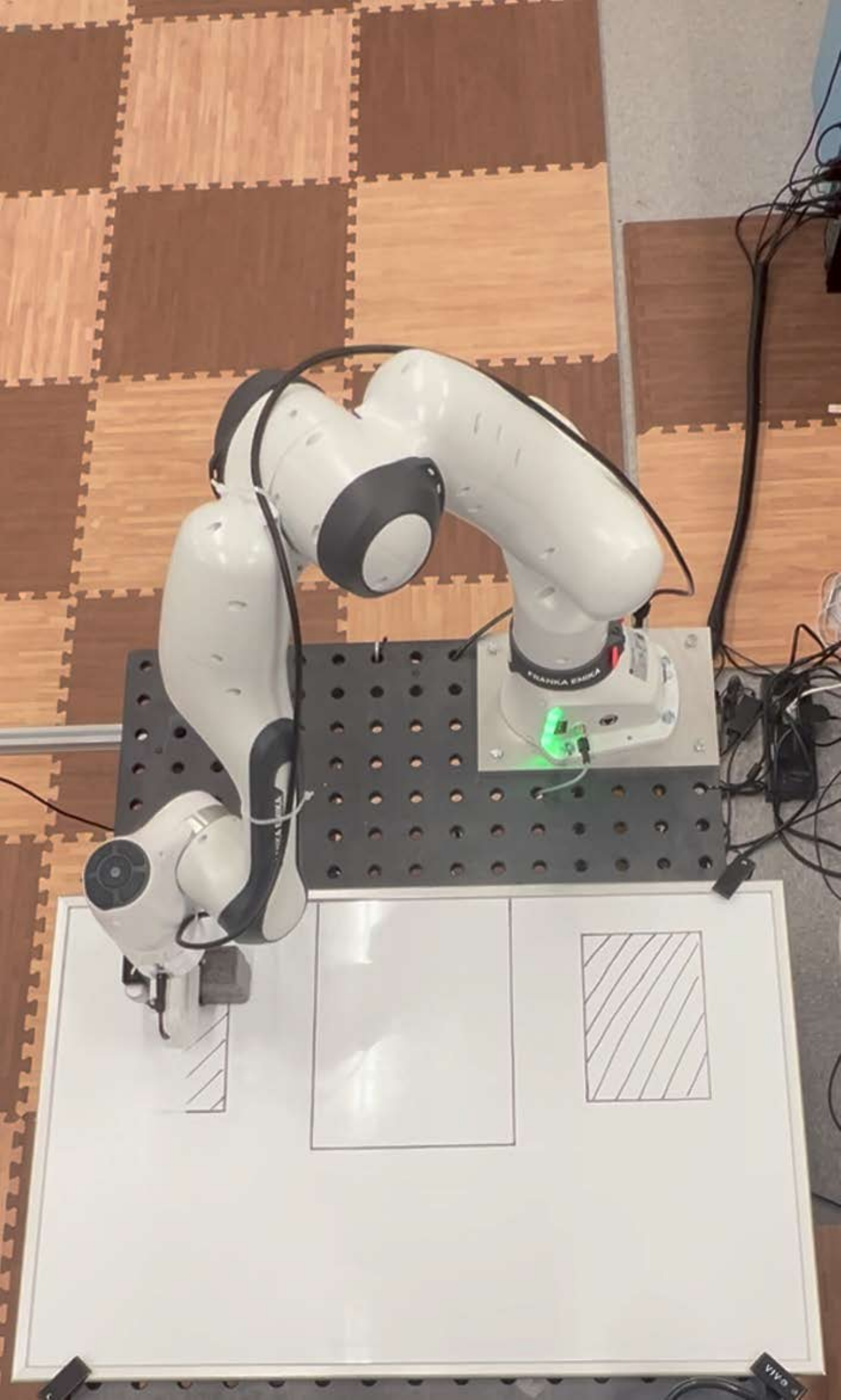}
        \caption{$t=4$~\si{s} (w/ circ.).}
        \label{fig:exp2_with_circ_t_4}
    \end{subfigure}
    \begin{subfigure}[b]{0.16\textwidth}
        \centering
        \includegraphics[width=\textwidth]{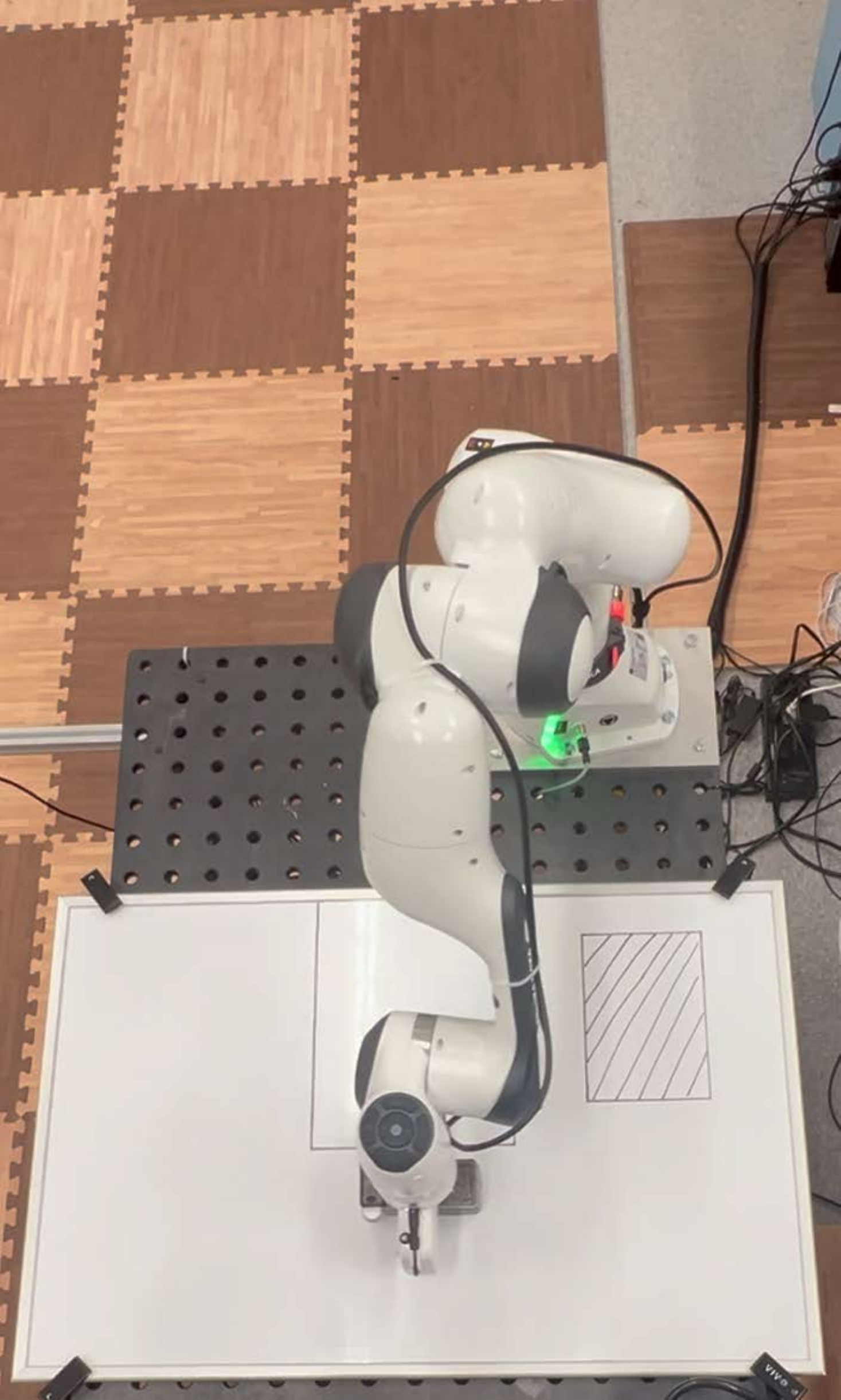}
        \caption{$t=12$~\si{s} (w/ circ.).}
        \label{fig:exp2_with_circ_t_12}
    \end{subfigure}
    \begin{subfigure}[b]{0.16\textwidth}
        \centering
        \includegraphics[width=\textwidth]{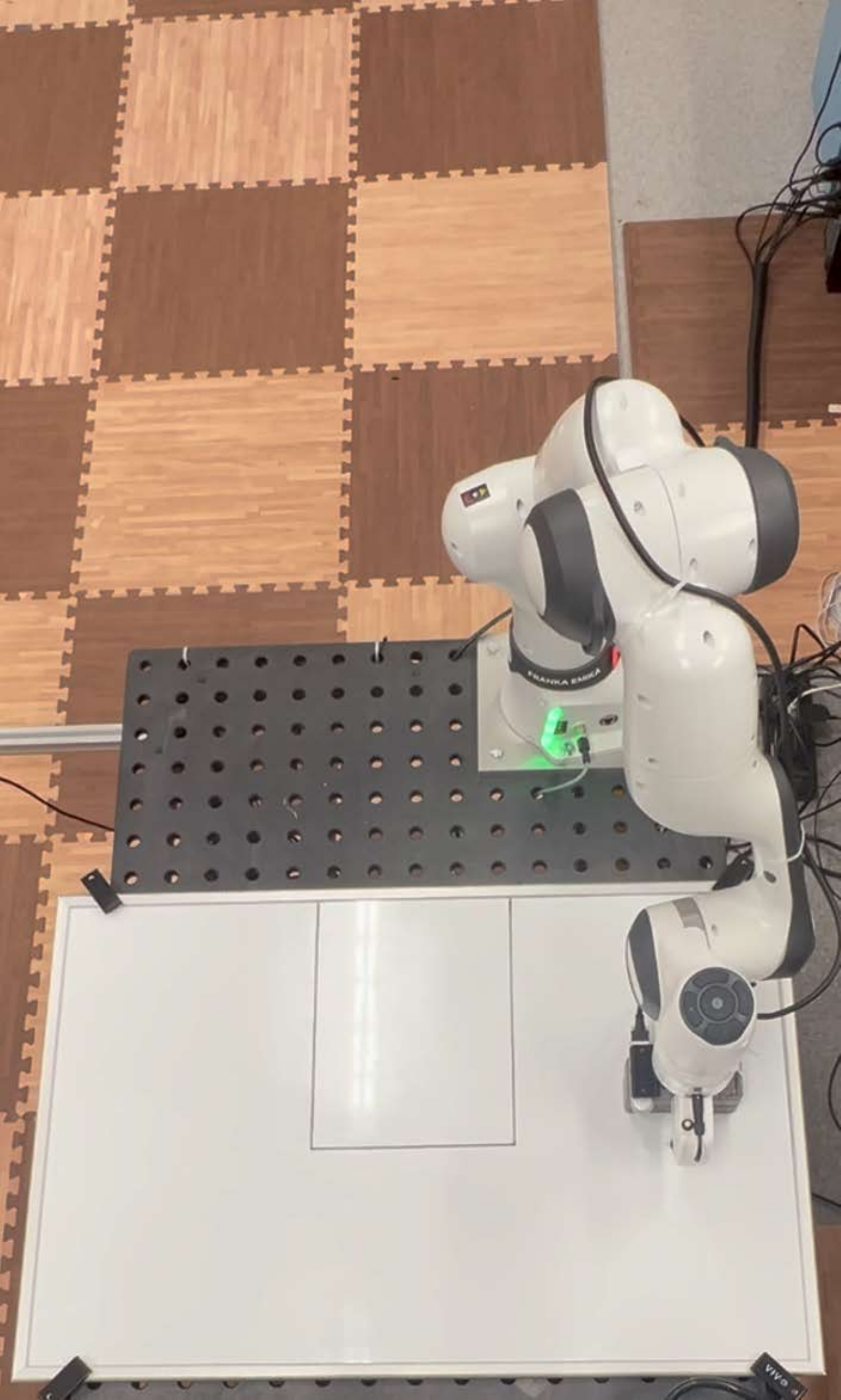}
        \caption{$t=19$~\si{s} (w/ circ.).}
        \label{fig:exp2_with_circ_t_19}
    \end{subfigure}
\caption{Snapshots of the experiment. \subref{fig:exp2_coordinates}: Overall setting of the whiteboard cleaning task. $p_0=(0.52,0.53,0.02)$, $p_1=(0.52,0.32,0.02)$, $p_2=(0.45,0.32,0.02)$, $p_3=(0.45,0.53,0.02)$, $p_4=(-0.05,0.53,0.02)$, $p_5=(-0.05,0.32,0.02)$, $p_6=(-0.12,0.32,0.02)$, $p_7=(-0.12,0.53,0.02)$. \subref{fig:exp2_no_circ_t_4} and \subref{fig:exp2_no_circ_t_10}: Poses without the circulation constraint. The manipulator enters an equilibrium from $t=10$~\si{s} (see \subref{fig:exp2_no_circ_t_10}). \subref{fig:exp2_with_circ_t_4}-\subref{fig:exp2_with_circ_t_19}: Poses with the circulation constraint. The task is successful.}
\label{fig:exp2_clips}
\end{figure*}

\begin{figure}[t]
    \centering
    \begin{subfigure}[b]{0.24\textwidth}
     \centering
     \includegraphics[width=\textwidth]{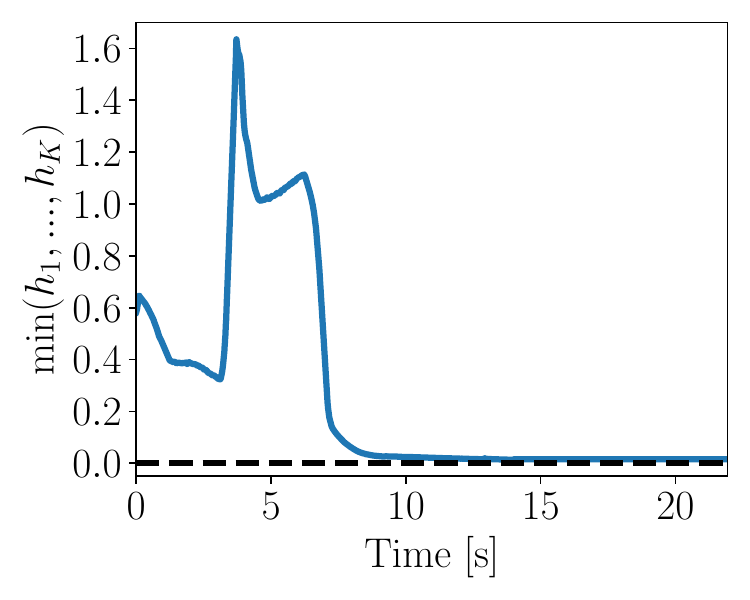}
     \caption{Without circulation.}
     \label{fig:exp2_cbf_no_circ}
 \end{subfigure}
 \begin{subfigure}[b]{0.24\textwidth}
     \centering
     \includegraphics[width=\textwidth]{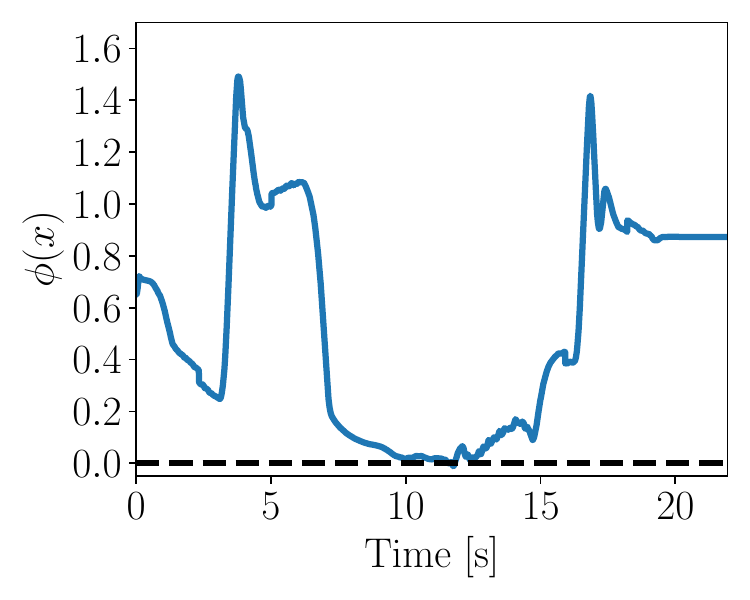}
     \caption{With circulation.}
     \label{fig:exp2_cbf_with_circ}
 \end{subfigure}
 \caption{Evolution of the CBFs.}
 \label{fig:exp2_cbf_values}
\end{figure}

In this section, we consider a whiteboard-cleaning task for the FR3 robotic manipulator. As shown in Fig.~\ref{fig:exp2_coordinates}, the gripper needs to hold the eraser (whose center coincides with the TCP) and clean the two shaded areas on the whiteboard while avoiding the rectangular area between them. To clean the whiteboard, the robot needs to press against the whiteboard and also apply a certain amount of force to overcome friction. To achieve this task, torque control is required for compliant motions. Therefore, the CBF has to be of a higher order. We note that this task cannot be directly achieved through existing first-order formulations in \cite{thirugnanam2022duality, wei2024diffocclusion, dai2023safe, thirugnanam2023nonsmooth}.

As the eraser is tightly pressed against the whiteboard, this task can be seen as collision avoidance in a 2D plane. The nominal controller, as illustrated in Fig.~\ref{fig:exp2_coordinates}, commands the TCP to track line trajectories that sequentially link the positions $p_0$, $p_1$, ..., $p_7$ while maintaining a constant orientation and height $z=0.02$. For this task, we are interested in the dynamics of the TCP in the 2D plane
\begin{equation} \label{eq:double_integrator_2d}
    \frac{d}{dt}\underbrace{\begin{bmatrix}
        p_{\text{tcp},1} \\
        p_{\text{tcp},2} \\
        v_{\text{tcp},1} \\
        v_{\text{tcp},2} 
    \end{bmatrix}}_{\hat{x}} = \underbrace{\begin{bmatrix}
        v_{\text{tcp},1} \\
        v_{\text{tcp},2} \\
        0 \\
        0 
    \end{bmatrix}}_{\hat{f}(\hat{x})} + 
    \underbrace{\begin{bmatrix}
        0 & 0 \\
        0 & 0 \\
        1 & 0 \\
        0 & 1 
    \end{bmatrix}}_{\hat{g}(\hat{x})} \underbrace{\begin{bmatrix}
        a_{\text{tcp},1} \\
        a_{\text{tcp},2} 
    \end{bmatrix}}_{\hat{u}}
\end{equation}
where $p_{\text{tcp},1}$ and $p_{\text{tcp},2}$ are the $x$ and $y$ coordinates of the origin of the TCP frame in the 2D plane $z=0.02$, respectively. $v_{\text{tcp},1}$ and $v_{\text{tcp},2}$ are the linear velocities, and $a_{\text{tcp},1}$ and $a_{\text{tcp},2}$ the linear accelerations. Then, for \eqref{eq:double_integrator_2d}, $\hat{\mathcal{X}}_e = \{[p_{\text{tcp}}^\top, 0]^\top \in \R^{4} \mid p_{\text{tcp},l} \leq  p_{\text{tcp}} \leq p_{\text{tcp},u} \}$ and $\hat{\zeta} (x_e) = 0$, where $p_{\text{tcp},l}, p_{\text{tcp},u} \in \R^2$ are the lower and upper bounds of the position of TCP. Given the time-varying trajectory $p_d (t)$, the nominal control is
\begin{equation}
    \hat{u}_n = \begin{bmatrix}
        \ddot{p}_{d,1} - K_{p,1}^p (p_{\text{tcp},1} - p_{d,1}) - K_{p,1}^d (v_{\text{tcp},1} - \dot{p}_{d,1}) \\
        \ddot{p}_{d,2} - K_{p,2}^p (p_{\text{tcp},2} - p_{d,2}) - K_{p,2}^d (v_{\text{tcp},2} - \dot{p}_{d,2})
    \end{bmatrix}
\end{equation}
where $K_{p,1}^p, K_{p,2}^p, K_{p,1}^d, K_{p,2}^d > 0$. 

The rectangular area that should not be cleaned is treated as a 2D polygon obstacle. To keep the eraser within the whiteboard, the four edges of the whiteboard are represented by four 2D hyperplanes. The eraser is bound with an ellipse in the 2D plane, resulting in a total of $K=5$ CBFs. For each CBF $h_i$ ($i=1, ..., K$), define again $\psi_{i,0} (\hat{x}) = h_i (\hat{x})$ and $\psi_{i,1} (\hat{x}) = \dot{\psi}_{i,0} (\hat{x})+ \gamma_1 \psi_{i,0} (\hat{x})$. Let $v_l = [v_{l,1}, v_{l,2}]^\top$ and $v_u = [v_{u,1}, v_{u,2}]^\top$ be the lower and upper bounds of the TCP's linear velocity in the 2D plane. Define also $h_{l,i}(x) = v_{\text{tcp},i} - v_{l,i}$ and $h_{u,i}(x) = v_{u,i}-v_{\text{tcp},i}$ with $i = 1,2$. Then, the HOCBF-QP is  
\begin{equation} \label{eq:exp2_hocbf_qp}
    \begin{aligned}
    \hat{\pi} (t, \hat{x}) &= \argmin_{\hat{u} \in \hat{\mathcal{U}}} \quad \lVert \hat{u}- \hat{u}_n(t, \hat{x}) \rVert_2^2 \\
    \textrm{s.t.} \quad & \dot{\psi}_{i,1} (\hat{x}) \geq -\gamma_2 \psi_{i,1} (\hat{x}), \ i=1,...,K,  \\
    & \dot{h}_{l,i} (\hat{x}) \geq -\gamma_l h_{l,i} (\hat{x}), \, i=1, 2, \\
    & \dot{h}_{u,i} (\hat{x}) \geq -\gamma_u h_{u,i} (\hat{x}), \, i=1, 2.
    \end{aligned}
\end{equation}

Similarly, define a new CBF based on the smooth minimum of $h_i$: $\phi(\hat{x}) = \varphi (h_1 (\hat{x}), ..., h_K (\hat{x})) - \varphi_0$. Define in addition $\psi_{\phi, 0} (\hat{x}) = \phi(\hat{x})$ and $\psi_{\phi,1} (\hat{x}) = \dot{\psi}_{\phi,0} (\hat{x})+ \gamma_1 \psi_{\phi,0} (\hat{x})$. The CHOCBF-QP is 
\begin{equation}
\begin{aligned}
    \hat{\pi}' (t, \hat{x}) &= \argmin_{\hat{u} \in \hat{\mathcal{U}}} \quad \lVert \hat{u}- \hat{u}_n(t, \hat{x}) \rVert_2^2 \label{eq:exp2_chocbf_qp} \\
    \textrm{s.t.} \quad & \dot{\psi}_{\phi,1} (\hat{x}) \geq -\gamma_2 \psi_{\phi,1} (\hat{x}),   \\
    & c(\hat{x})^\top [\hat{u} - \hat{\zeta}(P_{\hat{\mathcal{X}}_e} (\hat{x}))] \geq
    d(\phi(\hat{x}), \lVert \hat{x} - P_{\hat{\mathcal{X}}_e} (\hat{x}) \rVert_2)  \\
    & \dot{h}_{l,i} (\hat{x}) \geq -\gamma_l h_{l,i} (\hat{x}), \, i=1,2,  \\
    & \dot{h}_{u,i} (\hat{x}) \geq -\gamma_u h_{u,i} (\hat{x}), \, i=1,2 
\end{aligned}
\end{equation}
where the terms for the circulation inequality are $c(\hat{x}) = [-a_2 (\hat{x}), a_1 (\hat{x})]$, and $d(\phi(\hat{x}), \lVert \hat{x} - P_{\hat{\mathcal{X}}_e} (\hat{x}) \rVert_2) = d_1 (1-e^{d_2 (\phi(\hat{x}) - d_3)}) + d_4 (e^{-(v_{\text{tcp},1}^2 + v_{\text{tcp},2}^2 )/ d_5^2}-1)$ with $d_1, ..., d_5 > 0$. Then, the joint accelerations $\Tilde{u}$ are 
\begin{equation*}
    \Tilde{u} = \Tilde{u}_{\text{task}} + (I - J_{\text{tcp}}^\dagger J_{\text{tcp}}) \Tilde{u}_{\text{joint}}
\end{equation*}
where $\Tilde{u}_{\text{joint}}$ is given in \eqref{eq:joint_centering} and
\begin{align*}
    \Tilde{u}_{\text{task}} &= J_{\text{tcp}}^\dagger \Biggl ( - \dot{J}_{\text{tcp}} \dot{q} \\
    &+
    \begin{bmatrix}
        \hat{u} \\
        \ddot{p}_{d,z} - K_{p,z}^p (p_{\text{tcp},z} - p_{d,z}) - K_{p,z}^d (v_{\text{tcp},z} - \dot{p}_{d,z})\\
        - K_{r}^p \Log_{\text{SO}(3)} (R_{\text{tcp}} R_d^\top) - K_{r}^d \omega_{\text{tcp}} 
    \end{bmatrix} \Biggr)
\end{align*}
with $K_{p,z}^p, K_{p,z}^d > 0$ and $\hat{u} = \hat{\pi} (t, \hat{x})$ or $\hat{\pi}' (t, \hat{x})$. The actual torque inputs are 
\begin{equation*}
    u = M(q) \Tilde{u} + \sigma (q, \dot{q}) + J_{\text{tcp}}^\top (F_{\text{press}}+F_{\text{fric}})
\end{equation*}
where $F_{\text{press}}$ is the pressing force applied to the whiteboard and $F_{\text{fric}}$ is the friction between the eraser and whiteboard.

Figures~\ref{fig:exp2_no_circ_t_4} and \ref{fig:exp2_no_circ_t_10} show that the manipulator cannot move around the rectangular area and gets stuck in an equilibrium under the control $\hat{\pi}$ given by \eqref{eq:exp2_hocbf_qp}. As expected in Proposition~\ref{prop:equilibrium_on_boundary}, the minimum of $h_i$ at the equilibrium is near zero (see Fig.~\ref{fig:exp2_cbf_no_circ}). On the contrary, the control $\hat{\pi}'$ with circulation constraint commands the eraser to bypass the rectangular area and accomplishes the task (see Figs.~\ref{fig:exp2_with_circ_t_4}-\ref{fig:exp2_with_circ_t_19}, and \ref{fig:exp2_cbf_with_circ}). The values of the parameters used in this experiment are: $\alpha_0 = 1.03$, $\gamma_1 = \gamma_2 = 10$, $\gamma_l = \gamma_u = 40$, $\eta= 5$, $\varphi_0 = 0.3$, $d_1 = 200$, $d_2 = 5$, $d_3 = 0.07$, $d_4 = 1500$, and $d_5 = 0.75$. The p90 latency of this example is 1.9~\si{ms}.

\subsection{Avoiding a Flying Ball}\label{sec:exp_flying_ball}
\begin{figure}[t]
    \centering
    \begin{subfigure}[b]{0.24\textwidth}
     \centering
     \includegraphics[width=\textwidth]{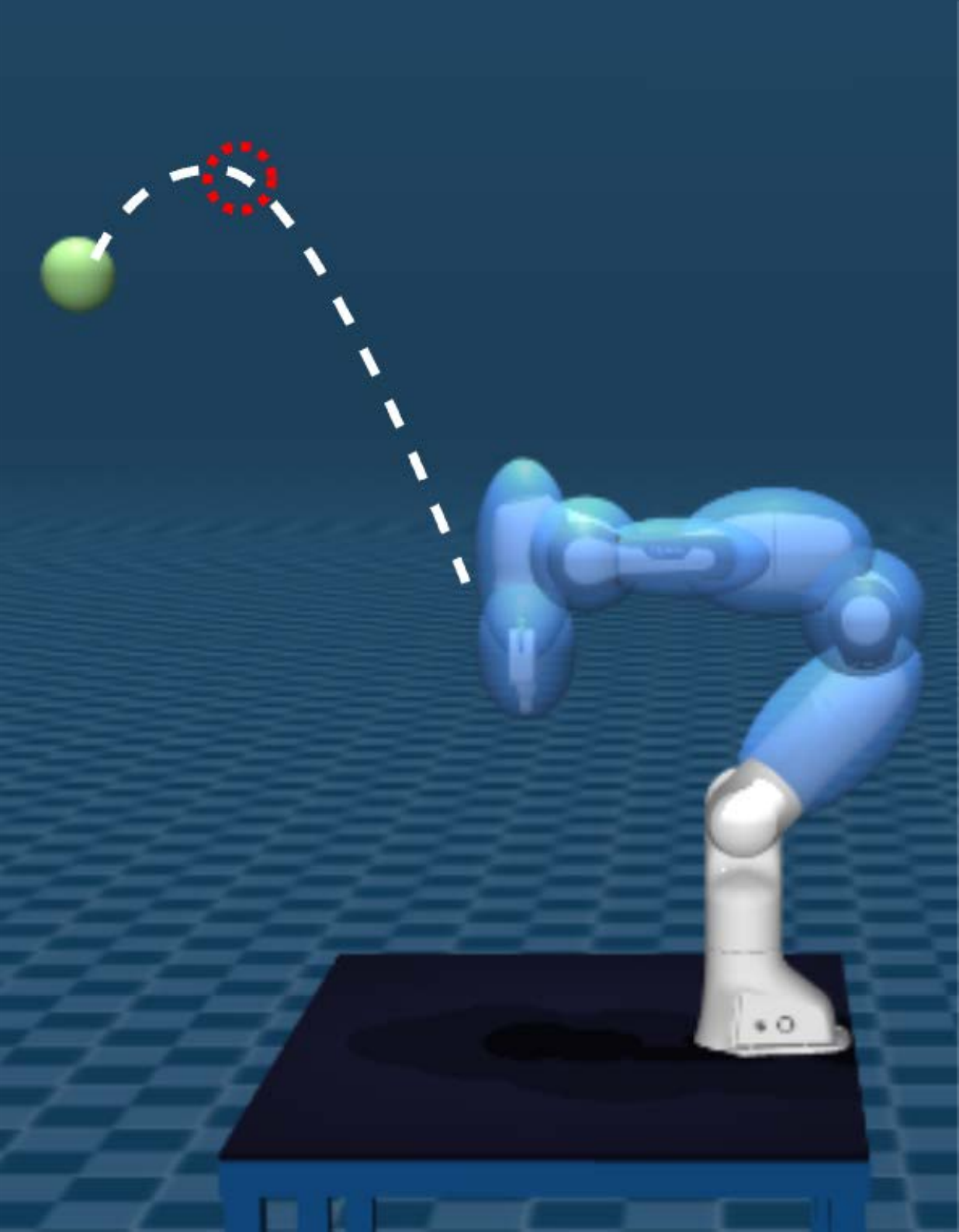}
     \caption{Overall setting.}
     \label{fig:exp3_bounding_shapes}
 \end{subfigure}
 \begin{subfigure}[b]{0.24\textwidth}
     \centering
     \includegraphics[width=\textwidth]{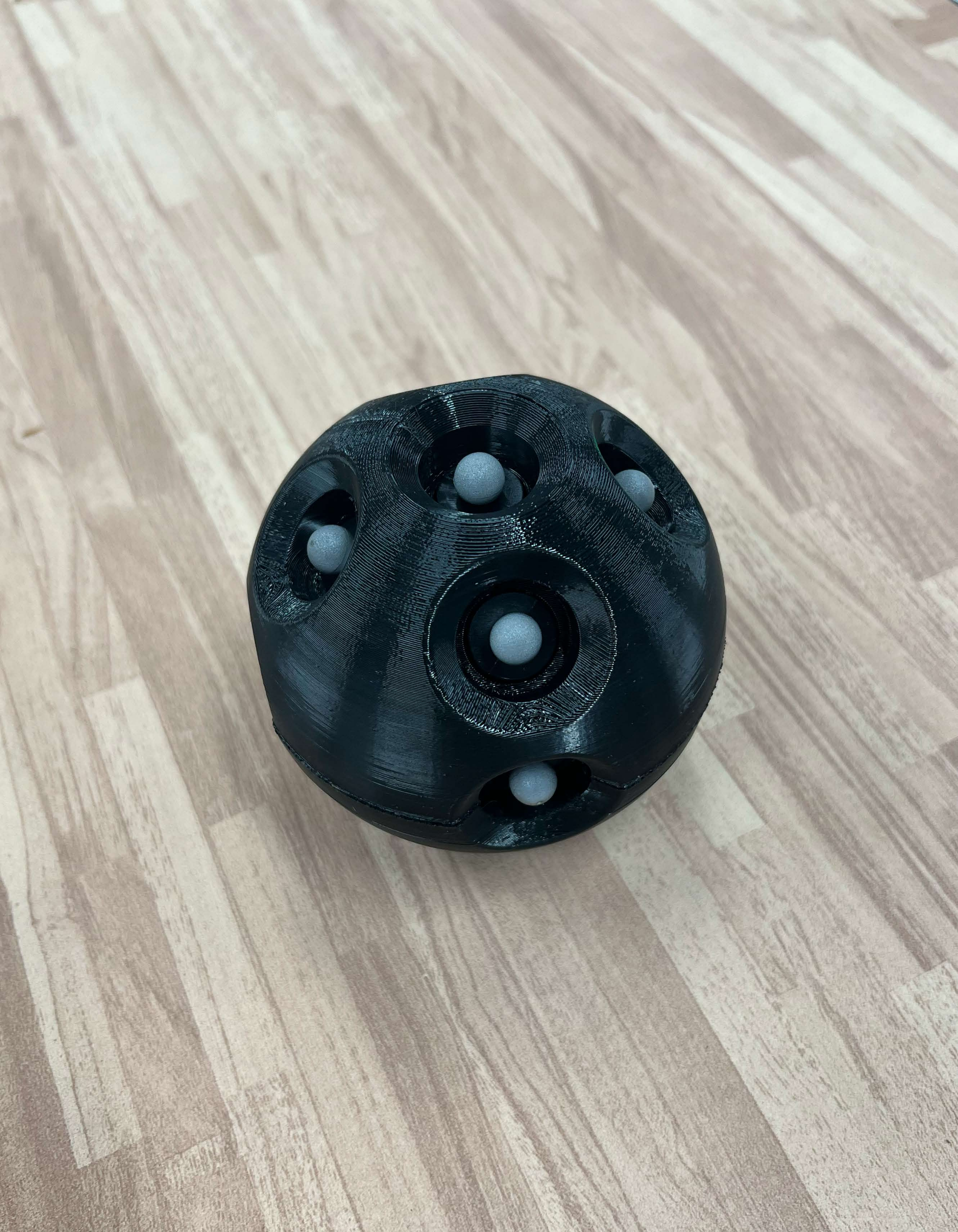}
     \caption{3D printed ball.}
     \label{fig:exp3_ball}
 \end{subfigure}
 \caption{Avoiding a flying ball. \subref{fig:exp3_bounding_shapes}: The FR3 robotic manipulator is bounded by seven ellipsoids. We account for both the actual ball (green), which is treated as a hard constraint, and a virtual ball (red dashed line), which is projected 150~\si{ms} ahead of the actual ball, which is treated as a soft constraint. The simulation figure is generated using MuJoCo. \subref{fig:exp3_ball}: A 3D-printed ball with markers for the Vicon motion capture system.}
 \label{fig:exp3_task}
\end{figure}

\begin{figure*}[t]
    \centering
    \begin{subfigure}[b]{0.24\textwidth}
        \centering
        \includegraphics[width=\textwidth]{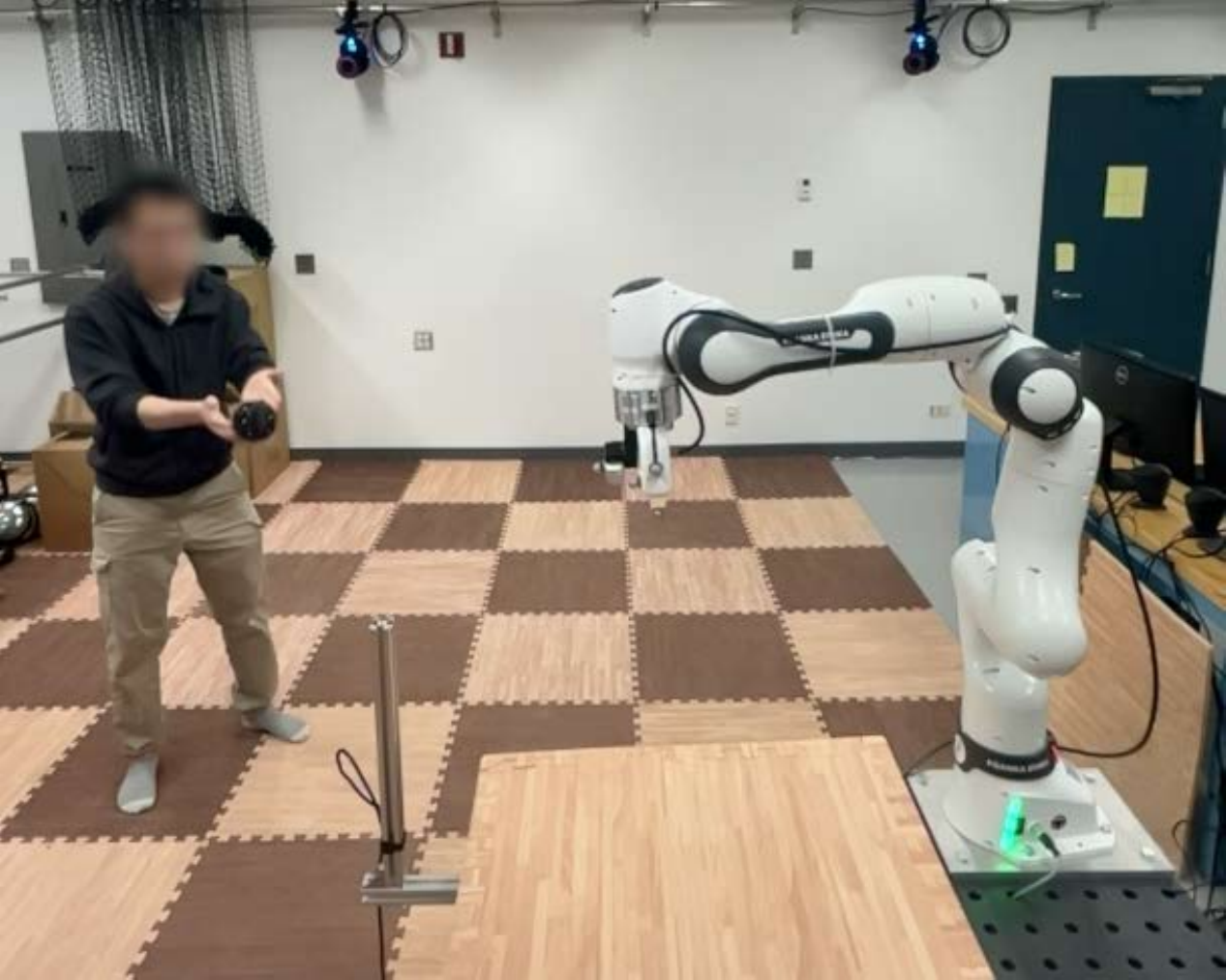}
        \caption{$t = 0$.}
        \label{fig:exp3_front_right_1}
    \end{subfigure}
    \begin{subfigure}[b]{0.24\textwidth}
        \centering
        \includegraphics[width=\textwidth]{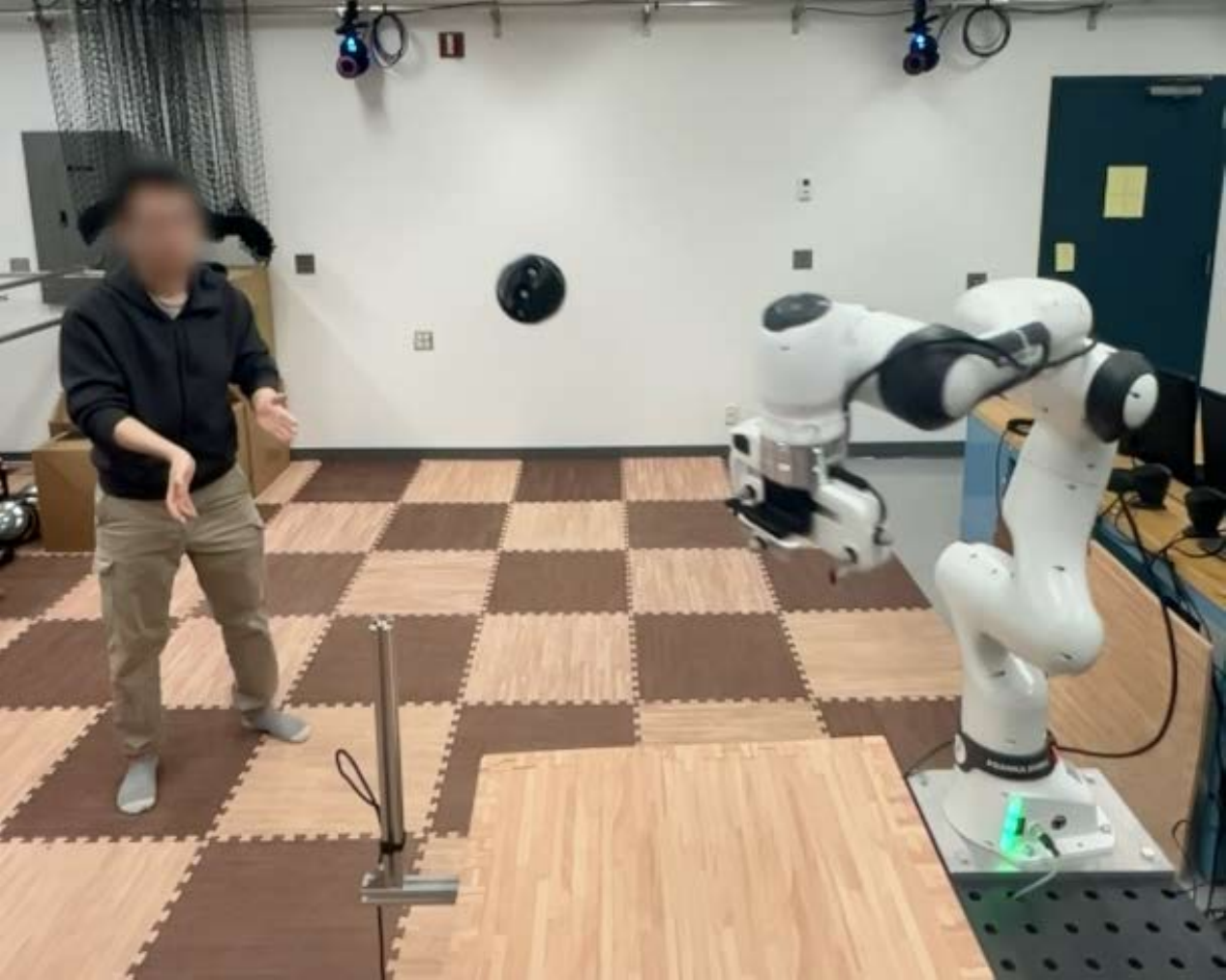}
        \caption{$t = 467$~\si{ms}.}
        \label{fig:exp3_front_right_2}
    \end{subfigure}
    \begin{subfigure}[b]{0.24\textwidth}
        \centering
        \includegraphics[width=\textwidth]{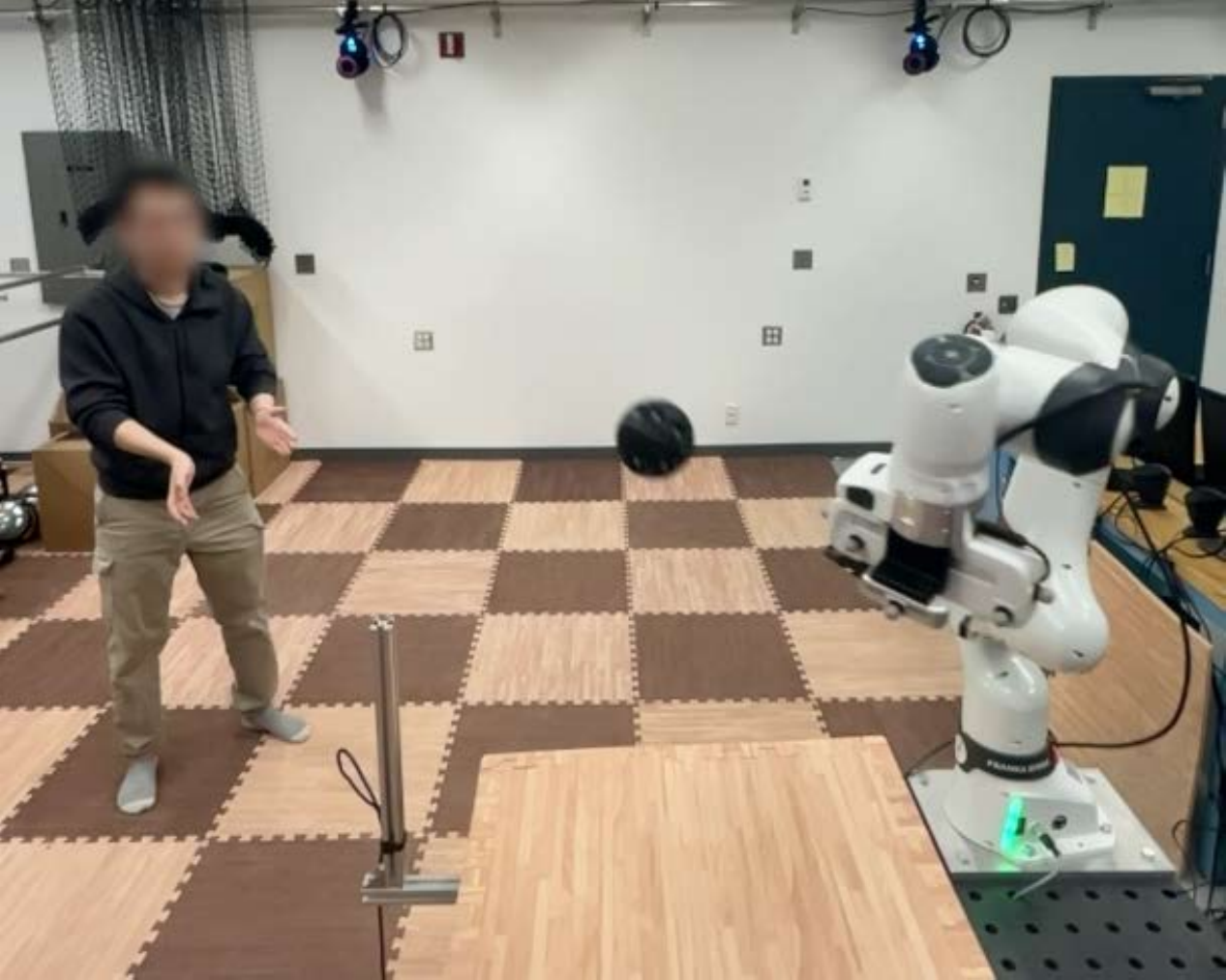}
        \caption{$t = 567$~\si{ms}.}
        \label{fig:exp3_front_right_3}
    \end{subfigure}
    \begin{subfigure}[b]{0.24\textwidth}
        \centering
        \includegraphics[width=\textwidth]{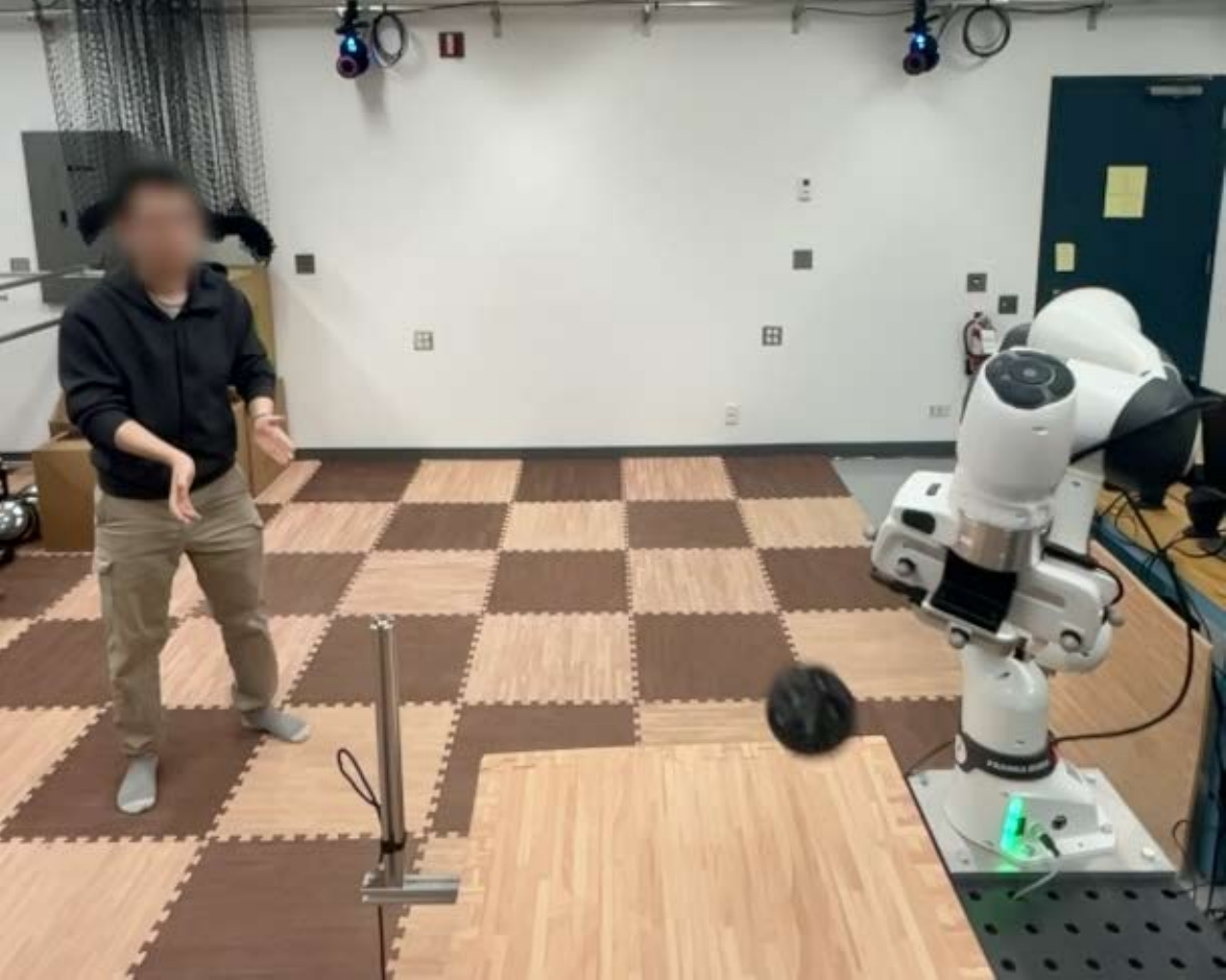}
        \caption{$t = 667$~\si{ms}.}
        \label{fig:exp3_front_right_4}
    \end{subfigure}
\caption{Throwing the ball from the front right of the robot.}
\label{fig:exp3_front_right}
\end{figure*}

\begin{figure*}[t]
    \centering
    \begin{subfigure}[b]{0.24\textwidth}
        \centering
        \includegraphics[width=\textwidth]{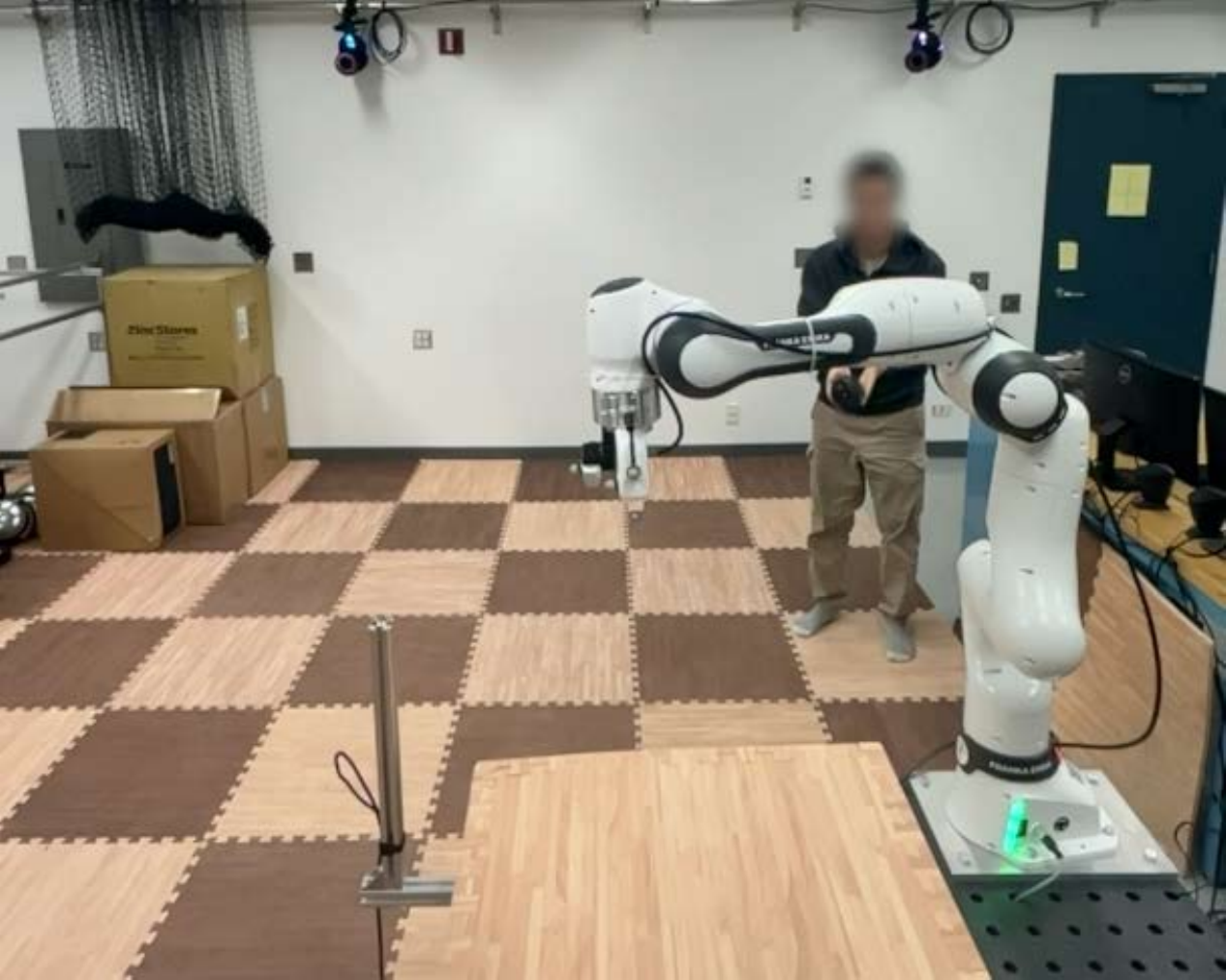}
        \caption{$t = 0$.}
        \label{fig:exp3_back_right_1}
    \end{subfigure}
    \begin{subfigure}[b]{0.24\textwidth}
        \centering
        \includegraphics[width=\textwidth]{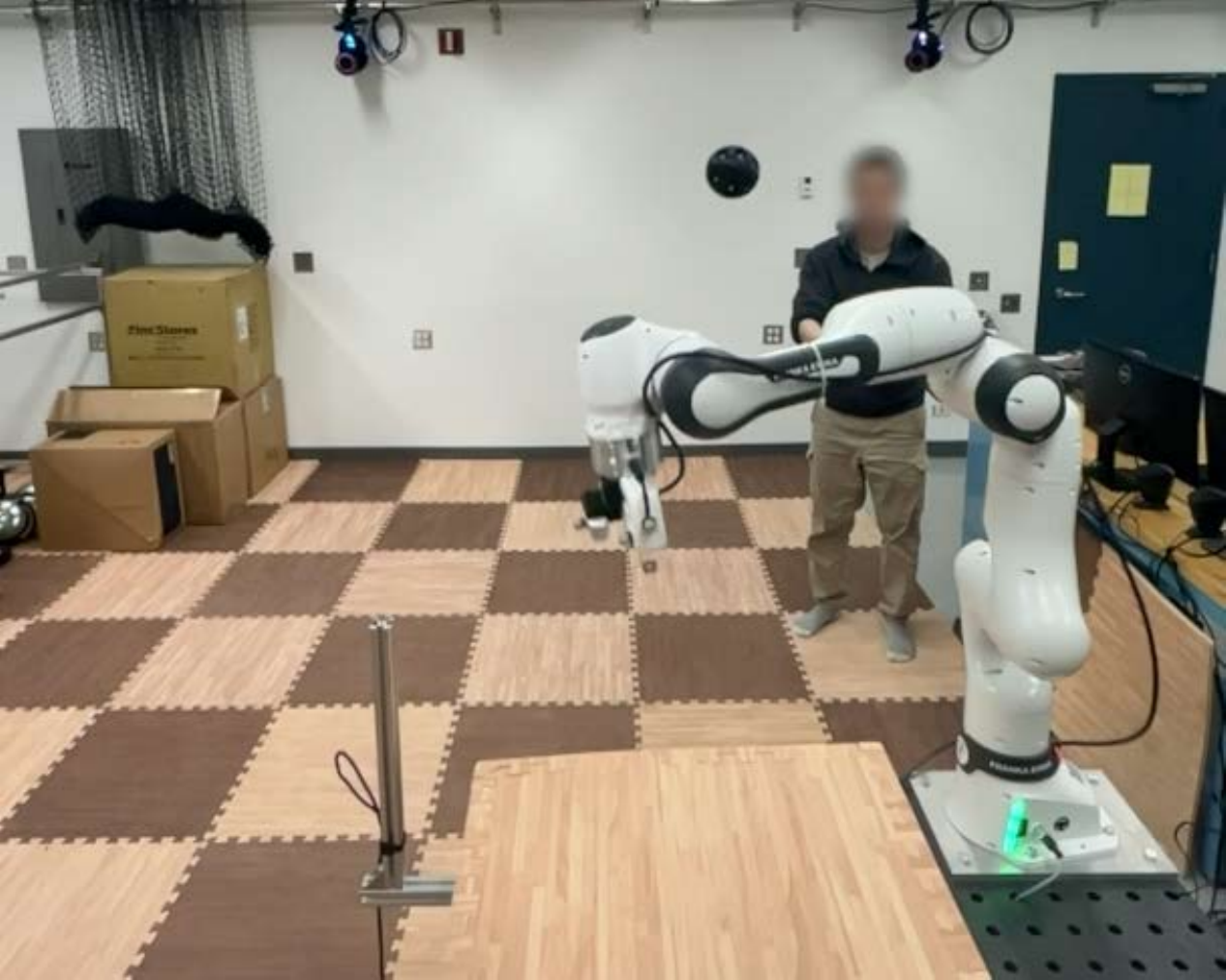}
        \caption{$t = 367$~\si{ms}.}
        \label{fig:exp3_back_right_2}
    \end{subfigure}
    \begin{subfigure}[b]{0.24\textwidth}
        \centering
        \includegraphics[width=\textwidth]{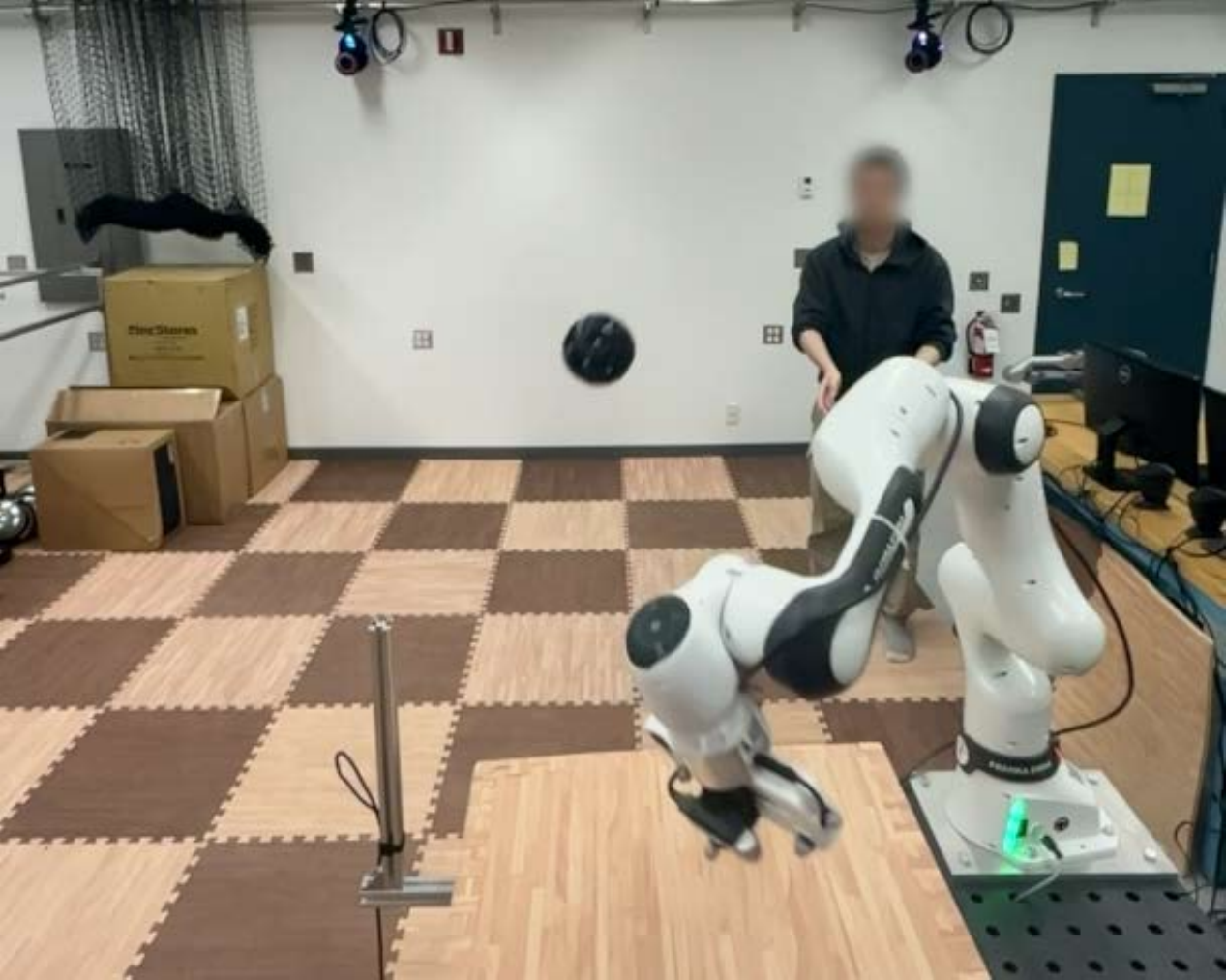}
        \caption{$t = 600$~\si{ms}.}
        \label{fig:exp3_back_right_3}
    \end{subfigure}
    \begin{subfigure}[b]{0.24\textwidth}
        \centering
        \includegraphics[width=\textwidth]{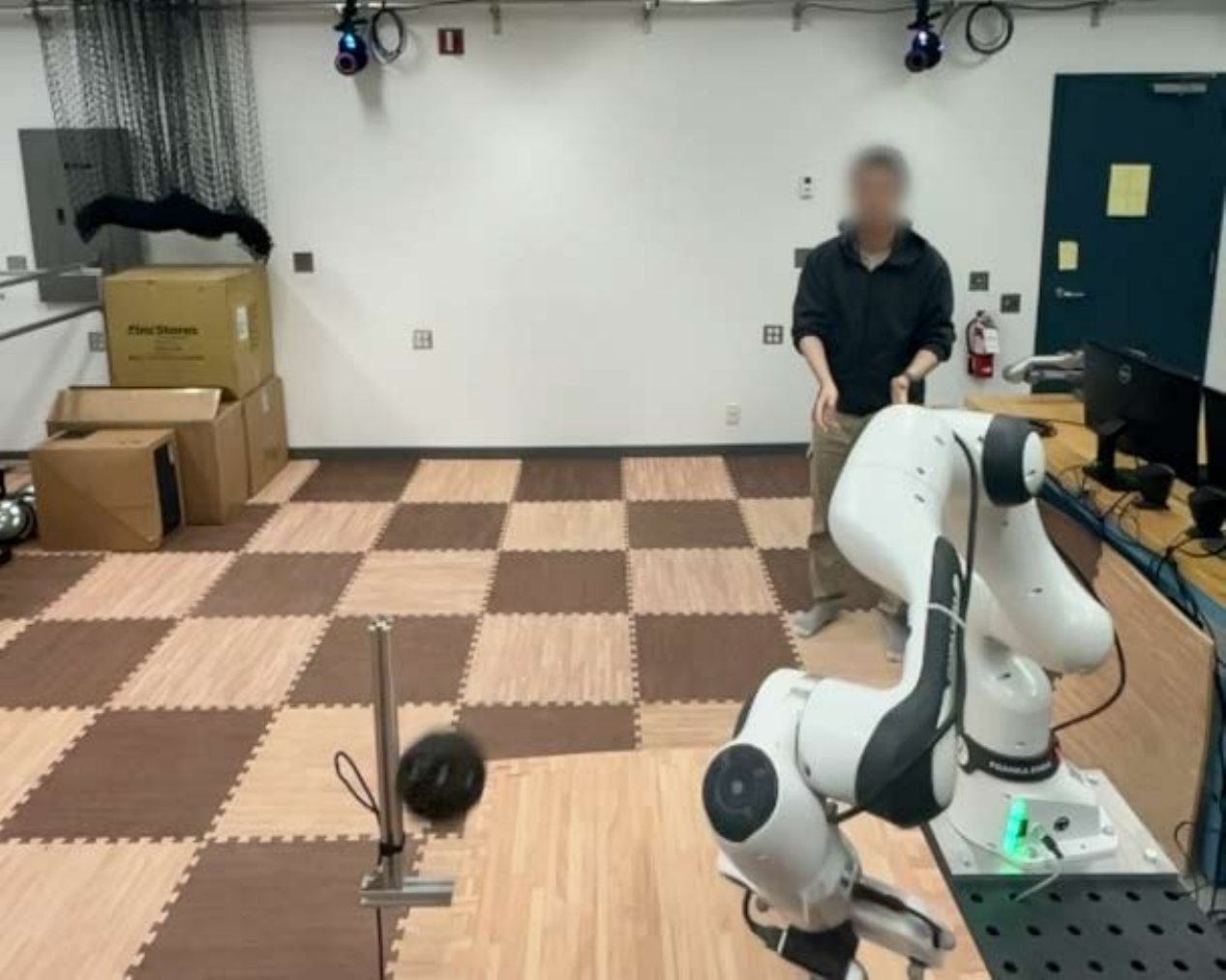}
        \caption{$t = 733$~\si{ms}.}
        \label{fig:exp3_back_right_4}
    \end{subfigure}
\caption{Throwing the ball from the back right of the robot.}
\label{fig:exp3_back_right}
\end{figure*}

\begin{figure*}[t]
    \centering
    \begin{subfigure}[b]{0.24\textwidth}
        \centering
        \includegraphics[width=\textwidth]{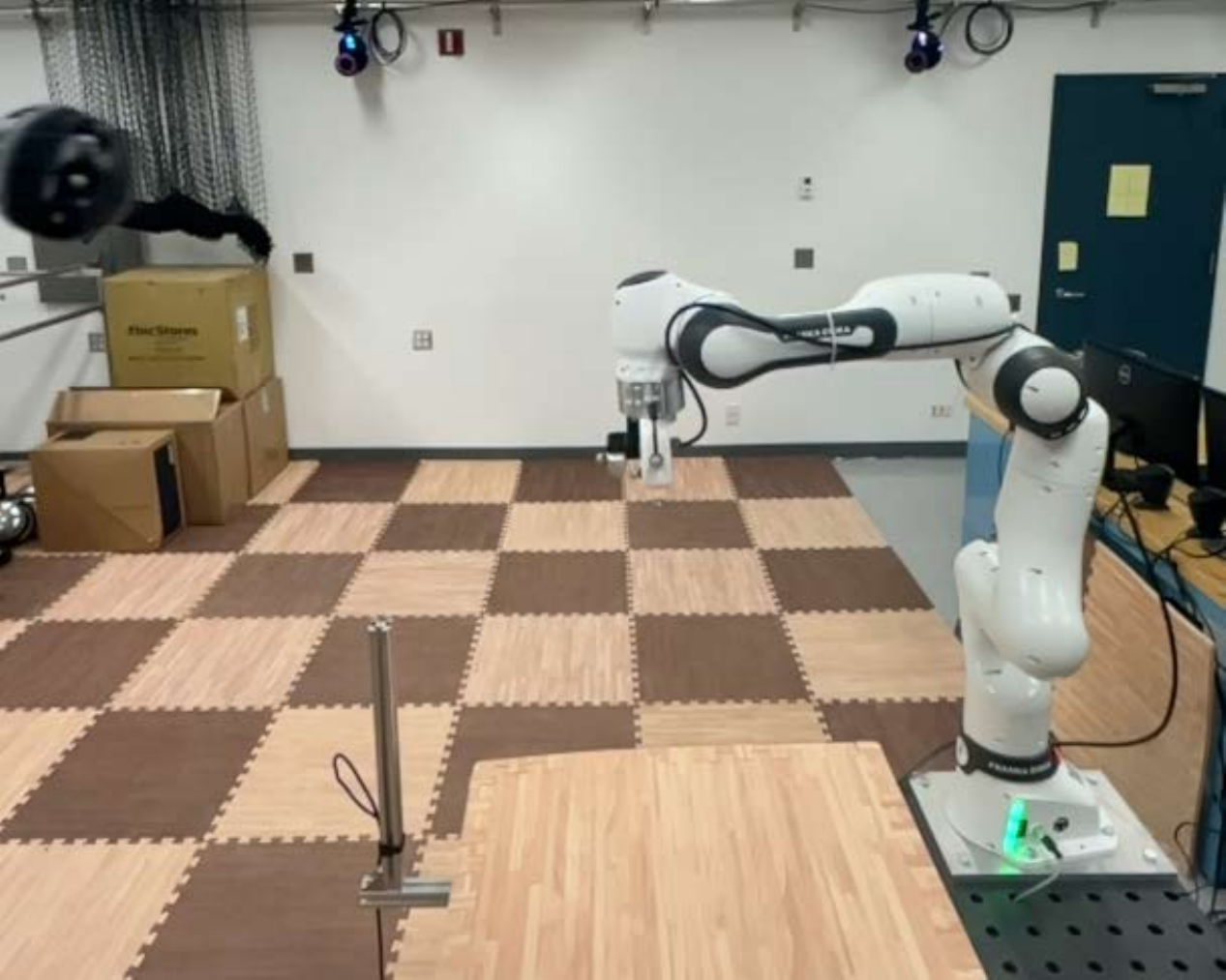}
        \caption{$t = 0$.}
        \label{fig:exp3_front_left_1}
    \end{subfigure}
    \begin{subfigure}[b]{0.24\textwidth}
        \centering
        \includegraphics[width=\textwidth]{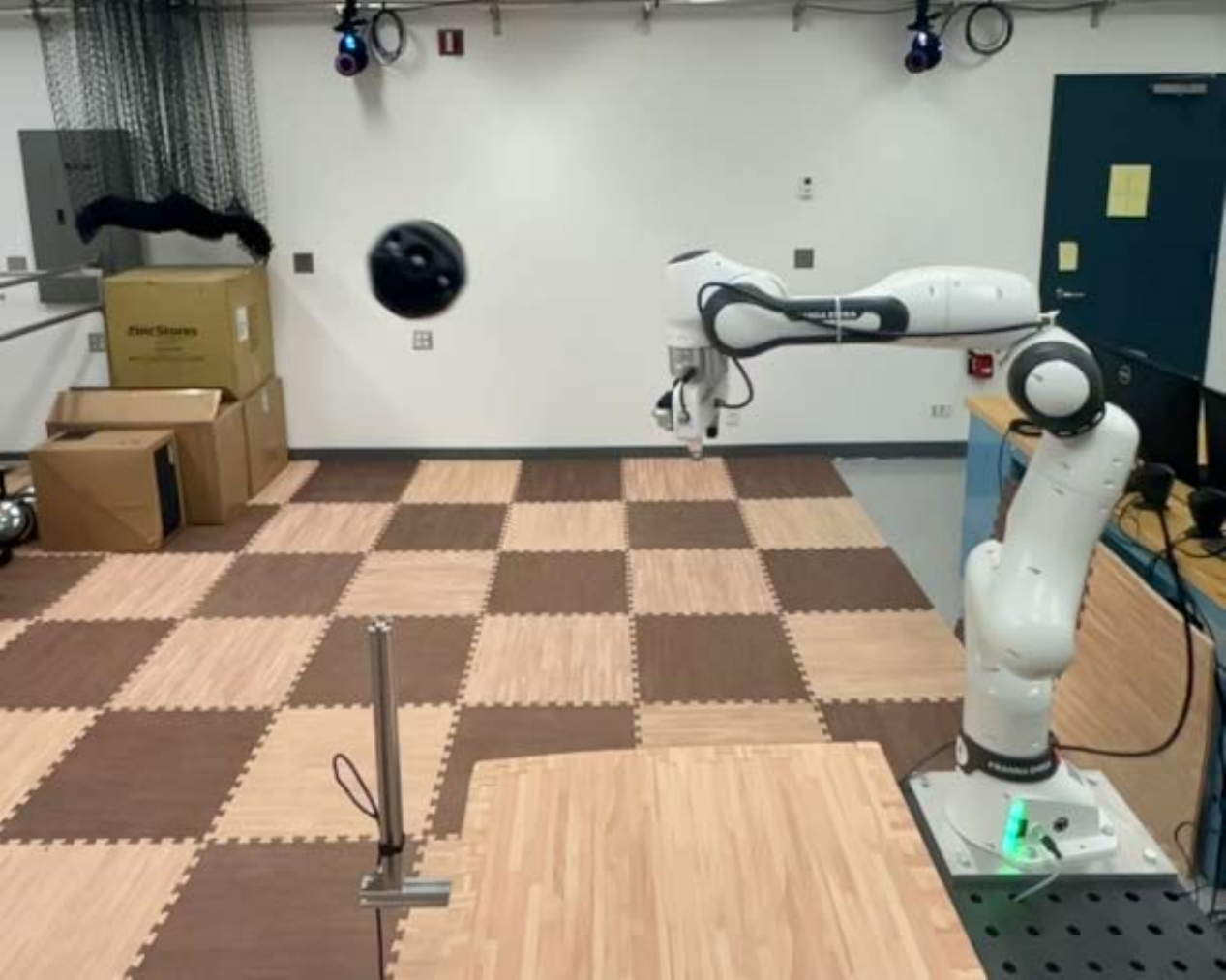}
        \caption{$t = 134$~\si{ms}.}
        \label{fig:exp3_front_left_2}
    \end{subfigure}
    \begin{subfigure}[b]{0.24\textwidth}
        \centering
        \includegraphics[width=\textwidth]{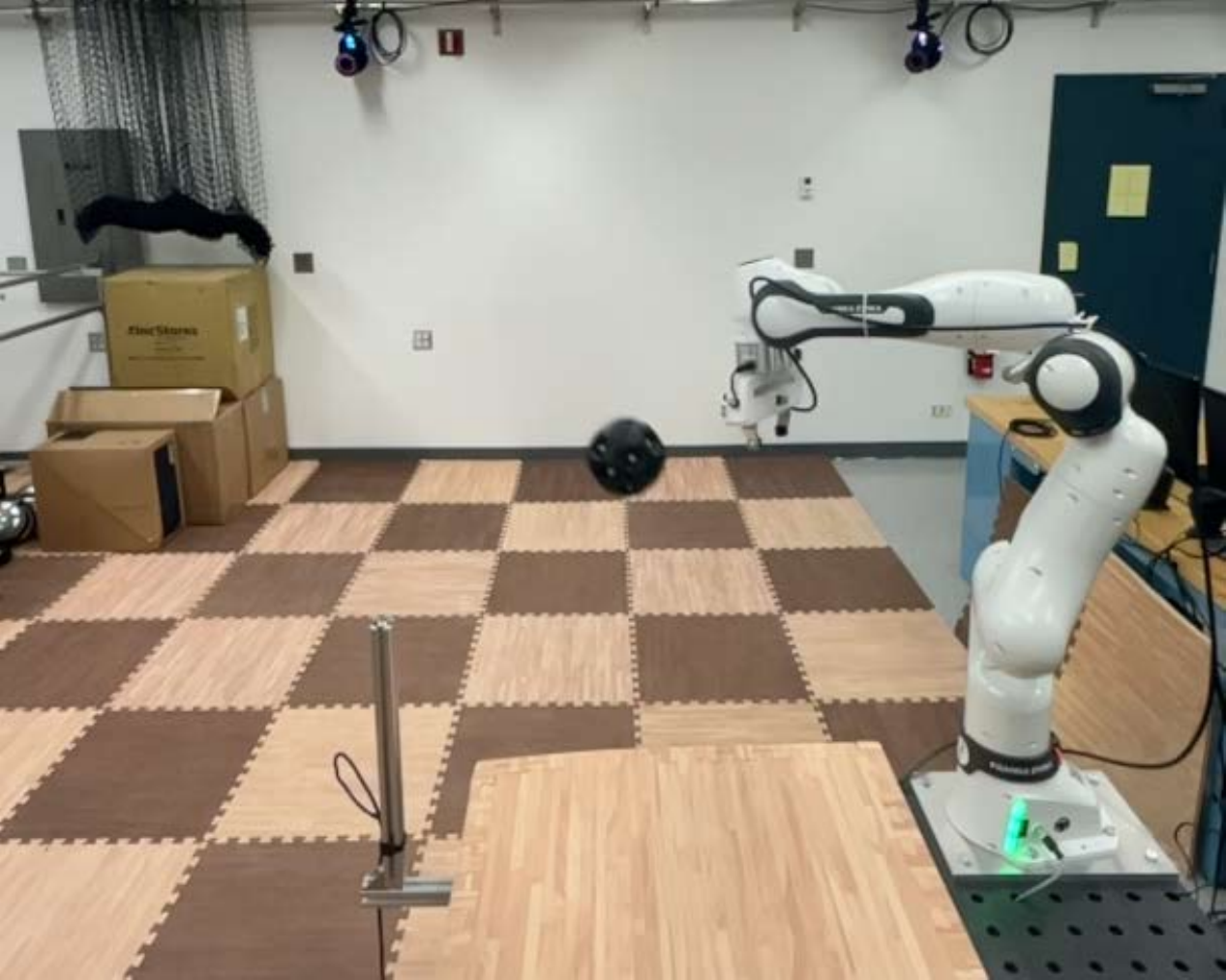}
        \caption{$t = 267$~\si{ms}.}
        \label{fig:exp3_front_left_3}
    \end{subfigure}
    \begin{subfigure}[b]{0.24\textwidth}
        \centering
        \includegraphics[width=\textwidth]{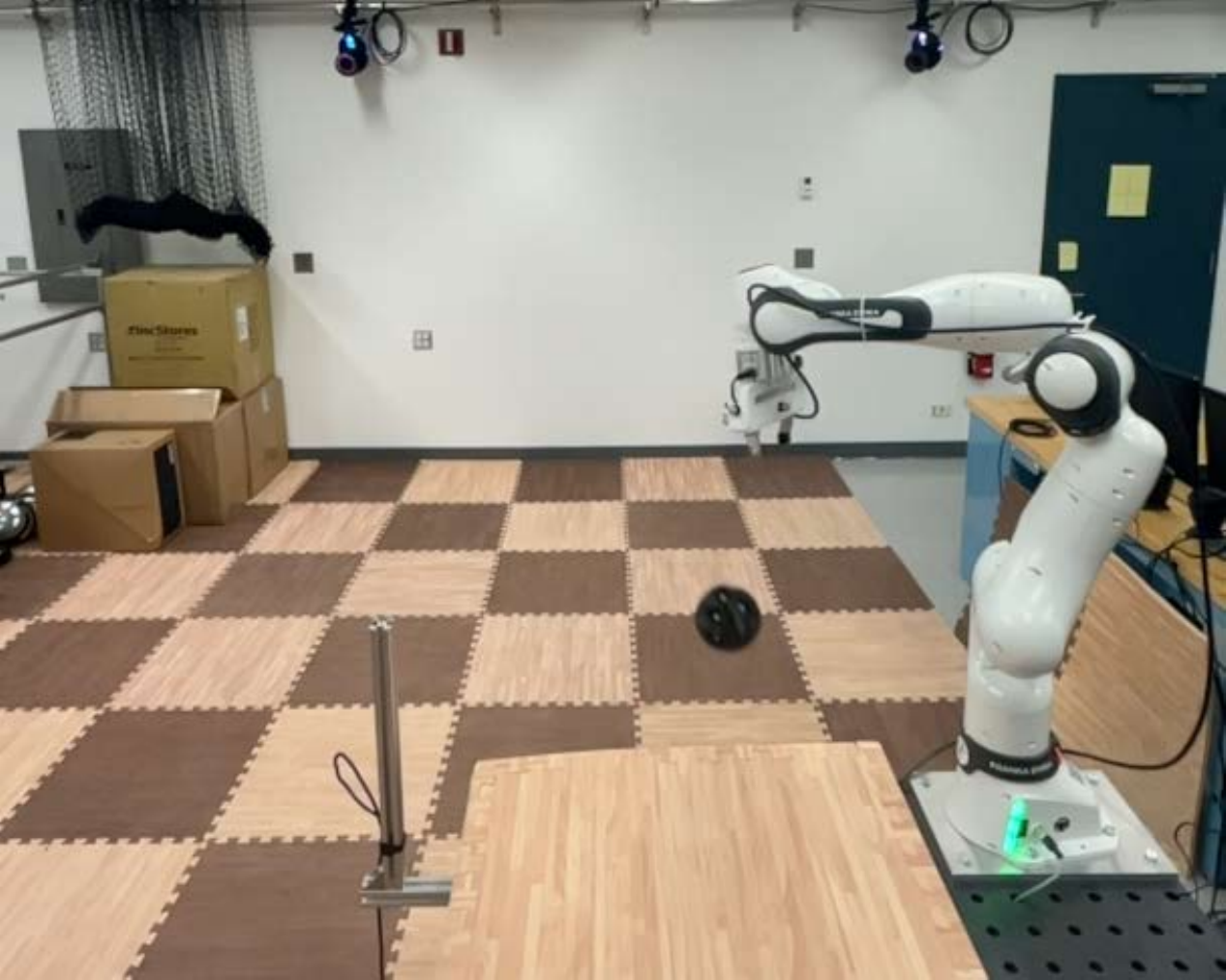}
        \caption{$t = 367$~\si{ms}.}
        \label{fig:exp3_front_left_4}
    \end{subfigure}
\caption{Throwing the ball from the front left of the robot. $t=0$ does not correspond to the moment the ball is released. }
\label{fig:exp3_front_left}
\end{figure*}

To further demonstrate the reactiveness and agility of a second-order CBF, we consider an application of commanding the FR3 manipulator to avoid a flying ball (see Fig.~\ref{fig:exp3_task}). As shown in Fig.~\ref{fig:exp3_bounding_shapes}, we consider the collision between the ball and the bounding ellipsoids of the robot. Including the dynamics of the ball, the overall system dynamics with joint accelerations $\Tilde{u}$ as control input are
\begin{equation}
    \frac{d}{dt} \underbrace{\begin{bmatrix}
        q \\ \dot{q} \\ p_{\text{ball}} \\ v_{\text{ball}} 
    \end{bmatrix}}_{\bar{x}} = \underbrace{\begin{bmatrix}
        \dot{q} \\ 0 \\ v_{\text{ball}}  \\ a_{\text{gravity}}
    \end{bmatrix}}_{\bar{f}(\bar{x})} + \underbrace{\begin{bmatrix}
        0 \\ I \\ 0 \\ 0
    \end{bmatrix}}_{\bar{g}(\bar{x})} \Tilde{u}
\end{equation}
where $p_{\text{ball}}, v_{\text{ball}} \in \R^3$ are the position and linear velocity of the center of the ball, respectively, and $a_{\text{gravity}} \in \R^3$ is the gravity acceleration.

In this task, the flying ball is modeled by a sphere with a radius of $5$~\si{cm}. The robot is bound by seven ellipsoids as shown in Fig.~\ref{fig:exp3_bounding_shapes}, leading to seven CBFs $h_1, h_2, ..., h_7$. The time between when the ball is released and when it reaches the initial position of the end-effector is approximately $550$ to $650$~\si{ms}. We run a Kalman filter based on the observed positions of the actual ball from the motion capture system, and fit a parabolic trajectory after the ball is released. To allow for anticipation, we define a virtual ball with a position and velocity $150$~\si{ms} ahead of the actual ball (see Fig.~\ref{fig:exp3_bounding_shapes}) based on the parabolic trajectory. We also define another seven CBFs $h_8, h_9, ..., h_{14}$ to account for the collision between the robot and this virtual ball, but the CBF constraints will be taken as soft constraints. For each CBF $h_i$ ($i = 1, ..., 14$), define $\psi_{i,0} (\bar{x}) = h_i (\bar{x})$ and $\psi_{i,1} (\bar{x}) = \dot{\psi}_{i,0} (\bar{x})+ \gamma_1 \psi_{i,0} (\bar{x})$. Let $q_l, q_u \in \R^7$ be the lower and upper bounds of the joint angles. To constrain the joint angles, define $\psi_{q_l,i,0} (\bar{x}) = q_i - q_{l,i}$, $\psi_{q_l, i,1} (\bar{x}) = \dot{\psi}_{q_l, i,0} (\bar{x})+ \gamma_{q_l,1} \psi_{q_l,i,0} (\bar{x})$, $\psi_{q_u,i,0} (\bar{x}) = q_{u,i} - q_i $, and $\psi_{q_u, i,1} (\bar{x}) = \dot{\psi}_{q_u, i,0} (\bar{x})+ \gamma_{q_u,1} \psi_{q_u,i,0} (\bar{x})$. In addition, let $\dot{q}_l, \dot{q}_u \in \R^7$ be the lower and upper bounds of the joint velocity. Define $h_{\dot{q}_l,i}(\bar{x}) = \dot{q}_i - \dot{q}_{l,i}$ and $h_{\dot{q}_u,i}(\bar{x}) = \dot{q}_{u,i}-\dot{q}_{i}$ with $i = 1,2,...,7$ to constrain the joint velocity. Let $\tau_l, \tau_u \in \R^7$ be the lower and upper bounds on the input torque. Then, the HOCBF-QP is defined as
\begingroup
\allowdisplaybreaks
\begin{align*}
    \Tilde{\pi} (\bar{x}) &= \argmin_{\Tilde{u} \in \Tilde{\mathcal{U}}} \quad  \lVert \Tilde{u}- \Tilde{u}_n(\bar{x}) \rVert_2^2 + 100 \delta^2 \\
    \textrm{s.t.} \quad & \dot{\psi}_{i,1} (\bar{x}) \geq -\gamma_2 \psi_{i,1} (\bar{x}), \ i=1,...,7  \\
    & \dot{\psi}_{i,1} (\bar{x}) \geq -\gamma_2 \psi_{i,1} (\bar{x}) + \delta, \ i=8,...,14\\
    & \dot{\psi}_{q_l,i,1} \geq -\gamma_{q_l,2} \psi_{q_l,i,1}, i=1,...,7,\\
    & \dot{\psi}_{q_u,i,1} \geq -\gamma_{q_u,2} \psi_{q_u,i,1}, i=1,...,7,\\
    & \dot{h}_{\dot{q}_l,i} (\bar{x}) \geq -\gamma_{\dot{q}_l} h_{\dot{q}_l,i} (\bar{x}), \, i=1,...,7, \\
    & \dot{h}_{\dot{q}_u,i} (\bar{x}) \geq -\gamma_{\dot{q}_u} h_{\dot{q}_u,i} (\bar{x}), \, i=1,...,7, \\
    & \tau_l \leq M(q) \Tilde{u} + \sigma (q, \dot{q})  \leq \tau_u
\end{align*}
\endgroup
where $\delta \in \R$ is a relaxation variable, and $\gamma_{q_l,2}, \allowbreak \gamma_{q_u,2}, \allowbreak \gamma_{\dot{q}_l}, \allowbreak \gamma_{\dot{q}_u} > 0$.  The nominal control is designed to track a constant configuration $q_d \in \R^7$, and we have
\begin{equation*}
    \Tilde{u}_n(\bar{x}) = - K_q^p (q - q_d) - K_q^d \dot{q}
\end{equation*}
where $K_q^p, K_q^d > 0$ are the 7-by-7 gain matrices. 

As shown in Figs.~\ref{fig:exp3_front_right}, \ref{fig:exp3_back_right}, and \ref{fig:exp3_front_left}, we throw the ball at the end-effector of the robot from three different directions, and the robot successfully avoids the flying ball in time, demonstrating the effectiveness and agility of our HOCBF. The values of the parameters used in this experiment are: $\alpha_0 = 1.03$, $\gamma_1 = \gamma_2 = 15$, $\gamma_{q_l,1} = \gamma_{q_u,1} = \gamma_{q_l,2} = \gamma_{q_u,2} = \gamma_{\dot{q}_l} = \gamma_{\dot{q}_u} = 10$. The p90 latency of this example is 7.6~\si{ms}.

\section{Discussion and Conclusion}\label{sec:conclusion}
In conclusion, this work presents a novel framework for addressing the fundamental challenge of collision avoidance between general convex primitives in torque-controlled systems by leveraging differentiable optimization and HOCBFs. Our approach comes with a systematic way of representing convex primitives as scaling functions and a circulation mechanism to prevent spurious equilibria on the boundary of the safe set, extending the applicability of CBFs to a broader class of use cases. We have provided the proof of continuous differentiability for the minimal scaling factor and the construction of smooth scaling functions. Additionally, the practical performance of our framework is demonstrated through real-world experiments with the Franka Research 3 robotic manipulator, showcasing its effectiveness in achieving collision-free motion. The insights and tools provided in this work will significantly contribute to the safety and reliability of dynamical systems in robotics, autonomous driving, and other safety-critical applications.

\appendix
\subsection{Proof of Lemma~\ref{lemma:uniqueness}}
\label{sec:proof_uniqueness}
If $\mathcal{F}_A$ is strictly convex in $p$ for each $\theta$, then \eqref{eq:opt} has a unique primal variable $p^\star$, as \eqref{eq:opt} is feasible and has a strictly convex objective function. If $\mathcal{F}_B$ is strictly convex in $p$ for each $\theta$, we prove by contradiction that the optimal primal variable $p^\star$ is unique. Assume that there exist two minimizers $p_1$ and $p_2$ (with $p_1 \neq p_2$) such that they are feasible and achieve the same optimal value $\alpha^\star$, i.e., $\mathcal{F}_A (p_i, \theta) = \alpha^\star$ and $\mathcal{F}_B (p_i, \theta) \leq 1$ for $i=1,2$. Take $0<\eta<1$ and $p_3 = \eta p_1 + (1-\eta) p_2$. As $\mathcal{F}_A$ is convex in $p$, we have $\mathcal{F}_A (p_3, \theta) \leq \eta \mathcal{F}_A (p_1, \theta) + (1-\eta) \mathcal{F}_A (p_2, \theta) = \alpha^\star$. As $\mathcal{F}_B$ is strictly convex in $p$, we have $\mathcal{F}_B (p_3, \theta) < \eta \mathcal{F}_B (p_1, \theta) + (1-\eta) \mathcal{F}_B (p_2, \theta) \leq 1$. Then, there exists $\varepsilon > 0$ such that $\mathcal{B}(p_3, \varepsilon) \subset B$ where $\mathcal{B}(p_3, \varepsilon)$ is the $\ell_2$ ball centered at $p_3$ with a radius of $\varepsilon$. In addition, as $p_3 \in B$ and $A \intersect B = \varnothing$, we have $\mathcal{F}_A(p_3, \theta) > 1$ and $\frac{\partial \mathcal{F}_A}{\partial p}(p_3, \theta) \neq 0$ by Lemma~\ref{lemma:non_zero_grad}. Then, there must exist some $p' \in \mathcal{B}(p_3, \varepsilon) \subset B$ such that $\mathcal{F}_A(p', \theta) < \mathcal{F}_A(p_3, \theta) \leq \alpha^\star$. This contradicts the assumption that $\alpha^\star$ is the optimal value. Therefore, the optimal primal variable $p^\star$ is unique, and the optimal dual variable $\lambda^\star$ is uniquely determined by \eqref{eq:kkt_p}.

\subsection{Proof of Theorem~\ref{thm:high_order_cd_alpha}}
\label{sec:proof_high_order_cd_alpha}
For a function $\psi$, we use $\psi^{(k)}$ as its $k$th-order derivative and $\psi^{[k]}$ as a short-hand notation for the collection of $\psi^{(0)}, ..., \psi^{(k)}$. Let $X = [p^{\star \top}, \lambda^\star]^\top$ where $p^\star$ and $\lambda^\star$ are the optimal primal and dual variable to \eqref{eq:opt}. We prove by mathematical induction the following claim: if $\mathcal{F}_A$ and $\mathcal{F}_B$ are $\mathcal{C}^{k+1}$ in $p$ and $\theta$, then: (a)
$X^{(k)} = F^k (N^{-1}, N^{[k-1]}, \Omega^{[k-1]})$ and $\alpha^{\star (k)} = G^k(p^{\star [k]}, \mathcal{F}_A^{[k]}(p^\star, \theta))$, where $F^k$ and $G^k$ are $\mathcal{C}^{\infty}$ functions in their arguments; (b) $X$ is $\mathcal{C}^k$ in $\theta$. The matrices $N$ and $\Omega$ are defined in \eqref{eq:implicit_eq}.

For $k = 1$, we already proved the claim in Theorem~\ref{thm:cd_alpha}. We had $X^{(1)} = N^{-1} \Omega = F^1 (N^{-1}, N^{[0]}, \Omega^{[0]})$ and $\alpha^{\star (1)} = \frac{\partial \mathcal{F}_A}{\partial p}(p^\star, \theta) \frac{\partial p^\star}{\partial \theta}(\theta) + \frac{\partial \mathcal{F}_A}{\partial \theta}(p^\star, \theta) = G^1(p^{\star [1]}, \mathcal{F}_A^{[1]}(p^\star, \theta))$. $X$ is $\mathcal{C}^1$ in $\theta$ by the implicit function theorem.

Assume that the claim holds for $k = k_0$ with $k_0 \geq 1$, i.e., $X^{(k_0)} = F^{k_0} (N^{-1}, N^{[k_0-1]}, \Omega^{[k_0-1]})$ and $\alpha^{\star (k_0)} = G^{k_0}(p^{\star [k_0]}, \mathcal{F}_A^{[k_0]}(p^\star, \theta))$, and $X$ is $\mathcal{C}^{k_0}$ in $\theta$.

For $k = k_0 +1$, assume that $\mathcal{F}_A$ and $\mathcal{F}_B$ are $\mathcal{C}^{k_0+2}$ (by the condition of Theorem~\ref{thm:high_order_cd_alpha}). Noting that $\frac{\partial N^{-1}}{\partial \theta_i} = - N^{-1} \frac{\partial N}{\partial \theta_i} N^{-1}$ where $\theta_i$ is the $i$th component of $\theta$, we have again $X^{(k_0 +1)} = F^{k_0 + 1} (N^{-1}, N^{[k_0]}, \Omega^{[k_0]})$. Differentiating $\alpha^{\star (k_0)}$ w.r.t. $\theta$, we have $\alpha^{\star (k_0 +1)} = G^{k_0+1}(p^{\star [k_0+1]}, \mathcal{F}_A^{[k_0+1]}(p^\star, \theta))$. $N^{[k_0]}$ and $\Omega^{[k_0]}$ are continuous in $\theta$, because $X$ is $\mathcal{C}^{k_0}$ in $\theta$ and $\mathcal{F}_A$ and $\mathcal{F}_B$ are $\mathcal{C}^{k_0+2}$ in $p$ and $\theta$. Therefore, $X^{(k_0 +1)}$ is continuous in $\theta$, i.e., $X$ is $\mathcal{C}^{k_0 + 1}$ in $\theta$. 

Finally, as $\alpha^{\star (k)} = G^k(p^{\star [k]}, \mathcal{F}_A^{[k]}(p^\star, \theta))$, $p^\star$ is $\mathcal{C}^k$ in $\theta$, and $\mathcal{F}_A$ is $\mathcal{C}^{k+1}$ in $p$ and $\theta$, we have the desired result, i.e., $\alpha^\star$ is $\mathcal{C}^k$ in $\theta$.

\subsection{Details on the Scaling Functions of Convex Polygons}
\label{sec:log_sum_exp_suggestions}
Let $G = [a_1', ..., a_N']^\top \in \mathbb{R}^{N \times n_p}$. Given $p$, take $c = \max_{i=1}^N \left(a_i'^\top p + b_i'\right)$ and we have $\mathcal{F}_{P_r} (p,\theta) = \frac{1}{\kappa}\ln \left( \frac{1}{N} \sum_{i=1}^N e^ { \kappa \left(a_i'^\top p + b_i'- c \right) } \right)+c+1$ which avoids the overflow of the exponentials as $a_i'^\top p + b_i'  - c \leq 0$ for all $i = 1,..., N$ (known as the log-sum-exp trick). Then, let $z = [e^ { \kappa \left(a_1'^\top p + b_1' - c\right) }, ..., e^ { \kappa \left(a_N'^\top p + b_N'- c \right) }]^\top \in \mathbb{R}^{N}$, and we have 
\begin{align}
    \mathcal{F}_{P_r} (p,\theta) &= \frac{1}{\kappa}\ln (1^\top z) + c - \frac{1}{\kappa}\ln N + 1, \\
    \frac{\partial \mathcal{F}_{P_r}}{\partial p}(p, \theta) &= \frac{z^\top G}{ 1^\top z }, \\
    \frac{\partial^2 \mathcal{F}_{P_r}}{\partial p^2}(p, \theta) &= \frac{\kappa G^\top \operatorname{diag}(z) G}{1^\top z} - \frac{\kappa G^\top z z^\top G}{(1^\top z)^2},
\end{align}
which shows the Hessian of \eqref{eq:log_sum_exp} is on the order of $O(\kappa)$. Therefore, a reasonable choice of $\kappa$ along with the log-sum-exp trick can avoid numerical instabilities. Theorems~\ref{thm:cd_alpha} and \ref{thm:high_order_cd_alpha} require the Hessian of either $\mathcal{F}_A$ or $\mathcal{F}_B$ to be positive definite, but \eqref{eq:log_sum_exp} is merely convex, whose Hessian is only positive semidefinite. Therefore, our framework does not deal with polytope-polytope collision avoidance.

\subsection{Details of the Calculation of $\frac{\partial^2 \alpha^\star}{\partial \theta^2}$}
\label{sec:cal_of_hessian}
From \eqref{eq:alpha_gradient}, we get $
    \frac{\partial^2 \alpha^\star}{\partial \theta^2}(\theta) = \frac{\partial^2 \mathcal{F}_A}{\partial \theta^2}(p^\star, \theta) + \frac{\partial p^\star}{\partial \theta}^\top  \frac{\partial^2 \mathcal{F}_A}{\partial p^2}(p^\star, \theta) \frac{\partial p^\star}{\partial \theta}  \allowdisplaybreaks
    + \frac{\partial^2 \mathcal{F}_A}{\partial \theta \partial p} (p^\star, \theta) \frac{\partial p^\star}{\partial \theta} + \frac{\partial p^\star}{\partial \theta}^\top \frac{\partial^2 \mathcal{F}_A}{\partial p \partial \theta} (p^\star, \theta) \allowdisplaybreaks 
    +  \sum_{i=1}^{n_p} \frac{\partial \mathcal{F}_A}{\partial p_i}(p^\star, \theta) \frac{\partial^2 p_i^\star}{\partial \theta^2}$
where $\frac{\partial p^\star}{\partial \theta}$ can be obtained from \eqref{eq:gradient_implicit}. To obtain $\frac{\partial^2 p_i^\star}{\partial \theta^2}$, we propose the following approach. As $N$ in \eqref{eq:gradient_implicit} is nonsingular and $\frac{\partial N^{-1}}{\partial \theta_i} = - N^{-1} \frac{\partial N}{\partial \theta_i} N^{-1}$, we have 
\begin{equation}\label{eq:hessian_implicit}
    \begin{bmatrix}
        \frac{\partial^2 p^\star}{\partial \theta_j \partial \theta} \\ \frac{\partial^2 \lambda^\star}{\partial \theta_j \partial \theta}
    \end{bmatrix} = - N^{-1} \frac{\partial N}{\partial \theta_j} N^{-1} \Omega + N^{-1} \frac{\partial \Omega}{\partial \theta_j}
\end{equation}
for $j = 1, ..., n_{\theta}$. Then, we only need to stack $\frac{\partial^2 p^\star}{\partial \theta_j \partial \theta}$ and extract $\frac{\partial^2 p_i^\star}{\partial \theta^2}$.

\subsection{Details on the Pick-and-Place Example}
\label{sec:pick_and_place_details}
Following Table~\ref{tab:scaling_functions}, we use a plane for the table surface (and the space under it) and four convex polygons (with $\kappa=80$ in \eqref{eq:log_sum_exp}) for the four vertical sides of the box (measuring 33.2~\si{cm} $\times$ 28.7~\si{cm} $\times$ 48.5~\si{cm}). The three links of interest of the robotic manipulator in Fig.~\ref{fig:exp1_bounding_boxes} are modeled by ellipsoids. For each ellipsoid, it needs to avoid the potential collision with five different obstacles, which leads to 15 CBFs. Each ellipsoid-obstacle pair satisfies the conditions of Theorems~\ref{thm:cd_alpha} and \ref{thm:high_order_cd_alpha}, and thus generates one CBF. More specifically, the ellipsoid provides the scaling function $\mathcal{F}_A$ in \eqref{eq:opt} and the obstacle provides $\mathcal{F}_B$. To solve \eqref{eq:opt}, we apply Theorem~\ref{thm:rimon_solution} for the ellipsoid-ellipsoid case. The ellipsoid-plane case reduces to a QP with one linear inequality constraint and has a closed-form solution. The ellipsoid-polytope case is an exponential cone problem, and we solve it with SCS \cite{odonoghue2016conic}.

\bibliographystyle{IEEEtran}
\bibliography{master}

@incollection{jittorntrum1984solution,
	address = {Berlin, Heidelberg},
	author = {Jittorntrum, Krisorn},
	booktitle = {Sensitivity, Stability and Parametric Analysis},
	pages = {127--138},
	publisher = {Springer Berlin Heidelberg},
	title = {Solution point differentiability without strict complementarity in nonlinear programming},
	year = {1984}}

@article{cohen2024safety,
	author = {Cohen, Max H and Molnar, Tamas G and Ames, Aaron D},
	journal = {Annual Reviews in Control},
	pages = {100947},
	publisher = {Elsevier},
	title = {Safety-critical control for autonomous systems: Control barrier functions via reduced-order models},
	volume = {57},
	year = {2024}}

@inproceedings{taylor2022safe,
	address = {(Cancun, Mexico)},
	author = {Taylor, Andrew J and Ong, Pio and Molnar, Tamas G and Ames, Aaron D},
	booktitle = {Proc. IEEE Conf. on Decision and Control},
	month = {Dec.},
	title = {Safe backstepping with control barrier functions},
	year = {2022}}

@inproceedings{usevitch2021adversarial,
	address = {(New Orleans, LA)},
	author = {Usevitch, James and Panagou, Dimitra},
	booktitle = {Proc. American Control Conf.},
	month = {May},
	title = {Adversarial resilience for sampled-data systems using control barrier function methods},
	year = {2021}}

@inproceedings{lee2023hierarchical,
	address = {(Detroit, MI)},
	author = {Lee, Jaemin and Kim, Jeeseop and Ames, Aaron D},
	booktitle = {Proc. IEEE/RSJ International Conf. on Intelligent Robots and Systems},
	month = {Oct.},
	title = {Hierarchical relaxation of safety-critical controllers: Mitigating contradictory safety conditions with application to quadruped robots},
	year = {2023}}

@inproceedings{nguyen2016exponential,
	address = {(Boston, MA)},
	author = {Nguyen, Quan and Sreenath, Koushil},
	booktitle = {Proc. American Control Conf.},
	month = {July},
	title = {Exponential control barrier functions for enforcing high relative-degree safety-critical constraints},
	year = {2016}}

@inproceedings{koditschek1987exact,
	address = {(Raleigh, NC)},
	author = {Koditschek, Daniel},
	booktitle = {Proc. International Conf. on Robotics and Automation},
	month = {March},
	title = {Exact robot navigation by means of potential functions: Some topological considerations},
	year = {1987}}

@article{robinson1980strongly,
	author = {Robinson, Stephen M},
	journal = {Mathematics of Operations Research},
	number = {1},
	pages = {43--62},
	publisher = {INFORMS},
	title = {Strongly regular generalized equations},
	volume = {5},
	year = {1980}}

@inproceedings{long2024safe,
	address = {(Toronto, Canada)},
	author = {Long, Kehan and Tran, Khoa and Leok, Melvin and Atanasov, Nikolay},
	booktitle = {Proc. American Control Conf.},
	month = {July},
	title = {Safe Stabilizing Control for Polygonal Robots in Dynamic Elliptical Environments},
	year = {2024}}

@inproceedings{manjunath2021safe,
	address = {(Austin, TX)},
	author = {Manjunath, Aniketh and Nguyen, Quan},
	booktitle = {Proc. IEEE Conf. on Decision and Control},
	month = {Dec.},
	title = {Safe and robust motion planning for dynamic robotics via control barrier functions},
	year = {2021}}

@inproceedings{ahmad2022adaptive,
	address = {(Cancun, Mexico)},
	author = {Ahmad, Ahmad and Belta, Calin and Tron, Roberto},
	booktitle = {Proc. IEEE Conf. on Decision and Control},
	month = {Dec.},
	title = {Adaptive sampling-based motion planning with control barrier functions},
	year = {2022}}

@article{tan2021high,
	author = {Tan, Xiao and Cortez, Wenceslao Shaw and Dimarogonas, Dimos V},
	journal = {IEEE Transactions on Automatic Control},
	number = {6},
	pages = {3021--3028},
	publisher = {IEEE},
	title = {High-order barrier functions: Robustness, safety, and performance-critical control},
	volume = {67},
	year = {2021}}

@article{gonccalves2024smooth,
	author = {Gon{\c{c}}alves, Vinicius Mariano and Tzes, Anthony and Khorrami, Farshad and Fraisse, Philippe},
	journal = {IEEE Transactions on Robotics},
	pages = {2950-2966},
	publisher = {IEEE},
	title = {Smooth Distances for Second Order Kinematic Robot Control},
	volume = {40},
	year = {2024}}

@inproceedings{wei2024confidence,
	address = {(Toronto, Canada)},
	author = {Wei, Shiqing and Krishnamurthy, Prashanth and Khorrami, Farshad},
	booktitle = {Proc. American Control Conf.},
	month = {July},
	title = {Confidence-Aware Safe and Stable Control of Control-Affine Systems},
	year = {2024}}

@article{odonoghue2016conic,
	author = {O'donoghue, Brendan and Chu, Eric and Parikh, Neal and Boyd, Stephen},
	journal = {Journal of Optimization Theory and Applications},
	pages = {1042--1068},
	publisher = {Springer},
	title = {Conic optimization via operator splitting and homogeneous self-dual embedding},
	volume = {169},
	year = {2016}}

@article{stellato2020osqp,
	author = {Stellato, Bartolomeo and Banjac, Goran and Goulart, Paul and Bemporad, Alberto and Boyd, Stephen},
	journal = {Mathematical Programming Computation},
	number = {4},
	pages = {637--672},
	publisher = {Springer},
	title = {OSQP: An operator splitting solver for quadratic programs},
	volume = {12},
	year = {2020}}

@article{gao2023non,
	author = {Gao, Yan and Bai, Chenggang and Fu, Rao and Quan, Quan},
	journal = {Robotics and Autonomous Systems},
	pages = {104291},
	publisher = {Elsevier},
	title = {A non-potential orthogonal vector field method for more efficient robot navigation and control},
	volume = {159},
	year = {2023}}

@incollection{chang2004gyroscopic,
	author = {Chang, Dong Eui and Marsden, Jerrold E},
	booktitle = {New trends in nonlinear dynamics and control and their applications},
	pages = {145--159},
	publisher = {Springer},
	title = {Gyroscopic forces and collision avoidance with convex obstacles},
	year = {2004}}

@article{mestres2022optimization,
	author = {Mestres, Pol and Cort{\'e}s, Jorge},
	journal = {IEEE Control Systems Letters},
	pages = {367--372},
	publisher = {IEEE},
	title = {Optimization-based safe stabilizing feedback with guaranteed region of attraction},
	volume = {7},
	year = {2022}}

@inproceedings{thirugnanam2022duality,
	address = {(Atlanta, GA)},
	author = {Thirugnanam, Akshay and Zeng, Jun and Sreenath, Koushil},
	booktitle = {Proc. American Control Conf.},
	month = {June},
	title = {Duality-based convex optimization for real-time obstacle avoidance between polytopes with control barrier functions},
	year = {2022}}

@article{thirugnanam2023nonsmooth,
	author = {Thirugnanam, Akshay and Zeng, Jun and Sreenath, Koushil},
	journal = {arXiv preprint arXiv:2306.13259},
	title = {Nonsmooth control barrier functions for obstacle avoidance between convex regions},
	year = {2023}}

@inproceedings{koutras2023enforcing,
	address = {(London, UK)},
	author = {Koutras, Leonidas and Vlachos, Konstantinos and Kanakis, George S and Dimeas, Fotios and Doulgeri, Zoe and Rovithakis, George A},
	booktitle = {Proc. IEEE International Conf. on Robotics and Automation},
	month = {May},
	title = {Enforcing Constraints for Dynamic Obstacle Avoidance by Compliant Robots},
	year = {2023}}

@article{murtaza2022safety,
	author = {Murtaza, Muhammad Ali and Aguilera, Sergio and Waqas, Muhammad and Hutchinson, Seth},
	journal = {IEEE Robotics and Automation Letters},
	number = {4},
	pages = {10659--10664},
	publisher = {IEEE},
	title = {Safety compliant control for robotic manipulator with task and input constraints},
	volume = {7},
	year = {2022}}

@article{xiao2023barriernet,
	author = {Xiao, Wei and Wang, Tsun-Hsuan and Hasani, Ramin and Chahine, Makram and Amini, Alexander and Li, Xiao and Rus, Daniela},
	journal = {IEEE Transactions on Robotics},
	number = {3},
	pages = {2289--2307},
	publisher = {IEEE},
	title = {Barriernet: Differentiable control barrier functions for learning of safe robot control},
	volume = {39},
	year = {2023}}

@article{wang2017safety,
	author = {Wang, Li and Ames, Aaron D and Egerstedt, Magnus},
	journal = {IEEE Transactions on Robotics},
	number = {3},
	pages = {661--674},
	publisher = {IEEE},
	title = {Safety barrier certificates for collisions-free multirobot systems},
	volume = {33},
	year = {2017}}

@article{ohnishi2019barrier,
	author = {Ohnishi, Motoya and Wang, Li and Notomista, Gennaro and Egerstedt, Magnus},
	journal = {IEEE Transactions on Robotics},
	number = {5},
	pages = {1186--1205},
	publisher = {IEEE},
	title = {Barrier-certified adaptive reinforcement learning with applications to brushbot navigation},
	volume = {35},
	year = {2019}}

@article{emam2021data,
	author = {Emam, Yousef and Glotfelter, Paul and Wilson, Sean and Notomista, Gennaro and Egerstedt, Magnus},
	journal = {IEEE Transactions on Robotics},
	number = {3},
	pages = {1671--1685},
	publisher = {IEEE},
	title = {Data-driven robust barrier functions for safe, long-term operation},
	volume = {38},
	year = {2021}}

@article{cavorsi2023multi,
	author = {Cavorsi, Matthew and Sabattini, Lorenzo and Gil, Stephanie},
	journal = {IEEE Transactions on Robotics},
	publisher = {IEEE},
	title = {Multi-robot adversarial resilience using control barrier functions},
	year = {2023}}

@article{brunke2022safe,
	author = {Brunke, Lukas and Greeff, Melissa and Hall, Adam W and Yuan, Zhaocong and Zhou, Siqi and Panerati, Jacopo and Schoellig, Angela P},
	journal = {Annual Review of Control, Robotics, and Autonomous Systems},
	number = {1},
	pages = {411--444},
	publisher = {Annual Reviews},
	title = {Safe learning in robotics: From learning-based control to safe reinforcement learning},
	volume = {5},
	year = {2022}}

@article{tee2009barrier,
	author = {Tee, Keng Peng and Ge, Shuzhi Sam and Tay, Eng Hock},
	journal = {Automatica},
	number = {4},
	pages = {918--927},
	publisher = {Elsevier},
	title = {Barrier {Lyapunov} functions for the control of output-constrained nonlinear systems},
	volume = {45},
	year = {2009}}

@article{kim1992real,
	author = {Kim, Jin-Oh and Khosla, Pradeep},
	journal = {IEEE Transactions on Robotics and Automation},
	number = {3},
	pages = {338-349},
	title = {Real-time obstacle avoidance using harmonic potential functions},
	volume = {8},
	year = {1992}}

@article{marcucci2023motion,
	author = {Marcucci, Tobia and Petersen, Mark and von Wrangel, David and Tedrake, Russ},
	journal = {Science robotics},
	number = {84},
	pages = {eadf7843},
	publisher = {American Association for the Advancement of Science},
	title = {Motion planning around obstacles with convex optimization},
	volume = {8},
	year = {2023}}

@article{liu2018convex,
	author = {Liu, Changliu and Lin, Chung-Yen and Tomizuka, Masayoshi},
	journal = {SIAM Journal on Control and optimization},
	number = {4},
	pages = {2712--2733},
	publisher = {SIAM},
	title = {The convex feasible set algorithm for real time optimization in motion planning},
	volume = {56},
	year = {2018}}

@article{shi2023safety,
	author = {Shi, Kaige and Hu, Guoqiang},
	journal = {IEEE Robotics and Automation Letters},
	number = {1},
	pages = {547--554},
	publisher = {IEEE},
	title = {Safety-Guaranteed and Task-Consistent Human-Robot Interaction Using High-Order Time-Varying Control Barrier Functions and Quadratic Programs},
	volume = {9},
	year = {2023}}

@inproceedings{sathya2018embedded,
	address = {(Limassol, Cyprus)},
	author = {Sathya, Ajay and Sopasakis, Pantelis and Van Parys, Ruben and Themelis, Andreas and Pipeleers, Goele and Patrinos, Panagiotis},
	booktitle = {Proc. European Control Conf.},
	month = {June},
	title = {Embedded nonlinear model predictive control for obstacle avoidance using PANOC},
	year = {2018}}

@article{zhang2020optimization,
	author = {Zhang, Xiaojing and Liniger, Alexander and Borrelli, Francesco},
	journal = {IEEE Transactions on Control Systems Technology},
	number = {3},
	pages = {972--983},
	publisher = {IEEE},
	title = {Optimization-based collision avoidance},
	volume = {29},
	year = {2020}}

@book{lavalle2006planning,
	author = {LaValle, Steven M},
	publisher = {Cambridge university press},
	title = {Planning algorithms},
	year = {2006}}

@article{kavraki1996probabilistic,
	author = {Kavraki, Lydia E and Svestka, Petr and Latombe, J-C and Overmars, Mark H},
	journal = {IEEE Transactions on Robotics and Automation},
	number = {4},
	pages = {566--580},
	publisher = {IEEE},
	title = {Probabilistic roadmaps for path planning in high-dimensional configuration spaces},
	volume = {12},
	year = {1996}}

@article{gonccalves2024control,
	author = {Gon{\c{c}}alves, Vinicius Mariano and Krishnamurthy, Prashanth and Tzes, Anthony and Khorrami, Farshad},
	journal = {IEEE Transactions on Control Systems Technology},
	publisher = {IEEE},
	title = {Control Barrier Functions With Circulation Inequalities},
	year = {2024}}

@article{reis2020control,
	author = {Reis, Matheus F and Aguiar, A Pedro and Tabuada, Paulo},
	journal = {IEEE Control Systems Letters},
	number = {2},
	pages = {731--736},
	title = {Control barrier function-based quadratic programs introduce undesirable asymptotically stable equilibria},
	volume = {5},
	year = {2020}}

@article{tan2024undesired,
	author = {Tan, Xiao and Dimarogonas, Dimos V},
	journal = {Automatica},
	pages = {111359},
	title = {On the undesired equilibria induced by control barrier function based quadratic programs},
	volume = {159},
	year = {2024}}

@article{xiao2021high,
	author = {Xiao, Wei and Belta, Calin},
	journal = {IEEE Transactions on Automatic Control},
	number = {7},
	pages = {3655--3662},
	publisher = {IEEE},
	title = {High-order control barrier functions},
	volume = {67},
	year = {2021}}

@article{rimon1997obstacle,
	author = {Rimon, Elon and Boyd, Stephen P},
	journal = {Journal of Intelligent and Robotic Systems},
	pages = {105--126},
	publisher = {Springer},
	title = {Obstacle collision detection using best ellipsoid fit},
	volume = {18},
	year = {1997}}

@article{wei2024diffocclusion,
	author = {Wei, Shiqing and Dai, Bolun and Khorrambakht, Rooholla and Krishnamurthy, Prashanth and Khorrami, Farshad},
	journal = {IEEE Robotics and Automation Letters},
	number = {4},
	pages = {3235-3242},
	title = {DiffOcclusion: Differentiable Optimization Based Control Barrier Functions for Occlusion-Free Visual Servoing},
	volume = {9},
	year = {2024}}

@book{siciliano2008robotics,
	author = {Siciliano, Bruno and Sciavicco, Lorenzo and Villani, Luigi and Oriolo, Giuseppe},
	publisher = {Springer London},
	title = {Robotics: Modelling, Planning and Control},
	year = {2008}}

@inproceedings{carpentier2019pinocchio,
	address = {(Paris, France)},
	author = {Carpentier, Justin and Saurel, Guilhem and Buondonno, Gabriele and Mirabel, Joseph and Lamiraux, Florent and Stasse, Olivier and Mansard, Nicolas},
	booktitle = {Proc. IEEE/SICE International Symposium on System Integration},
	month = {Jan.},
	pages = {614--619},
	title = {The {Pinocchio} {C++} library: A fast and flexible implementation of rigid body dynamics algorithms and their analytical derivatives},
	year = {2019}}

@inproceedings{tracy2023differentiable,
	address = {(London, UK)},
	author = {Tracy, Kevin and Howell, Taylor A and Manchester, Zachary},
	booktitle = {Proc. IEEE Int. Conf. on Robotics and Automation},
	month = {May},
	pages = {3663--3670},
	title = {Differentiable collision detection for a set of convex primitives},
	year = {2023}}

@article{dai2023safe,
	author = {Dai, Bolun and Khorrambakht, Rooholla and Krishnamurthy, Prashanth and Gon{\c{c}}alves, Vin{\'\i}cius and Tzes, Anthony and Khorrami, Farshad},
	journal = {IEEE Robotics and Automation Letters},
	number = {9},
	pages = {5376-5383},
	title = {Safe Navigation and Obstacle Avoidance Using Differentiable Optimization Based Control Barrier Functions},
	volume = {8},
	year = {2023}}

@article{hager1979lipschitz,
	author = {Hager, William W},
	journal = {SIAM J. on Control and Optimization},
	number = {3},
	pages = {321--338},
	publisher = {SIAM},
	title = {{Lipschitz} continuity for constrained processes},
	volume = {17},
	year = {1979}}

@book{boyd2004convex,
	author = {Boyd, Stephen P and Vandenberghe, Lieven},
	publisher = {Cambridge University Press},
	title = {Convex optimization},
	year = {2004}}

@article{ames2016control,
	author = {Ames, Aaron D and Xu, Xiangru and Grizzle, Jessy W and Tabuada, Paulo},
	journal = {IEEE Trans. on Automatic Control},
	number = {8},
	pages = {3861--3876},
	publisher = {IEEE},
	title = {Control barrier function based quadratic programs for safety critical systems},
	volume = {62},
	year = {2016}}

\begin{IEEEbiography}[{\includegraphics[width=1in,height=1.25in,clip,keepaspectratio]{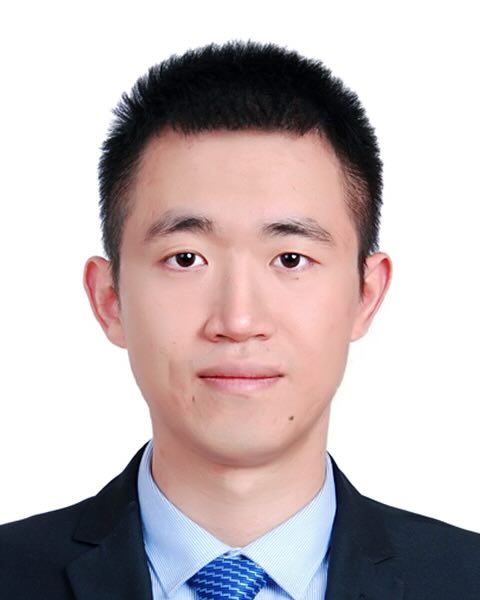}}]{Shiqing Wei} (Graduate Student Member, IEEE) received his BEng and MEng degrees in Mechanical Engineering from Shanghai Jiao Tong University, China, in 2018 and 2021, respectively. He also obtained a Dipl. Ing. with a minor in Applied Mathematics from École des Mines de Paris (now MINES Paris -- PSL), France, in 2019 through a double-degree program. Currently, he is pursuing a Ph.D. in Electrical Engineering at New York University Tandon School of Engineering in Brooklyn, NY, USA. His research interests include safe control synthesis, learning-based control methods, and verification of neural networks.
\end{IEEEbiography}

\begin{IEEEbiography}[{\includegraphics[width=1in,height=1.25in,clip,keepaspectratio]{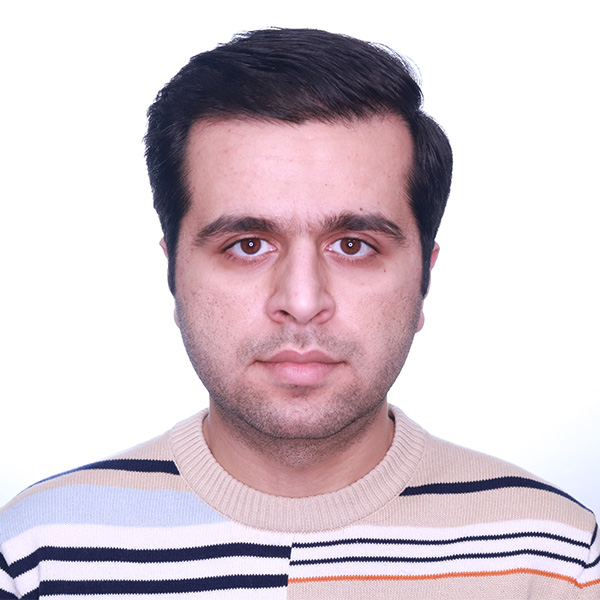}}]{Rooholla Khorrambakht} (Student Member, IEEE) received his Master's degree in Mechatronics Engineering from K.N.Toosi University of Technology, Iran, in 2020. Currently, he is pursuing a Ph.D. in Electrical Engineering at New York University Tandon School of Engineering in Brooklyn, NY, USA. His research interests include planning/control using learned world and reward models and safety filters for reliable deployment and online exploration. 
\end{IEEEbiography}

\begin{IEEEbiography}[{\includegraphics[width=1in,height=1.25in,clip,keepaspectratio]{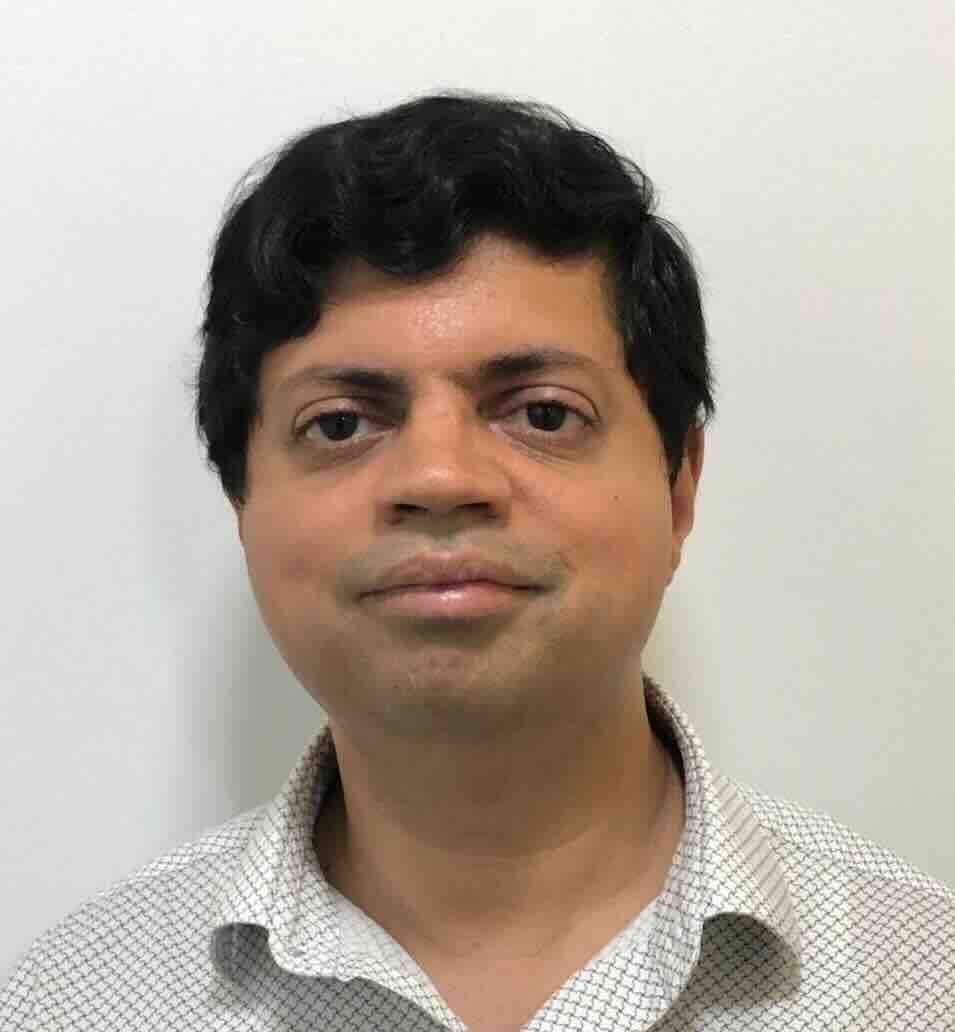}}]{Prashanth Krishnamurthy} (Member, IEEE) received B.Tech. degree in electrical engineering from Indian Institute of Technology, Chennai, India in 1999, and M.S. and Ph.D. degrees in electrical engineering from Polytechnic University (now NYU) in 2002 and 2006, respectively. He is a Research Scientist and Adjunct Faculty with Dept. of ECE at NYU Tandon School of Engineering. He has co-authored over 175 journal and conference papers. He has also co-authored the book ``Modeling and Adaptive Nonlinear Control of Electric Motors'' published by Springer Verlag in 2003. His research interests include autonomous vehicles and robotic systems, multi-agent systems, nonlinear control, resilient control,   machine learning, embedded systems, cyber-physical systems and cyber-security, and decentralized and large-scale systems.
\end{IEEEbiography}

\begin{IEEEbiography}[{\includegraphics[width=1in,height=1.25in,clip,keepaspectratio]{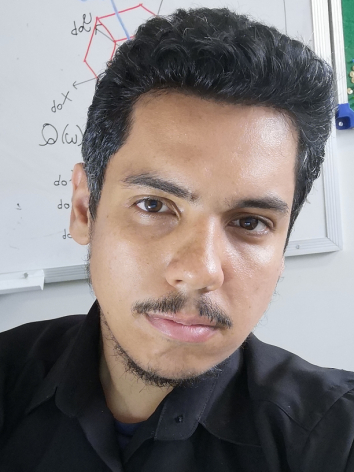}}]{Vinicius Mariano Gon\c{c}alves} received his B.S. degree in Control and Automation Engineering and his Ph.D. degree in Electrical Engineering from the Universidade Federal de Minas Gerais (UFMG), Belo Horizonte, Brazil, in 2011 and 2014, respectively. Since 2017, he has been a Tenured Assistant Professor in the Department of Electrical Engineering at UFMG. From 2022 to 2024, he served as a Research Associate at the Center for Artificial Intelligence and Robotics (CAIR) at New York University Abu Dhabi. His current research interests include control, robotics, optimization, and discrete-event systems.
\end{IEEEbiography}

\begin{IEEEbiography}[{\includegraphics[width=1in,height=1.25in,clip,keepaspectratio]{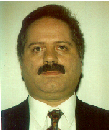}}]{Farshad Khorrami} (Fellow, IEEE) received the Bachelor's degrees in mathematics and electrical engineering in 1982 and 1984 respectively from The Ohio State University. He also received the Master’s degree in mathematics and Ph.D. in electrical engineering in 1984 and 1988 from The Ohio State University, Columbus, Ohio, USA. 
He is currently a professor in the Electrical and Computer Engineering Department at NYU, Brooklyn, NY where he joined as an assistant professor in Sept. 1988. His research interests include adaptive and nonlinear controls, robotics and automation, autonomous vehicles, cyber security for CPS, embedded systems security, machine learning, and large-scale systems and decentralized control. He has published over 370 refereed journal and conference papers in these areas. His book ``Modeling and Adaptive Nonlinear Control of Electric Motors'' was published by Springer  Verlag in 2003. 
He also has fifteen U.S. patents on novel smart micro-positioners, control systems, cyber security, and wireless sensors and actuators. 
He has developed and directed the Control/Robotics  Research  Laboratory at Polytechnic University (Now NYU) and is the Co-Director of the Center for AI and Robotics (CAIR) at NYU Abu Dhabi.
Dr. Khorrami has also commercialized UAVs as well as developed auto-pilots for various autonomous vehicles. His research has been supported by the DOE, ARO, NSF, ONR, DARPA, ARL,  AFRL and several corporations. He has served as conference organizing committee member of several international conferences. 
\end{IEEEbiography}

\end{document}